\documentclass[12pt]{article}
\pdfoutput=1

\usepackage{amsmath,amssymb}
\usepackage{amsthm}
\usepackage{natbib}
\usepackage{ulem}
\usepackage{colonequals}
\usepackage{stmaryrd}

\usepackage{xifthen}

\newcommand{\DeclareIOrWe}[1]{%
  \def\Iwe/{%
    \ifthenelse{\equal{#1}{I}}{
      I}{
      we}}
  \def\IWe/{%
    \ifthenelse{\equal{#1}{I}}{
      I}{
      We}}
    \def\myour/{%
    \ifthenelse{\equal{#1}{I}}{
      my}{
      our}}
  \def\MyOur/{%
    \ifthenelse{\equal{#1}{I}}{
      My}{
      Our}}
  }

\newcommand{\stderr}{\operatorname{s.e.}}
\renewcommand{\sqrt}[1]{\surd #1}
\newcommand{\EE}{\operatorname{E}}
\newcommand{\PP}{\operatorname{P}}
\newcommand{\E}[1]{\ifthenelse{\isin{(}{#1}}%
{
\ifthenelse{\isin{[}{#1}}
{\EE \left\{#1\right\}}
{\EE \left[#1\right]}
}
{\EE \left(#1\right)}}
\renewcommand{\Pr}[1]{\ifthenelse{\isin{(}{#1}}%
{\PP \left[#1\right]}%
{\PP \left(#1\right)}}
\newcommand{\var}{\operatorname{Var}}
\newcommand{\cov}{\operatorname{Cov}}
\newcommand{\Var}[1]{\ifthenelse{\isin{(}{#1}}%
{\var \left[#1\right]}%
{\var (#1)}}
\newcommand{\Cov}[2][]%
{\ifthenelse{\isempty{#1}}%
{\ifthenelse{\isin{(}{#2}}%
{\cov \left[#2\right]}%
{\cov (#2) }}
{\ifthenelse{\isin{(}{#1} \OR \isin{(}{#2}}%
{\cov \left[#1, #2\right]}%
{\cov ( #1, #2) }}
}
\usepackage{delimset}
\newcommand{\bkt}[2][1]{%
\ifthenelse{#1=1}{\brk{#2}}{%
\ifthenelse{#1=2}{\brk[s]{#2}}{%
\ifthenelse{#1=3}{\brk[c]{#2}}{\brk!{#2}%
} } }
}
\newcommand{\bbkt}[2][1]{%
\ifthenelse{#1=1}{\brk*{#2}}{%
\ifthenelse{#1=2}{\brk[s]*{#2}}{%
\ifthenelse{#1=3}{\brk[c]*{#2}}{\brk*{#2}%
} } }
}

\newcommand{\piP}[1]{\ensuremath{\pi_{\mathcal{P}#1}}}
\newcommand{\EEp}{\ensuremath{\EE_{\mathcal{P}}}}
\newcommand{\EEq}{\ensuremath{\EE_{\mathcal{Q}}}}
\newcommand{\PPp}{\ensuremath{\PP_{\mathcal{P}}}}
\newcommand{\PPq}{\ensuremath{\PP_{\mathcal{Q}}}}

\newcommand{\atob}[2]{\ensuremath{#1\mathbin{:}#2}}

\newcommand{\piQ}[1]{\ensuremath{\pi_{\mathcal{Q}#1}}}
\newcommand{\pibs}{\ensuremath{\pi_{\mathbf{s}}}}
\newcommand{\picse}{\operatorname{pic\_se}}
\newcommand{\card}[2][1]{\ensuremath{%
    \ifthenelse{#1=0}{%
      n_{#2}%
    }{%
      \mathrm{card}\bkt[#1]{#2}%
    }%
  }%
}
\newcommand{\pen}[1]{\ensuremath{\alpha(#1)}}
\newcommand{\fip}[2]
{\langle #1, #2\rangle_{F}}

\newcommand{\betatrue}[1][]{\ensuremath{\beta_{n #1}}}

\newcommand{\psiC}[1][]{\ensuremath{c_{0#1}}}
\newcommand{\psiCi}[1][]{\ensuremath{c_{1#1}}}

\newcommand{\Vns}{\ensuremath{V^{(n, \mathbf{s})}}}
\newcommand{\Vnstrunc}[1]{\ensuremath{\bar{V}^{(\mathbf{s})}_{n#1}}}

\usepackage{mathtools}
\DeclarePairedDelimiter{\indicator}{\llbracket}{\rrbracket}

\newcommand{\defeq}{\ensuremath{\colonequals}}
\newcommand{\eqdef}{\ensuremath{\equalscolon}}

\DeclareIOrWe{I}

\newcounter{A-enumi}
\newcounter{save-enumi} 

\newtheorem{prop}{Proposition}
\newtheorem*{corollary}{Corollary}
\newtheorem{lemma}{Lemma}

\theoremstyle{remark}

\providecommand{\keywords}[1]
{
  \small	
  \textbf{\textit{Key words and phrases:}} #1
}

\renewcommand{\marginpar}[1]{}

\author{Ben B. Hansen\thanks{This work has
    benefitted from comments of Jake Bowers, Joshua
    Errickson, Mark Fredrickson, Xuming He, Peter Schochet, Stilian Stoev and Lan
    Wang. Responsibility rests with the author for any shortcomings that remain.}}
\title{Matching calipers and the precision of index estimation}

\begin{document}
\maketitle
\begin{abstract}
  This paper characterizes the precision of index estimation as it carries over into precision of matching. In a model assuming Gaussian covariates and making best-case assumptions about matching quality, it sharply characterizes average and worst-case discrepancies between paired differences of true versus estimated index values. 
In this optimistic setting, worst-case true and estimated index differences decline to zero if $p=o[n/(\log n)]$, the same restriction on model size that is needed for consistency of common index models. This remains so as the Gaussian assumption is  relaxed to sub-gaussian, if in that case the characterization of paired index errors is less sharp.  The formula derived under Gaussian assumptions is used as the basis for a matching caliper.  Matching such that paired differences on the estimated index fall below this caliper brings the benefit that after matching, worst-case differences onan underlying index tend to 0 if $p = o\{[n/(\log n)]^{2/3}\}$. (With a linear index model, $p=o[n/(\log n)]$ suffices.) A proposed refinement of the caliper condition brings the same benefits without the sub-gaussian condition on covariates.  When strong ignorability holds and the index is a well-specified propensity or prognostic score, ensuring in this way that worst-case matched discrepancies on it tend to 0 with increasing $n$ also ensures the consistency of matched estimators of the treatment effect.
\end{abstract}

\section{Introduction}
In preparing a matched observational study, estimation of a treatment
propensity briefly takes center stage, as covariates are chosen
and a model specification is selected.  These models quickly recede
from view
once propensity score estimates have been extracted from them,
despite their carrying essential information about those estimates'
likely precision.  The situation is little different in matching on
prognostic or principal stratum scores: sampling variability of the
model standing behind a matching index is rarely
so much as even appraised.
We seem to take it for granted that errors of
estimation of a matching index can't possibly be so large as to
threaten the integrity of matching.

Propensity matching is understood to be a
large sample technique, as are logistic and other regression methods
typically used for index model estimation, and classical
asymptotics may seem to encourage
inattention to index estimation error. As treatment/covariate samples
$(z_i, \vec{x}_i)$ accumulate from any reasonable distribution of
fixed dimension $p+1$, one expects errors of estimation of the index,
$\{|\vec{x}_i\hat\beta - \vec{x}_i\betatrue|: i\}$, to be increasingly
negligible, decreasing with or near $n^{-1/2}$, just as
$|\hat\beta -\betatrue|_2 = O(n^{-1/2})$.
The problem is that the matching canon discourages
parsimony in propensity modeling  \citep{rubin:thom:1996}, and 
fixed-$p$ large sample theory may describe non-parsimonious models poorly.

Increasing-$p$ asymptotics for logistic regression and similar
techniques are available
\citep{portnoy1988exponential,he2000parameters}, if less widely known.
Given that $p = o\bkt[2]{(\log n)/n}$, they deliver
$|\hat\beta -\betatrue|_2 = O_P\bkt[2]{(p/n)^{1/2}}$, not
$O_P(n^{-1/2})$, as reviewed in \S~\ref{sec:context} below.
This suggests a still larger order, $p/n^{1/2}$, for
errors of form $\vec{X}(\hat\beta -\betatrue)$,
$\sum_{j=1}^p X_j^2=O(p)$ corresponding to
$|\vec{X}|_2 = \bbkt[1]{\sum_{j=1}^p X_j^2}^{1/2}= O(p^{1/2})$.  While somewhat of a simplification, the
suggestion is correct in its implication that if $p$ increases in
proportion with $n^{1/2}$, for example, then index estimation errors
need not diminish even as coefficient estimation errors do.  Outside
of fixed-$p$ asymptotics, consistency of the index model does not in
itself make index errors asymptotically negligible: that calls for
stronger assumptions, specialized matching techniques or a combination
of the two.

For control of index estimation error by way of stronger assumptions,
note that if $p$ is assumed to increase slowly enough, increasing
dimension regression asymptotics resemble those with fixed $p$.
It happens that $p \propto n^{1/2}$ is slightly too large for such
correspondence to obtain, so it is unsurprising that fixed-$p$
intuitions should fail in that regime.  (Asymptotic
normality of $\hat\beta$, for example, calls for
$p^2\log(p) = o(n)$, not $p = o(n^{1/2})$ [\citealp{he2000parameters}].)
But say the index model has \textit{sub}-$\sqrt{n}$ dimension, in the
specific sense that $p = o\bkt[2]{\bkt[1]{n/\log
    n}^{1/2}}$. 
Section~\ref{sec:dimin-index-errors} of this paper shows that for
consistently estimated index models with sub-gaussian covariates and
sub-$\sqrt{n}$ dimension, estimation errors of realized values of the
index tend to zero.  This convergence is in the strong, $l_\infty$,
sense of $\max_{i\leq n} |\vec{X}_i(\hat\beta - \betatrue)| = o_P(1)$,
so it justifies the common practices of matching, subclassifying or
simply trimming on the estimated propensity score, as analytic
interventions to secure overlap assumptions of the stronger type,
$c_l \leq \PP(Z=1|\vec{X}) \leq c_u$ with $[c_l, c_u] \subset (0,1)$,
as applied to the subset of available observations that remain after
pruning.

Section~\ref{sec:covh-moder-high} goes on to study ordinary
$\Cov{\hat\beta}$ estimates' adaptability to characterizing likely
sizes of index estimation errors, $\vec{x}_i(\hat\beta - \betatrue)$.
A fitted index model's Fisher information gives an estimate $\hat{C}$
of ${C}_{\hat\beta}=\Cov{\hat\beta}$, either directly or as part of an
Eicker-Huber-White sandwich. In sub-${n}^{1/2}$ dimensional regimes,
\marginpar{Is centering nec. here? To be
  revisited\ldots}
corresponding standard errors $\stderr\bkt[1]{\vec{x}\hat\beta}=\bkt[2]{\vec{x}\hat{C}\vec{x}'}^{1/2}$
will be seen to estimate closely\footnote{%
For $|\vec{v}|_2 \neq 0$, $|\vec{v}|_2^{-1}\bkt[3]{\stderr\bkt[2]{\vec{v}(\hat\beta
  - \betatrue)} - \var^{1/2}\bkt[2]{\vec{v}(\hat\beta
  - \betatrue)}}$ is
$o_P\bkt[2]{(p/n)^{1/2}}$, whereas
$|\vec{v}|_2^{-1}\stderr\bkt[2]{\vec{v}(\hat\beta  -
  \betatrue)}$ and $|\vec{v}|_2^{-1}\var^{1/2}\bkt[2]{\vec{v}(\hat\beta
  - \betatrue)}$ are both $O_P\bkt[2]{(p/n)^{1/2}}$.
} %
the sampling variabilities
$\var^{1/2}\bkt[2]{(\vec{x} -\bar{\vec{x}})(\hat\beta - \betatrue)}$.
As the sub-$\sqrt{n}$ condition is relaxed to sub-$n$
($p=o\bkt[2]{n/\log(n)}$), information- and sandwich-based estimators
underestimate $\Cov{\hat\beta}$. This limits their utility for
inference about $\betatrue$, and according their behavior outside of
sub-$\sqrt{n}$ regimes has received less study.  It does not follow,
however, that they are are ill-suited to inform the selection of
index-based matches.  We explore conditions under which analytic
$\Cov{\hat\beta}$-estimators continue to characterize sampling
variabilities of a linearization of $\hat\beta$, offering a basis for
estimators capturing the
better part of $\var\bkt[2]{\vec{x}(\hat\beta - \betatrue)}$.  In both
regimes, the largest values of
$\{\var^{1/2}\bkt[2]{\vec{x}_{i}(\hat\beta - \betatrue)}:i\}$ may be
well separated from the rest, in themselves appreciably increasing $\EE \bkt[2]{\max_{i\leq n}
  |\vec{X}_i(\hat\beta - \betatrue)|}$.
\marginpar{Revise/update me\ldots}
Results presented in
Section~\ref{sec:errors-paired-index} are helpful for identifying the worst offenders, subjects $i$ with large
$\stderr\bkt[2]{\vec{x}_{i}\hat\beta}$.  Either or both of the sub-gaussian and
sub-$\sqrt{n}$ conditions can be relaxed, but then control of index
estimation errors necessitates that such subjects be pruned from the sample. 

In such circumstances, matching offers alternate practical remedies
that can retain more of the sample. 
It helps first by shifting attention from particular
$\{\vec{x}_i\hat\beta: i\}$ or $\{\vec{x}_i\betatrue: i\}$ to paired
differences of indices,  $(\vec{x}_i - \vec{x}_j)\hat\beta$ or
$(\vec{x}_i - \vec{x}_j)\betatrue$, as does
\S~\ref{sec:errors-paired-index} below.  While underestimates
of $\var^{1/2}\bkt[2]{(\vec{x}_i - \vec{x}_j)\hat\beta}$ may be unhelpful for inference about
$(\vec{x}_i - \vec{x}_j)\betatrue$, but they can certainly inform
matching procedures.
Section~\ref{sec:mean-worst-case-Gaussian} uses $\hat{C}_{\hat{\beta}}$ to associate
deliberately reduced standard errors with 
{p}aired {c}ontrasts such as $(\vec{x}_i - \vec{x}_j)\hat\beta$, so
constructed that their average, the
\underline{p}aired \underline{i}ndex \underline{c}ontrast
{s}ummary \underline{s}tandard \underline{e}rror (PIC SE),
inexpensively approximates the root mean square of matched
discrepancies on estimated scores, $(\vec{x}_i-\vec{x}_j)\hat\beta$,
across pairs $\{i, j\}$ with little or no discrepancy on the true
score, $\vec{x}_i\beta \approx \vec{x}_j\beta$.


The PIC SE tends to zero at the same
rate as $|\hat\beta -\betatrue|_2$, making it useful
as a yardstick for matching.
With sub-gaussian data and
sub-$\sqrt{n}$ model dimension, using it to set the width
of a caliper on the index  ---
permitting $i$'s pairing with $j$ only if $|(\vec{x}_i - \vec{x}_j)\hat\beta| \leq c_{n}
\picse(\hat\beta)$, with $c_{n} = 2$, say  ---
forces matched differences on the true index to vanish in the
asymptotic limit:
\begin{equation}
  \label{eq:2}
\max_{1\leq i \sim j\leq n}|(\vec{x}_i -
\vec{x}_j)\betatrue| \stackrel{P}{\rightarrow} 0 \quad (\text{where
}``i \sim j" \text{ means } ``i \text{ is matched to } j"),  
\end{equation}
in virtue of  $\picse(\hat\beta) = O_P\bkt[2]{(p/n)^{1/2}}$. Indeed,
under the $p = o(n^{-1}\log n)$ growth condition needed for
consistency of common index models, \eqref{eq:2} holds with
non-constant $c_{n}$, provided that $c_{n} = O_{P}\bkt[2]{(\log n)^{1/2}}$.
Such matching requirements are generally less likely to exclude
subjects then comparable trimming rules, as they permit inclusion even
of subjects far from the center of the distribution whenever
the contrasting study arm has similarly situated subjects.
\marginpar{Cut next sentences? Revisit after drafting data
  sections\ldots}
If $i$ is excluded from the matched sample for lack of counterparts
$j$ within caliper distance, it must have been separated from its
comparison group by a distance exceeding the resolution of the
estimate of the index.  As will be seen in the data example, this
resolution can be strikingly large, much larger than extant caliper
width recommendations \marginpar{+ \citet{lunt2013PScalipers},
  \ldots?}
\citep{rosenbaum:rubi:1985a,rubin:thom:2000,austin2011optimal,wang2013PScalipers3groups};
in many cases it will be much more sparing in its exclusions from the
matched sample.

This is fortunate, because the paper will recommend a nonstandard
strengthening of the requirement that
$|(\vec{x}_i - \vec{x}_j)\hat\beta|$ be less than the designated
caliper width, excluding potential pairs $\{i,j\}$ either because
$|(\vec{x}_i - \vec{x}_j)\hat\beta|$ is too large or because $\stderr \bkt[2]{(\vec{x}_i - \vec{x}_j)\hat\beta}$ is. (Alternately put,
because $\vec{x}_{i}$ and $\vec{x}_{j}$ are too separated on either
the index itself or on a certain
index estimator-dependent Mahalanobis distance.)
The pairs $\mathcal{E} \subseteq \{\{i, j\}: 1\leq i \neq j \leq n\}$ that
remain eligible by this criterion satisfy
$\max_{\{i, j\} \in \mathcal{E}}|(\vec{x}_i - \vec{x}_j)(\hat\beta
-\betatrue)| = o_P(1)$, even as
$\max_{i\leq n} |\vec{X}_i(\hat\beta - \betatrue)|$ may no longer tend
to 0.  Thus~\eqref{eq:2} is maintained, even as the sub-gaussian
assumption on $\vec{X}$ is relaxed to a fourth moment condition, and
$p = o\bkt[3]{\bkt[2]{n/\log(n)}^{1/2}}$ is relaxed to
$p=o\bkt[3]{\bkt[2]{n/\log(n)}^{2/3}}$. In matching on propensity scores,
selecting pairs from within such an $\mathcal{E}$ ensures that the
overlap criterion can be assessed in terms of estimated scores,
because
$\max_{i \sim j}|(\vec{x}_i - \vec{x}_j)\hat\beta -(\vec{x}_i -
\vec{x}_j)\betatrue| = o_P(1)$.
As the recommended requirement can be
viewed as a varying (by value of $\stderr
\bkt[2]{(\vec{x}_i - \vec{x}_j)\hat\beta}$) limit on $|(\vec{x}_i -
\vec{x}_j)\hat\beta|$, \Iwe/ continue to call it a caliper.   In many cases, these
PICSE calipers
continue to be more inclusive than would Rosenbaum and Rubin's
\citeyearpar{rosenbaum:rubi:1985a} canonical
$|(\vec{x}_i - \vec{x}_j)\hat\beta| \leq
0.2\operatorname{s.d.}(\mathbf{x}\hat\beta)$ requirement, and in all
cases PICSE calipers cause exclusion of a unit $i$ only when our
best estimate of the index distances from it to each potential
counterpart $j$ exceeds the resolution of the index's estimation.  


\section{Context}\label{sec:context}

\subsection{Estimable index scores} \label{sec:estim-index-scor}

In an observational study with a treatment and a control condition, the
propensity score is a function
$\vec{x} \mapsto g^{-1}\bkt[2]{\PP(Z=1\mid \vec{X}=\vec{x})}$, where
$g: \Re \rightarrow [0,1]$ is continuous and increasing
\citep{rosenbaum:rubi:1983}. The $z$-on-$\vec{x}$ regression is often
assumed to follow a generalized linear model such as the logistic,
$\PP(Z=1\mid \vec{X}=\vec{x}) = \bkt[2]{1 + e^{-\vec{x}\beta}}^{-1}$.
Taking $g$ as that model's link function identifies the propensity
score with the index $\vec{x}\beta$.  \citet{rosenbaum:rubi:1985a}
recommended matching on $\vec{x}\hat\beta$, not
$g^{-1}(\vec{x}\hat\beta)$. Similarly \marginpar{Cites}prognostic scoring, confounder and risk
scoring, and principal stratum scoring fit (parametric) regression models in order to extract indices
$\vec{x}\hat\beta$ for use in matching, subclassification
or sample trimming, if not also in weighting- or covariance
adjustments to be applied once an analysis sample has been selected.

Let $R$ denote the dependent variable of the index model, e.g. $Z$ for
a propensity score or a response $Y$ for a prognostic
score.  Let $\hat\beta$ be the solution of
\begin{equation}\label{eq:22}
  \bbkt[2]{\sum_{i=1}^{n}\psi(r_{i}, \vec{x}_{i};
  {\beta})} + \alpha(\beta) = 0
\end{equation}
 in $\beta$, where $\psi (r, \vec{x}; \beta)$ is the $\Re^{p}$-valued gradient of a scalar-valued
function $\rho(r, \vec{x}; \beta)$ that is convex in $\beta$.  (In maximum likelihood estimation $\alpha(\cdot)=0$,
but
Bayesian estimation [\citealp[e.g.,][]{gelman2008weakly}] and
certain frequentist bias-reduction schemes
[\citealp{firth:1993,kosmidisFirth2009bias}] minimize a penalized
objective, in which cases $\alpha(\cdot)$ is the penalty term's gradient.) 
The index score (IS) is \textit{estimable} if 
\marginpar{Update to allow for $y$?}
\begin{equation}\label{eq:23}
  \E{ \bkt[2]{\sum_{i=1}^{n}\psi(R_{i}, \vec{x}_{i}; \beta)} +
  \alpha(\beta) } =0
\end{equation}
has a unique root $\betatrue$, with
$\sup_{\gamma: |\gamma - \betatrue|_{2}\leq 1}|\alpha(\gamma)|_{2} =
o_{P}(n^{1/2})$.  Following \citet{he2000parameters},
near-roots of \eqref{eq:22} and/or
\eqref{eq:23} are acceptable when the equations do not have exact solutions, provided
that there is a unique nearest root, but these exact or nearest
roots are assumed to satisfy 
$\sum_{j=1}^{p}\bbkt[2]{\sum_{i=1}^{n}\psi_{j}(r_{i}, \vec{x}_{i};
  \hat{\beta})}^{2} = o_{P}(n^{1/2})$  and $\sum_{j=1}^{p}\bbkt[2]{\E{\sum_{i=1}^{n}\psi_{j}(R_{i}, \vec{x}_{i};
  \betatrue)}}^{2} = o(n^{1/2})$.  Assume there are functions
$\psiC(\cdot, \cdot, \cdot)$ and $w(\cdot)$ such that
\marginpar{Why not just $\psiC(r,
  \beta_{0}+\vec{x}\beta)$, w/o $\vec{x}$ arg?\ldots}%
\begin{equation}
\psi(r, \vec{x},
\beta) = \psiC(r, \vec{x}, \beta_{0}+\vec{x}\beta) w(\vec{x})
          (1,\vec{x})'.
  \label{eq:20}
\end{equation}
This structure accommodates robust
\citep{cantoniRonchetti2001robustglms} and generalized
\citep{liangZeger86} estimating equations as well as score
functions such as 
logistic regression's,
$\psi(r, \vec{x}, \beta) = \bkt[3]{r -
  \bkt[2]{1+\exp(-\beta_{0}-\vec{x}\beta)}^{-1}}(1,\vec{x})'$. If
there are pre-existing strata 1, \ldots, $L$, with matches to be made
within strata and any subclasses to further divide them, then
\eqref{eq:20} may be modified by replacement of $\beta_{0}$ with
stratum-specific intercepts $\beta_{1}, \ldots, \beta_{L}$.
Sufficient conditions for consistency of an IS will be reviewed in
\S~\ref{sec:cons-estim-iss}.

\subsection{Sub-gaussian random variables}
\label{sec:sub-gaussian-random}

A real-valued random variable $V$ is sub-gaussian if its tails are no
heavier than that of a centered Normal distribution: 
there is a finite constant $s_{V}$ such that for $t>0$,
\begin{equation*}
    \PP(V < t) \leq \exp\bkt[2]{-t^2/(2s_{V}^2)}, \quad  \PP(V > t) \leq \exp\bkt[2]{-t^2/(2s_{V}^2)}.
\end{equation*}
When this holds, $s_{V}$ can be taken to be a constant multiple of
$\|V\|_{\psi_{2}}$, the sub-gaussian norm of $V$.  This norm is
defined as the infimum of
$\{t>0: \EE\bkt[2]{\exp\bkt[1]{V^{2}/t^{2}}} \leq 2\}$, a nonempty set for
sub-gaussian $V$, or as $\infty$ if $V$ is not sub-gaussian.  For
vector $\vec{V}$,
$\|\vec{V}\|_{\psi_{2}} = \sup\{\|\vec{V}\gamma\|_{\psi_{2}} :
|\gamma|_{2}=1\}$, and $\vec{V}$ is said to be sub-gaussian if
$\|\vec{V}\|_{\psi_{2}} < \infty$
\citep[\S~2.5, 3.4]{vershynin18HDPbook}.  It follows directly from these definitions that for fixed
vectors $\gamma$, $\|V\gamma\|_{\psi_{2}}=|\gamma|_{2}\|V\|_{\psi_{2}}$,
$\|\vec{V}\gamma\|_{\psi_{2}} \leq
|\gamma|_{2}\|\vec{V}\|_{\psi_{2}}$ and $\|\gamma\|_{\psi_{2}}=(\log 2)^{-1/2}|\gamma|_{2}$; and that
for fixed matrices $M$,
$\|\vec{V}M\|_{\psi_{2}} \leq
|M|_{2}\|\vec{V}\|_{\psi_{2}}$. 

Sums of sub-gaussian variables are sub-gaussian.  Hoeffding's
inequality bounds tails of sums of independent sub-gaussians in terms
of the sum of the summands' squared sub-gaussian norms.  Maxima of
sequences of sub-gaussian random variables grow slowly: for an
absolute constant $k_{0}$,
\begin{equation}\label{eq:3}
  \EE \max_{1\leq i \leq n}V_{i} \leq
  k_{0}\bbkt[1]{\max_{1 \leq i \leq n}\|V_{i}\|_{\psi_{2}}}\bkt[1]{\log n}^{1/2},
\end{equation}
with $\{V_{i}: i\}$ independent or dependent. For
$\{G_{i}: i\}$ Normal with variance 1 or less, \eqref{eq:3} holds with
$\sqrt{2}$ in place of $k_{0}\max_{i\leq n}\|V_{i}\|_{\psi_{2}}$; if $\{G_{i}:i\}$ are
independent $\mathrm{N}(0,1)$ then this bound is sharp, in the sense of
$\lim_{n\uparrow\infty}(\log n)^{-1/2}\EE \max_{1\leq i \leq n} G_{i}
= \sqrt{2}$ 
\citep[\S~2.5]{boucheronLugosiMassart13}. Maxima of
sub-gaussian vector sequences grow slowly as well: there are absolute
constants $k_{1}$ and $k_{2}$ such that for any sub-gaussian $\{\vec{V}_{i}\}$
\marginpar{Abbreviate for version to submit \ldots}
%
with mean 0 and covariance $I$, and any deterministic matrices
$\{M_{i}\}$ with column dimension matching the extent of 
$\vec{V}$, 
\begin{equation}
  \EE \bbkt[1]{\max_{1\leq i \leq n} |M_{i}\vec{V}_{i}'|_{2}} 
  \leq \max_{ i \leq n}\|\vec{V}_{i}\|_{\psi_{2}}
    \bbkt[3]{
    k_{1}\max_{ i \leq n} |M_{i}|_{F}
  + k_{2}\bbkt[1]{\max_{ i \leq n}
    |M_{i}|_{2}}\bkt[1]{\log n}^{1/2}}\label{eq:30}
\end{equation}
where $|M|_{F}=\operatorname{tr}(M'M)^{1/2}$
\citep[Ex. 6.3.5]{vershynin18HDPbook}. 
A consequence is that
if for each $n$ $\{\vec{X}_{ni}:1\leq i \leq n\}$ are
independent random vectors of length $p_{n}$ such that
$|\Cov{\vec{X}_{ni}}|_{2}$ and $\|\vec{X}_{ni}\|_{\psi_{2}}$ are
uniformly bounded, then $\max_{1\leq i \leq
  n}|\vec{X}_{i}|_{2}=O_{P}\bkt[2]{\max(p_{n}, \log n)^{1/2}}$.  This property of sub-gaussian covariates, \ref{A-boundedXes} in Section~\ref{sec:cons-estim-iss}
below, will be assumed in Section~\ref{sec:asympt-exact-post} but
then relaxed in Section~\ref{sec:errors-paired-index}. The scalar $\psiC(R, \vec{x}, \vec{x}\beta)$ in
\eqref{eq:20} above, on the other
hand, will consistently be required to be sub-gaussian,
via Section~\ref{sec:cons-estim-iss}'s \ref{A-c0moments}.

\subsection{Consistently estimable index scores} \label{sec:cons-estim-iss}
Let the data and model parameter be arranged in triangular arrays,
with sample and model $n$ having $n$ observations of $p$ 
independent variables ${x}_{i}$, $p$ (strictly, $p_{n}$) increasing with $n$.  Consistency of
$\hat{\beta}$ for $(\betatrue: n)$ will mean that $p = o(n)$ and
$|\hat\beta - \betatrue|_{2}^{2} = O_{P}(p/n)$.  Conditions for such
consistency are available in the literature.  We order\marginpar{?`We
  order?}
to present them along
with accompanying conditions characterizing the $\hat{\beta}$'s
relationship to its closest linear approximation.

For estimable $\betatrue$, define $A_{n}=A_{n}(\betatrue)$ and
$\hat{A}_{n}=A_{n}(\hat\beta)$, where
\marginpar{Check me, near convergence\ldots}
\begin{equation}
    \label{eq:7}
A_{n}(\gamma ) = \frac{1}{n}\Bigg\{  \sum_{i=1}^{n}\nabla_{\beta} \left. \E{
                  \psi(R_{i}, \vec{x}_{i};
                  \beta) } \right|_{\beta =  \gamma} +
                                                    \pen{\gamma}\Bigg\},
\end{equation}
with ``$\nabla_\beta$'' interpreted in
terms of weak differentiation if ordinary partial derivatives do not
exist for some $\beta$ values. For invertible
$A_{n}$, the linearization of estimator $\hat{\beta}$ is given by 
\begin{equation}
  \label{eq:4}
  \tilde{\beta}_{n} = A_{n}^{-1} \frac{1}{n}\bkt[3]{\bkt[2]{\sum_{i=1}^{n}\psi(r_{i},
  \vec{x}_{i}; \betatrue)} + \pen{\betatrue}}. 
\end{equation}
The random vector $n^{-1} \sum_{i=1}^{n} \psi(R_{i}, \vec{x}_{i};
\betatrue )$ 
has covariance  $n^{-1}B_{n}$ where
\begin{equation}
  \label{eq:13}
 B_{n}  =  B_{n}(\betatrue) \text{ and } B_{n}(\gamma) = n^{-1}\EE                                                   
                                               \sum_{i=1}^{n}
                                               \psi(R_{i},
                                               \vec{x}_{i}; \gamma )
                                               \psi(R_{i},
                                               \vec{x}_{i};  \gamma )'.
\end{equation}

\begin{prop}[\citealp{he2000parameters}] \label{lem:betahat-consist}
  Under~\ref{A-estimable}, \ref{A-psismooth}, \ref{A-l2Sfinite} and
  \ref{A-consistencyrates} as stated below:
  \begin{enumerate}
  \item $|\hat\beta -
    \betatrue|_{2} \stackrel{P}{\rightarrow} 0$, with rate
    $O_{P}\bkt[2]{\bkt[1]{p/n}^{1/2}}$;
  \item \label{lem:betahat-consist-btilde}
  $|\hat\beta - \tilde{\beta}_{n}|_{2} \stackrel{P}{\rightarrow} 0$, with
    rate $O_{P}\bkt[2]{(p/n)\bkt[1]{\log n }^{1/2}}$. 
  \end{enumerate}
\end{prop}

Proposition~\ref{lem:betahat-consist} restates He and Shao's
Theorems~2.1 and 2.2 as applied to models in which only the
coefficient parameter $\beta_{n}$ grows in dimension with $n$, with a
slight strengthening of rate condition. (They assume
$p\log(p) = o(n)$, while \ref{A-consistencyrates} says $p\log(n) = o(n)$.)  Their
Theorem~2.2 characterizes decline of the linearization error only
for sub-$\sqrt{n}$ dimensional models, but a straightforward adaptation
of its proof gives part~\ref{lem:betahat-consist-btilde} of
the proposition.  See also He and Shao's Example 3.

Our regularity assumptions are as follows.  
\renewcommand{\theenumi}{A\arabic{enumi}}
\setcounter{save-enumi}{\value{enumi}}
\setcounter{enumi}{0}

\begin{enumerate}
\item \label{A-centering} The columns of $\mathbf{x}$ are
  centered. There may be pre-existing
  strata 1, \ldots, $L$, with $L/n \rightarrow 0$.  In this case the columns of
  $\mathbf{x}$ are also stratum-centered: for stratifying variable $v$, $\sum_{i:
    v_{i}=\ell}\vec{x}_{i}=0$, $\ell=1, \ldots, L$. 
\item\label{A-estimable} The IS is estimable
  (as defined in \S~\ref{sec:estim-index-scor}) and linear in $\mathbf{x}$.
\item\label{A-invertibleA}  $A_{n}$ is invertible. Furthermore there is $\delta>0$
  such that $A_{n}(\gamma)$ is invertible whenever $|\gamma
  -\betatrue|_{2} < \delta$, and $\sup_{\gamma: |\gamma
  -\betatrue|_{2} <\delta} |A_{n}(\gamma)^{-1}|_{2}$ is bounded. 
\item \label{A-psismooth} $\psi(r, \mathbf{x}, \beta)$ is of form
  \eqref{eq:20}. The functions $\psiC(r,\vec{x}, \cdot)$
  (i.e. $\eta \mapsto \psiC(r,\vec{x}, \eta)$) are Lipschitz
  continuous with a common Lipschitz constant, as are $(\partial/\partial \eta) \EE \psiC(R,\vec{x}, \eta)$.
\item \label{A-c0moments} The random variables
  $\psiC(R_{i},\vec{x}_{i}, \vec{x}_{i}\betatrue)$
  have bounded sub-gaussian norm. 
\item \label{A-l2Sfinite} For $\ell=0$,
  2 and 4,
  \begin{equation*}
    \max_{\substack{\gamma, \delta: |\gamma|_{2}=\\ |\delta|_{2} =1}}
\sum_{i=1}^{n}w(\vec{x}_{i})^{\ell} (\vec{x}_{i}\gamma)^{2}
(\vec{x}_{i}\delta)^{2} = O(n).
  \end{equation*}
\item \label{A-regPS} $p^{-1}|\betatrue|_{2}^{2} = p^{-1} \sum_k
  \betatrue[k]^2$ is bounded.
\item\label{A-PSvar}  $s^{2}(\mathbf{x}\betatrue) =
  \betatrue'S^{(x)}\betatrue $ tends to a limit in $(0,
  \infty]$.
\item \label{A-consistencyrates} $p = o\bkt[2]{n/(\log n)}$, i.e. $(p
  \log n)/n \rightarrow 0$ (sub-$n$
  model dimension).
\setcounter{A-enumi}{\value{enumi}}
\end{enumerate}
\setcounter{enumi}{\value{save-enumi}}
\renewcommand{\theenumi}{\roman{enumi}}

While the link function $g$ is assumed the same for each $n$, the
coefficient vector $\beta$ grows in length, and $\betatrue$ needn't
converge.  Indeed, according to \ref{A-PSvar} $s(\mathbf{s}\betatrue)$
is permitted to diverge (but not tend to 0).

It is less burdensome here to assume invertibility of $A_{n}$, as
\ref{A-invertibleA} does, than in other regression contexts, as for present applications one can freely change the basis of the design matrix, there being no interest in particular elements or contrasts of $\beta$.   
Via the Cauchy-Schwartz inequality, \ref{A-l2Sfinite} entails that
$\max_{\gamma:
  |\gamma|_{2}=1}n^{-1}\sum_{i=1}^{n}w(\vec{x}_{i})^{\ell}
(\vec{x}_{i}\gamma)^{2} = O(1)$ for $\ell \in \{0, 1, 2\}$,
in turn giving $|S^{(x)}|_{2}=O(1)$.  The appropriateness of these
commitments can be evaluated in advance of IS estimation, whereas
\ref{A-invertibleA} calls for inspection of model-fitting artifacts
after estimation of $\beta$.  A simple measure to improve the fit of \ref{A-invertibleA}, as well as
\ref{A-fullrankB} below, is to
trim explanatory variables that contribute relatively little to index
model fit, as indicated by common model selection criteria.

Certain results take stronger forms with one or more
secondary conditions. 
\setcounter{save-enumi}{\value{enumi}}
\renewcommand{\theenumi}{A\arabic{enumi}}
\begin{enumerate}
\setcounter{enumi}{\value{A-enumi}}
\item \label{A-paramrates} $p^2 = o\bkt[2]{n/(\log n)}$
  (sub-$\sqrt{n}$ model dimension).
\item \label{A-boundedXes} $\max_{i\leq n}|\vec{x}_{i}|_{2}^{2} =
  O\bkt[2]{\max(p, \log n)}$, and
  $\max_{i \leq n} w(\vec{x}_{i})^{2} = O(\log n)$ (sub-gaussian covariates).
\item \label{A-fullrankB} Each $S^{(x)} =
  (n-L)^{-1}\mathbf{x}'\mathbf{x}$ and $B_{n}$ are of full rank, with 
  $|(S^{(x)})^{-1}|_{2}$ and $|B_{n}^{-1}|_{2}$ bounded (full-rank covariance).
  \setcounter{save-enumi}{\value{enumi}}
\end{enumerate}
\renewcommand{\theenumi}{\roman{enumi}}
If the $\vec{x}$es (and $w(\vec{x})$) are sub-gaussian in the sense of
being realizations of $\vec{X}$ with $\EE|\vec{X}|_{\psi_{2}}$
uniformly bounded, while also \ref{A-l2Sfinite} holds in the sense of
$|\Cov{\vec{X}}|_{2}$ (and thus $p^{-1}\EE|\vec{X}|_{2}^{2}$) being
uniformly bounded as well, then \ref{A-boundedXes} follows from
\eqref{eq:30}, as discussed in \S~\ref{sec:sub-gaussian-random}
following \eqref{eq:30}.
\marginpar{Abbreviate for version to submit}
According to Proposition~\ref{lem:ChatC} below, sandwich estimators of
$\Cov{\hat\beta}$ generally require sub-gaussian covariates, and
\ref{A-paramrates}, sub-$\sqrt{n}$ model dimension; but the covariance
estimator based on $\hat{A}_{n}$ but not $\hat{B}_{n}$ generally
requires only sub-$n$ model dimension, \ref{A-consistencyrates}, and
for present purposes will be similarly beneficial even when it lacks
Fisher consistency as compared to the sandwich estimator. According to
Proposition~\ref{prop:stderr-consist}, full-rank covariance (\ref{A-fullrankB})
makes index sampling variabilities
$\Var{\vec{x}\hat\beta}$ and $\Var{(\vec{x}_{i} -
  \vec{x}_{j})\hat\beta}$ estimable without attention to size of
$\vec{x}$ or $\vec{x}_{i} - \vec{x}_{j}$.  However, our method for
asymptotically exact matching  does not require this, and is valid
with $S$ or $B_{n}$ of less than full rank.

\section{Asymptotically exact matching with Gaussian or sub-gaussian $\vec{X}$}
\label{sec:asympt-exact-post}

For each $n$ let $\mathcal{S}_{n}$ be a random partition of
$\{1, \ldots, n\}$.  It is not presumed that $\mathcal{S}_{n}$
expands or extends any earlier
partition $\mathcal{S}_{1}, \ldots, \mathcal{S}_{n-1}$. Denote by
$[i]_{\mathcal{S}_{n}}$ the unique $\mathcal{S}_{n}$
element containing $i \leq n$, and write
$i \stackrel{\mathcal{S}_{n}}{\sim} j$ when there is
$\mathbf{s} \in \mathcal{S}_{n}$ such that $i, j \in \mathbf{s}$.\label{defsec:isimj}
Absent ambiguity as to which 
$n$ or partition sequence is intended, these symbols are given as ``$i
{\sim} j$'' or ``$[i]$,'' respectively.
The
progression $\{\mathcal{S}_{n}\}$ constitutes an 
\textit{asymptotically exact}
index post-stratification if $\mathcal{S}_{n}$-stratum
width in the direction of the underlying index, 
$\max \{|(\vec{x}_{i}-\vec{x}_{j})\betatrue|: i
\stackrel{\mathcal{S}_{n}}{\sim} j \}$, tends in probability to 0.
This section presents a tolerance for $\mathcal{S}_{n}$-stratum
width in the direction of the estimated index, 
$\max \{|(\vec{x}_{i}-\vec{x}_{j})\hat\beta|: i
\stackrel{\mathcal{S}_{n}}{\sim} j \}$, that is narrow enough to ensure asymptotic exactness
in the special case of a sub-gaussian covariate (\ref{A-boundedXes}).
It is also is sufficiently wide that, in
a further special case to be described in
Section~\ref{sec:mean-worst-case-Gaussian},  no $i$ meriting
placement in a poststratum with representation of the contrasting group
can be excluded from such placement in virtue of
$\vec{x}_{i}\hat\beta$ being isolated relative to
$\{\vec{x}_{j}\hat\beta: z_{j}\neq z_{i}\}$.

When $\betatrue$ is subject to estimation, the observable counterparts
of differences $(\vec{x}_{i}-\vec{x}_{j})\betatrue$, $1\leq i, j
\leq n$,  that is contrasts of form $(\vec{x}_{i}-\vec{x}_{j})\hat{\beta}$,
are termed \textit{paired index contrasts} ({PIC}s). The discrepancy
between a PIC and the paired contrast it estimates,
$(\vec{x}_{i}-\vec{x}_{j})(\hat{\beta} - \betatrue)$, is a \textit{PIC error}.
Ensuring that a post-stratification is asymptotically exact calls for
separate attention to PIC errors versus the PICs themselves.  Our
width tolerance will involve a novel estimate of PIC
error size, the \textit{PIC SE}.


\subsection{PIC errors in the sub-gaussian case} \label{sec:dimin-index-errors}

Recall that errors $\hat\beta - \betatrue$ of index coefficient
estimates decompose as $(\hat\beta - \tilde\beta) + (\tilde\beta - \betatrue)$,
with $\tilde\beta$ the linearization defined in \eqref{eq:4}.
Index and PIC errors decompose similarly.  

\begin{corollary}[of Proposition~\ref{lem:betahat-consist} part  \ref{lem:betahat-consist-btilde}] 
  If \ref{A-estimable}, \ref{A-psismooth}, \ref{A-l2Sfinite}, 
  \ref{A-consistencyrates}, and
  \ref{A-boundedXes}, then
  $\max_{i} |\vec{x}_{i}(\hat{\beta} - \tilde{\beta}_{n})|
  = O_{P}\bkt[2]{p^{3/2}(\log n)^{1/2}/n}$. (Unless $p =
  o(\log n)$, in which case $\max_{i} |\vec{x}_{i}(\hat{\beta} -
  \tilde{\beta}_{n})| = O_{P}\bkt[2]{p(\log n)/n}$).
\end{corollary}
\begin{prop} \label{prop:btilde-subgnorm}
  If \ref{A-centering}, \ref{A-estimable}, \ref{A-c0moments},
  and \ref{A-l2Sfinite}, then $\|\tilde{\beta}_{n} -
  \betatrue\|_{\psi_{2}} = O\bkt[1]{n^{-1/2}}$.
\end{prop}
\begin{corollary}[of Proposition~\ref{prop:btilde-subgnorm}] 
  If \ref{A-centering}, \ref{A-estimable}, \ref{A-c0moments},
  \ref{A-l2Sfinite}, and \ref{A-boundedXes}, 
  then
    $\EE  {\max_{i\leq n}|\vec{x}_{i}(\tilde{\beta}_{n} -
    \betatrue)|} =  O\bbkt[3]{\bkt[2]{(p \log n)/n}^{1/2}}$.  (Unless
  $p = o(\log n)$, in which case $\EE  {\max_{i\leq n}|\vec{x}_{i}(\tilde{\beta}_{n} -
    \betatrue)|} =  O\bbkt[3]{\bkt[2]{(\log n)/n^{1/2}}}$.)
\end{corollary}

Proposition~\ref{lem:betahat-consist}'s corollary is immediate
from its part~\ref{lem:betahat-consist-btilde} in combination
with~\ref{A-boundedXes}.  Prop.~\ref{prop:btilde-subgnorm}
is proved in the appendix; its corollary flows from \ref{A-boundedXes} and
\eqref{eq:3}.  They follow \citet{damourDingFellerLeiSekhon21} in assuming sub-gaussian covariates,
in the sense of \ref{A-boundedXes}, an assumption to be relaxed in 
Section \ref{sec:errors-paired-index} below.  

The corollaries characterize errors of estimation of index values
rather than paired contrasts of them, but they have immediate
extensions giving the same rates of decline for 
$\max_{i,j\leq n} |(\vec{x}_{i}-\vec{x}_{j})(\hat{\beta} -
\tilde{\beta}_{n})|$
and $\EE  {\max_{i,j\leq n}|(\vec{x}_{i}-\vec{x}_{j})(\tilde{\beta}_{n} -
    \betatrue)|}$, respectively. 
Because for any collection of $\mathcal{W}$ of of length-$p$ row vectors,
\begin{equation}\label{eq:29}
  \sup_{\vec{w} \in \mathcal{W}}|\vec{w}(\hat\beta - \betatrue)| \leq
  \sup_{\vec{w} \in \mathcal{W}}|\vec{w}(\hat\beta - \tilde\beta)| +
\sup_{\vec{w} \in \mathcal{W}}|\vec{w}(\tilde\beta - \betatrue)|, 
\end{equation}
it follows that with sub-gaussian covariates the
worst-case PIC or index error tends to 0
provided that $\bkt[1]{p^{3/2}/n}\log n$ does, i.e. if $p =
o\bkt[3]{\bkt[2]{n/\bkt[1]{\log n}}^{2/3}}$.

When the covariate has sub-$\sqrt{n}$ dimension, the corollaries
indicate that in large samples the suprema of
$\{|\vec{x}_{i}(\hat\beta - \tilde\beta)|:i\leq n\}$
and
$\{|(\vec{x}_{i} - \vec{x}_{j})(\hat\beta - \tilde\beta)|:i,j\leq n\}$
will be smaller by an order of magnitude, $p/n^{1/2} = o_{P}(1)$,
than those of
$\{|\vec{x}_{i}(\tilde\beta - \betatrue)|: i\}$ and
$\{|(\vec{x}_{i} - \vec{x}_{j})(\tilde\beta - \betatrue)|: i,j\}$
(both of which are $O_{P}\bbkt[2]{\bkt[1]{(p /n)^{1/2}\log n}}$).
Of the two errors at right of \eqref{eq:29}, the $\sup |\vec{w}(\tilde\beta
- \betatrue)|$ term ordinarily dominates; 
we turn attention to it. 

\subsection{A thought experiment}
\label{sec:mean-worst-case-Gaussian}

In a special case making both $\vec{X}$ and $\hat\beta$ Gaussian,
sizes of PIC errors admit specific characterization in terms of
readily estimable quantities. \citet{ghoshCortes19}, among others,
consider related issues under an assumption of Gaussian $\vec{X}$. 
For vectors ${v} \in \Re^{m}$ let
$|{v}|_{2}$ and $|{v}|_{\infty}$ have their usual meanings,
$\bkt[1]{m^{-1}\sum_{i=1}^{m}v_{i}^{2}}^{1/2}$ and
$\sum_{i\leq m}|v_{i}|$ \label{defsec:ess-sup-norm} respectively.  For matrices $M$ and $N$ of
like dimension, $\fip{M}{N}$ denotes the Frobenius inner product $\operatorname{tr}(M'N)$.
\label{defsec:fip}

\begin{prop} \label{prop:gaussian-max-pic-e}
  Let $\mathcal{S}$ be a partition of $\{1, \ldots,
  n\}$ with an associated mapping $\vec{\mu}: \mathcal{S} \rightarrow 
  \Re^{p}$, and let
  $\EE_{\mathcal{S}}(\cdot)$ and $\cov_{\mathcal{S}}(\cdot)$ denote
  expectations calculated with $\mathcal{S}$ and $\{\mu_{\mathbf{s}}:
  \mathbf{s}\in \mathcal{S}\}$ held fixed. Given $\mathcal{S}$ and $\{\mu_{\mathbf{s}}:
  \mathbf{s}\in \mathcal{S}\}$ let 
  $\bkt[1]{\tilde{\beta}_{n}-\betatrue; \vec{X}_{i} - \vec{\mu}([i]_{\mathcal{S}});
    \vec{X}_{j} - \vec{\mu}([j]_{\mathcal{S}})}$ be
  jointly multivariate Normal, for each $i,j\leq n$,  with mean zero, $\cov_{\mathcal{S}}\bkt[1]{\tilde\beta} = C$, 
  $\cov_{\mathcal{S}}\bkt[1]{\vec{X}_{i}}=\cov_{\mathcal{S}}\bkt[1]{\vec{X}_{j}}=\Sigma$ and, if $i \neq j$, $\cov_{\mathcal{S}}\bkt[1]{\vec{X}_{i},
  \vec{X}_{j}} = 0$. Then we have
 \begin{align}
   \EE_{\mathcal{S}}\left[\big|\bkt[2]{(\vec{X}_{i} -
        \vec{X}_{j})(\tilde{\beta} - \betatrue)
    : i \stackrel{\mathcal{S}}{\sim} j, i < j}\big|_{2}^{2}\right] &=
                                                            \fip{2\Sigma}{C}\quad
                                                            \text{and}\label{eq:12} \\
    \EE_{\mathcal{S}}\left[\big|\bkt[2]{(\vec{X}_{i} -
        \vec{X}_{j})(\tilde{\beta} - \betatrue)
    : i \stackrel{\mathcal{S}}{\sim} j}\big|_{\infty}\right] &\leq
z^{*}_{\card[0]{\mathcal{S}}}\fip{2\Sigma}{C}^{1/2},\label{eq:14}
 \end{align}
where  $\card[0]{\mathcal{S}}\defeq \card{\{\{i,j\}: i \stackrel{\mathcal{S}}{\sim} j, i
 < j\}}$ and $z^{*}_{\card[0]{\mathcal{S}}}\defeq\bkt[1]{2\log
 2\card[0]{\mathcal{S}}}^{1/2}$ bounds $\EE \max_{1\leq i\leq
n} |G_{i}|$, $(G_{i}: i)$ independent $\mathrm{N}(0,1)$, as described
in Section~\ref{sec:sub-gaussian-random}. 
\end{prop}

Given its strong assumptions on the covariate, 
Proposition~\ref{prop:gaussian-max-pic-e} has limited practical use
for PIC error control; we shall arrive at methods for containment of $\bkt[2]{(\vec{X}_{i} -
        \vec{X}_{j})(\tilde{\beta} - \betatrue)
    : i \stackrel{\mathcal{S}}{\sim} j}$
that relax those assumptions to the moment condition \ref{A-l2Sfinite}.  But these
methods call for a limit on sizes of PIC errors that are to be tolerated, and 
\eqref{eq:14} will turn out to be helpfully specific in this regard.  

The statistician who sets out to select a matched sample has it as her
operating hypothesis that each member of the focal group has
counterparts that are close enough, in terms of $\vec{x}\betatrue$,
within the alternate group.  Simplifying, so as to remove the
question-begging ``close enough,'' let us suppose provisionally that
for each focal group member $i$, the available sample contains within
it at least one contrasting group member $j$ that would be a perfect
match on the underlying index,
$\vec{x}_{i}\betatrue = \vec{x}_{j}\betatrue$.  Continue the thought
experiment by supposing $\hat\beta$ and $\vec{X}$ to be as described
in Proposition~\ref{prop:gaussian-max-pic-e}, and by taking the focal
group to be the smaller of the treatment and control groups, implying
no fewer than $\min(n_{0}, n_{1})$ perfect pairs.  Let $\mathcal{S}$
to be the collection of equivalence classes induced by the relation
that $i \stackrel{\mathcal{S}}{\sim} j$ if and only if
$\vec{x}_{i}\betatrue=\vec{x}_{j}\betatrue$.  The proposition then
characterizes PICs $|(\vec{x}_{i} - \vec{x}_{j})\hat\beta|$ for which
the contrast on the underlying index,
$|(\vec{x}_{i} - \vec{x}_{j})\betatrue|$, is 0.  Each simplification
made thus far in order to apply the proposition should err in the
direction of understating the maximum PIC among pairs closely matched
on $\vec{x}\betatrue$; but even if we continue to arrange our thought
experiment so as to minimize this quantity we will find it to be
almost unworkably large, in a sense to be given presently.

As specified so far, our perfect pairing thought experiment permits
no $\vec{X}\betatrue$ variation within strata of $\mathcal{S}$.
This means the stratified covariance $\cov_{\mathcal{S}}(\vec{X})$ must satisfy
 $\betatrue'\Sigma\betatrue =0$.  Bending available covariate data to
 this constraint, take $\vec{X}$ to be distributed as in
 Proposition~\ref{prop:gaussian-max-pic-e} with 
 $\Sigma = S^{\perp \betatrue}\defeq (n-L)^{-1}\mathbf{x}^{\perp
   \mathbf{x}\betatrue\prime}\mathbf{x}^{\perp \mathbf{x}\betatrue}$, 
the observed covariates' covariance as projected onto the orthocomplement of
 the index, $\mathbf{x}^{\perp \mathbf{x}\betatrue}$.  (Here
 $\mathbf{x}$ is the $n \times p$ matrix of covariates as observed;
 $\mathbf{x}^{\perp v}$ denotes the $n \times p$ matrix of residuals
 arising from the $p$ regressions of $\mathbf{x}$-columns on
 $n$-vector $v$; and $L$ is the number of overt, preexisting strata,
 if such exist, and 1 otherwise. Following \ref{A-centering},
 $\mathbf{x}$ is assumed to be centered or stratum-centered, as
 appropriate, and $\mathbf{x}^{\perp \mathbf{x}\betatrue}$ inherits
 this centering.)  The natural estimate
 $S^{\perp \hat\beta}=(n-L)^{-1}\mathbf{x}^{\perp \mathbf{x}\hat\beta
   \prime}\mathbf{x}^{\perp \mathbf{x}\hat\beta}$ of
 $S^{\perp \betatrue}$ is appropriately consistent for
 $S^{\perp \betatrue}$, as noted in
 Proposition~\ref{prop:picse-consistency} below. Observe that use of
 $S^{\perp \hat\beta}$ (as opposed to $S$) again reduces \eqref{eq:12}
 and \eqref{eq:14}, if in increasing-$p$ regimes it leaves their order
 unchanged.

The maximum PIC bound \eqref{eq:14} is at its smallest, with
$\card[0]{\mathcal{S}}=\min(n_{0}, n_{1})$, when each member of
the smaller of the focal and comparison groups has just one perfectly matching counterpart.  Among such
configurations, the bound is sharpest when the
pairs do not overlap, as in matching without replacement.
(Conditionally given $\tilde\beta$ as well as $\mathcal{S}$ and
$\{\vec{\mu}\bkt[1]{\mathbf{s}}: \mathbf{s}\in \mathcal{S}\}$,
$(\vec{X}_{i} - \vec{X}_{j})(\tilde{\beta} - \betatrue)$ is
independent of
$(\vec{X}_{i'} - \vec{X}_{j'})(\tilde{\beta} - \betatrue)$, for
$i,j,i',j'$ with $\{i, j\} \neq \{i', j'\}$; by the discussion
following \eqref{eq:3} in Section~\ref{sec:sub-gaussian-random}, this
causes inequality \eqref{eq:5} in
Appendix~\ref{sec:proofs-sect-mean-worst-case-Gaussian} to be sharp.)
 Complete the specification of our perfect-pairing thought
experiment by supposing its $\card[0]{\mathcal{S}}= \min(n_{0}, n_{1})$ pairs
to be nonoverlapping.  Then \eqref{eq:14} more closely estimates the width in
$\vec{x}\hat\beta$ of $\mathcal{S}$. 

Across pairs $\mathcal{S}$ constituting the thought experiment,
the maximum PIC error is expected to be
$z^{*}_{\min(n_0,n_1)}\fip{2 S^{\perp \betatrue}}{{C}}^{1/2}$, of order $O\bkt[3]{\bkt[2]{(p \log
    n)/n}^{1/2}}$. The accompanying estimate is
$z^{*}_{\min(n_0,n_1)}\fip{2 S^{\perp \hat\beta}}{\hat{C}}^{1/2}$.  These quantities are small enough to tend to zero,
but only barely so: $\bkt[2]{(p\log n)/n}^{1/2}$ is precisely the rate
that \ref{A-consistencyrates} requires to decline to 0.  Among pairs
$\{i,j\}$ that in actuality are perfectly matched,
$(\vec{x}_{i} -\vec{x}_{j})\betatrue=0$, even tame, Gaussian variation
in other covariate directions engenders a maximum PIC
$|\bkt[2]{(\vec{X}_{i} - \vec{X}_{i})\hat\beta: i
  \stackrel{\mathcal{S}}{\sim} j}|_\infty$ of as large an order as can
be tolerated of separations on the actual index,
$|\bkt[2]{(\vec{X}_{i} - \vec{X}_{i})\betatrue: i
  \stackrel{\mathcal{S}}{\sim} j}|_\infty$, if the matching is to be
asymptotically exact.  So we will recommend this number as a matching
tolerance, not only when making the restrictive assumptions of
Proposition~\ref{prop:gaussian-max-pic-e} but also when entertaining
only the weaker \ref{A-centering}--\ref{A-consistencyrates}.

\subsection{A summary standard error for PICs}
\label{sec:covh-moder-high}

In light of \eqref{eq:4} and \eqref{eq:13}, the covariance of $\tilde{\beta}_{n}$ is
\begin{equation}
 C_{n} = n^{-1}A_{n}^{-1}B_{n}\bkt[1]{A_{n}^{-1}}'. 
\end{equation}

This $C_{n}$ approximates the covariance of $\hat{\beta}$,
particularly when the index model has sub-$\sqrt{n}$ dimension.  When
$p=o\bkt[3]{\bkt[2]{n/\log(n)}^{1/2}}$,
Prop.~\ref{lem:betahat-consist} says
$|\tilde{\beta} - \hat{\beta}|_{2} = o_{P}(n^{-1/2})$, small enough to
obviate distinctions between $C_{n} = \Cov{\tilde\beta}$ and
$\Cov{\hat\beta}$.  For example, Lemma~\ref{lem:C-rate} below entails
that
$\operatorname{s.d.}(\vec{x}\tilde\beta) =
\bkt[2]{\vec{x}{C}_{n}\vec{x}'}^{1/2}$ shares the order
$O_{P}(n^{-1/2}|\vec{x}|_{2})$ with $\vec{x}(\hat{\beta} -\betatrue)$,
whereas part~\ref{lem:betahat-consist-btilde} of
Prop.~\ref{lem:betahat-consist} says
$\vec{x}(\tilde{\beta} - \hat{\beta}) = o_{P}(n^{-1/2}|\vec{x}|_{2})$.
The proposition following the lemma will show the larger order
$O_{P}(n^{-1/2}|\vec{x}|_{2})$ also to be shared by
$\stderr(\vec{x}\tilde\beta) = \bkt[2]{\vec{x}\hat{C}_{n}\vec{x}'}^{1/2}$.

\begin{lemma} \label{lem:C-rate}
  Under~\ref{A-invertibleA}, \ref{A-c0moments}
  and~\ref{A-l2Sfinite}, $|B_{n}|_{2} = O(1)$ and $|C_{n}|_{2} = O(n^{-1})$.
\end{lemma}

If $\psi$ is the gradient a (well-specified) log-likelihood, then
$A_{n}=B_{n}$, and $n^{-1}\hat{A}_{n}^{-1}$ estimates $C_{n}$. More broadly, $C_{n}$ is
estimated by $n^{-1}{\hat{A}_{n}}^{-1} \hat{B}_{n}\bkt[1]{\hat{A}_{n}^{-1}}'$, where
\begin{align*}
\hat{B}_{n}  =&  \hat{B}_n(\hat\beta), \quad \hat{B}_n(\beta) = n^{-1} \sum_{i=1}^{n} \psi(r_{i}, \vec{x}_{i};  {\beta} ) \psi(r_{i}, \vec{x}_{i};  {\beta} )'.  
\end{align*}
The propositions that follow establish the consistency of natural
covariance estimators and related quantities. 

\begin{prop} \label{lem:ChatC}
  Under \ref{A-centering}--\ref{A-consistencyrates}, $|\hat{A}_{n}^{-1} -
  A_{n}^{-1}|_{2} \stackrel{P}{\rightarrow} 0$.  If also
  \ref{A-paramrates} (sub-$\sqrt{n}$ dimension) and \ref{A-boundedXes}
  (sub-gaussian covariates), then $|\hat{A}_{n}^{-1} \hat{B}_{n}
  \bkt{\hat{A}_{n}^{-1}}'
  - nC_{n}|_{2} \stackrel{P}{\rightarrow} 0$. 
\end{prop}


\begin{prop} \label{prop:picse-consistency}
  Let \ref{A-centering}--\ref{A-consistencyrates} hold, let $\hat{C}_{n}$ be a consistent estimate of $C_{n}$ ($|\hat{C}_{n} -
  C_{n}|_{2} = o_{P}(n^{-1})$) and let $S$,  $S^{\perp \betatrue}$
  and  $S^{\perp \hat\beta}$ be as defined in Section~\ref{sec:mean-worst-case-Gaussian}.  Then
\begin{equation*} \label{eq:36}
  | \fip{S}{\hat{C}} -
    \fip{S}{{C}} |= o_{P}\bkt[1]{p/n}
  \quad \text{and} \quad
    | \fip{S^{\perp\hat\beta}}{\hat{C}} -
    \fip{S^{\perp\betatrue}}{{C}} |= o_{P}\bkt[1]{p/n},
  \end{equation*}
  whereas $\fip{S}{{C}}
  =O_{P}\bkt[1]{p/n}$ and $\fip{S^{\perp\betatrue}}{{C}}
  = O_{P}\bkt[1]{p/n}.$
\end{prop}

The quantity $\fip{2 S^{\perp \hat\beta}}{\hat{C}_{\hat\beta}}^{1/2}$
will be termed the \textit{PIC standard error} (PIC SE);
Proposition~\ref{prop:picse-consistency} says that it consistently
estimates the analogous parameter appearing at right of \eqref{eq:12}
and \eqref{eq:14} in Proposition~\ref{prop:gaussian-max-pic-e}.
Proposition~\ref{lem:ChatC} is new as applied to increasing-$p$
regimes; Proposition~\ref{prop:picse-consistency} is entirely
new. Their proofs are given in Appendix~\ref{sec:Cproofs}, along with
demonstrations of intermediate results including
Lemma~\ref{lem:C-rate}.

\subsection{PIC SE calipers} \label{sec:pic-se-calipers}
We recommend matching within limits of
$z^{*}_{\min(n_0,n_1)}$ times the PIC SE
$\fip{S^{\perp\hat\beta}}{\hat{C}}^{1/2}$, whether or not the Gaussian
model of Proposition~\ref{prop:gaussian-max-pic-e} applies. If it does
apply, this ensures that the same multiple of the PIC SE characterizes
matched discrepancies on the underlying index
(Sections~\ref{sec:mean-worst-case-Gaussian} and
\ref{sec:covh-moder-high}). If it does not apply but the covariate is
sub-gaussian, \eqref{eq:14} may no longer limit sizes of PIC errors,
but they continue to tend to 0 as long as
$p = o\bkt[3]{\bkt[2]{n/\bkt[1]{\log n}}^{2/3}}$
(Section~\ref{sec:dimin-index-errors}).  If neither the Gaussian nor
sub-gaussian modeling assumptions apply, additional matching
requirements to be described below will be necessary to force the PIC
errors towards 0, but $z^{*}_{\min(n_0,n_1)}$ times
the PIC SE remains an appropriate tolerance for PICs. It tends to
zero, so its use as a caliper width forces PICs toward zero; it tends
to zero at the same $\bkt[2]{\bkt[1]{p\log n}/n}^{1/2}$ rate that
Proposition \ref{lem:betahat-consist} requires to tend to zero for
index model consistency, making it no more restrictive than is
necessary to force the PIC maximum to tend to 0.  In the context of
the idealized setting studied in
Section~\ref{sec:mean-worst-case-Gaussian} it was seen also to be
minimal in a more quantitatively specific sense, in virtue of its
sharp characterization of the notional experiment's maximum PIC error:
if such a paired experiment were to be lurking within the actual data,
setting a tolerance for matching on $\mathbf{x}\hat\beta$ any smaller
than $z^{*}_{\min(n_0,n_1)}$ times the PIC SE would
exclude pairs that are in fact perfectly matched on
$\mathbf{x}\betatrue$.

\section{Deconstructing Gaussian and sub-gaussian assumptions} \label{sec:errors-paired-index}
\subsection{PIC SEs with unrestricted $X$}\label{sec:r.m.s.-index-s.e}

The PIC SE averages expected PIC errors in either of two ways.  First,
if the sample available for matching contains a subsample of perfectly
matched subjects (as envisioned in
Section~\ref{sec:mean-worst-case-Gaussian}) for which the within-pair covariance of covariates is
$S^{\perp \betatrue}$, then the squared PIC SE
estimates the expected mean of squared PIC errors,
$\bkt[2]{(\vec{x}_{i} - \vec{x}_{j})(\hat\beta -\betatrue)}^{2}$,
across perfectly matched pairs $(i,j)$.  Second, taking the entirety
of the sample as-is but residualizing each subject's covariate for
$\mathbf{x}\betatrue$ (as also discussed in
Section~\ref{sec:mean-worst-case-Gaussian}), the squared PIC SE is
approximately the expected mean square of \textit{reduced} PIC errors,
$\bkt[2]{(\vec{x}_{i}^{\perp \mathbf{x}\betatrue} - \vec{x}_{j}^{\perp
    \mathbf{x}\betatrue}) (\hat\beta -\betatrue)}^{2}$, now across all pairs
$\{i,j\}$, $1\leq i < j \leq n$.

To see this, for $1 \leq i,j\leq n$ write
$\vec{d}_{ij}^{\perp\betatrue}\defeq \vec{x}_{i}^{\perp
  \mathbf{x}\betatrue} - \vec{x}_{j}^{\perp \mathbf{x}\betatrue}$,
noting that
$\vec{d}_{ij}^{\perp\betatrue}=\vec{x}_{i} - \vec{x}_{j}$ for
pairs $\{i,j\}$ that are perfectly matched for the index. Observe that
\begin{align*}
  \bkt[2]{\vec{d}_{ij}^{\perp\betatrue}(\tilde\beta -
  \betatrue)}^{2} &=\bkt[2]{\vec{d}_{ij}^{\perp\betatrue}(\tilde\beta -
  \betatrue)}'\bkt[2]{\vec{d}_{ij}^{\perp\betatrue}(\tilde\beta -
  \betatrue)}\\
  &=(\tilde\beta -
    \betatrue)'\bkt[1]{\vec{d}_{ij}^{\perp\betatrue\prime}\vec{d}_{ij}^{\perp\betatrue}}(\tilde\beta -
    \betatrue).
\end{align*}
Summing over perfectly matched pairs, for the first scenario,
or all ${n \choose 2}$ possible pairs, for the second, and in either case
letting $n_{p}$ denote the number of pairs contributing to the sum, we have
\begin{align}
  \frac{1}{n_{p}}\sum \bkt[2]{\vec{d}_{ij}^{\perp\betatrue}(\tilde\beta -
    \betatrue)}^{2} &=   (\tilde\beta -
    \betatrue)'\bbkt[1]{n_{m}^{-1}\sum \vec{d}_{ij}^{\perp\betatrue\prime}\vec{d}_{ij}^{\perp\betatrue}}(\tilde\beta -
                      \betatrue) \nonumber \\
  &= (\tilde\beta -
    \betatrue)' \bkt[1]{2S^{\perp \betatrue}}(\tilde\beta -
    \betatrue) \label{eq:16}\\
  & = \fip{2 S^{\perp \betatrue}}{(\tilde\beta -
    \betatrue) (\tilde\beta -
    \betatrue)'}\label{eq:6},
\end{align}
where \eqref{eq:16} invokes the U-statistic representation of
covariance, $(n-1)^{-1}\sum_{i=1}(w_{i}- \bar{w})(v_{i}-\bar{v}) = {n
  \choose 2}^{-1}\sum_{i=1}^{n-1}\sum_{j=i+1}^{n}\frac{1}{2}(w_{i}-w_{j})(v_{i}-v_{j})$,  and \eqref{eq:6} uses the sum of
elementwise products $\fip{M}{N}$ of matrices $M$ and $N$ to re-express
\eqref{eq:16}.  Accordingly
\begin{equation*}
    \EE\bbkt[3]{\frac{1}{n_{p}}\sum \bkt[2]{\vec{d}_{ij}^{\perp\betatrue}(\tilde\beta -
    \betatrue)}^{2}} = \fip{2S^{\perp \betatrue}}{C_{n}}.
\end{equation*}

A similar argument reveals the PIC SE's alternate interpretation as
root mean square of pairwise distances
$|\vec{d}_{ij}^{\perp\hat\beta}\hat{C}^{1/2}|_{2}=\bkt[1]{\vec{d}_{ij}^{\perp\hat\beta}\hat{C}\vec{d}_{ij}^{\perp\hat\beta\prime}}^{1/2}$
over pairs $\{i,j\}$. Invoking in turn the cyclic
property of the matrix trace, the definition of the Frobenius matrix
product $\fip{\cdot}{\cdot}$, the Frobenius product's bilinearity, and
the U-statistic representation of sample covariance:
\begin{align} \nonumber
                {\sum
                                                                     {\vec{d}_{ij}^{\perp\hat\beta}\hat{C}{d}_{ij}^{\perp\hat\beta\prime}}}
 &=
    {\sum
                                                                     \operatorname{tr}\bkt[1]{{\vec{d}_{ij}^{\perp\hat\beta}\hat{C}\vec{d}_{ij}^{\perp\hat\beta \prime}}}}
 =
    {\sum \operatorname{tr}\bkt[1]{{\vec{d}_{ij}^{\perp\hat\beta \prime}\vec{d}_{ij}^{\perp\hat\beta}\hat{C}}}}\\
  &= \nonumber 
    {\sum\fip{\vec{d}_{ij}^{\perp\hat\beta \prime}\vec{d}_{ij}^{\perp\hat\beta}}{\hat{C}}}
    = \fip{\sum\vec{d}_{ij}^{\perp\hat\beta \prime}\vec{d}_{ij}^{\perp\hat\beta}}{\hat{C}}
    = n_{p}\fip{2S^{\perp\hat\beta}}{\hat{C}}.
\end{align}
When the average is only over $\{i,j\}$ that are
perfectly matched for the index,
$\bkt[1]{\vec{d}_{ij}^{\perp\hat\beta}\hat{C}{d}_{ij}^{\perp\hat\beta\prime}}^{1/2}$
approximates $\operatorname{s.d.}[(\vec{x}_{i}-\vec{x}_{j})\hat\beta]$,
because $\vec{x}_{i}\betatrue = \vec{x}_{j}\betatrue$ means that
$\vec{d}_{ij}^{\perp\hat\beta}$ approximates $\vec{x}_{i} - \vec{x}_{j}$. The
interpretation 
as a pairwise covariate distance, within the
orthocomplement in $\mathbf{x}$ of $\mathbf{x}\hat\beta$ and after rescaling by $\hat{C}^{1/2}$,
is available both for the perfect pairing thought experiment and also
when the mean is over all $\{i,j\}$, $1\leq i<j\leq n$.

These arguments rely implicitly on
\ref{A-centering}--\ref{A-consistencyrates}, via Proposition
\ref{prop:picse-consistency}, 
but do not call for Gaussian covariates, nor for
boundedness of covariates' sub-gaussian norms. By the same token, none
admit extensions offering maximum, rather than average, PIC error
control, as is necessary for asymptotically exact matching.

\subsection{Caliper refinement with attention to index error distances} \label{sec:benchm-index-s.e.s}
It is intuitive that with covariates drawn from
heavy tailed distributions there may be PICs exceeding the PIC SE
by factors well above $z^{*}_{n_{p}}=\sqrt{\bkt{2 \log 2 n_{p}}}$, in contrast to the
Gaussian situation described in
Proposition~\ref{prop:btilde-subgnorm}. Writing $\delta\bkt[2]{x}$ for
the distribution placing point mass at $x$,  heavier tails on the
covariate mean heavier tails on the empirical distributions
${n \choose 2}^{-1}\sum_{i=1}^{n-1}\sum_{j=i+1}^{n}
\delta\bkt[2]{\operatorname{s.d.}\bkt[2]{(\vec{x}_{i}-\vec{x}_{j})\tilde\beta}}$, 
${n \choose 2}^{-1}\sum_{i=1}^{n-1}\sum_{j=i+1}^{n}
\delta\bkt[2]{(\vec{x}_{i}-\vec{x}_{j})(\tilde\beta - \betatrue)}$,
and in turn
${n \choose 2}^{-1}\sum_{i=1}^{n-1}\sum_{j=i+1}^{n}
\delta\bkt[2]{(\vec{x}_{i}-\vec{x}_{j})(\hat\beta - \betatrue)}$.

Fortunately, estimates $|(\vec{x}_{i}-\vec{x}_{j})\hat{C}^{1/2}|_{2}$
of standard deviations
$\operatorname{s.d.}(i,j) =
|(\vec{x}_{i}-\vec{x}_{j})\hat{C}^{1/2}|_{2}$ are available at the
time of matching: one can simply avoid pairings $\{i,j\}$ for which
$|(\vec{x}_{i}-\vec{x}_{j})\hat{C}^{1/2}|_{2}$ is too large.  Call
$|(\vec{x}_{i}-\vec{x}_{j})\hat{C}^{1/2}|_{2}$ the \textit{index error
  distance} separating $i$ from
$j$. Proposition~\ref{prop:stderr-consist} says index error distances
can estimate pairwise index sampling variabilities uniformly well,
even for models of sub-$\sqrt{n}$ dimension if the estimator $\hat{C}$
is appropriately chosen.

\begin{prop} \label{prop:stderr-consist} 
  \textit{i.}  Let
  $\psi(r, \vec{x}, \beta)$ be the
  gradient of a likelihood function governing the conditional
  distribution of $R$ given $\vec{X}$, with $\hat{C}_{n} =
  n^{-1}\hat{A}_{n}^{-1}$. Under
  \ref{A-centering}--\ref{A-consistencyrates}, 
  \begin{equation}\label{eq:1}
    \left|\bbkt[1]{1 -
        \frac{|\vec{x}_{i}\hat{C}^{1/2}|_{2}}{|\vec{x}_{i}{C}^{1/2}|_{2}}:
      i\leq n}\right|_{\infty},\,
    \left|\bbkt[2]{1 -
        \frac{|(\vec{x}_{i}-\vec{x}_{j})\hat{C}^{1/2}|_{2}}{|(\vec{x}_{i}-\vec{x}_{j}){C}^{1/2}|_{2}}:
      i, j\leq n }\right|_{\infty}
    \stackrel{P}{\rightarrow} 0.
  \end{equation}
  (Here $0/0$ is taken to be 1.)
  \begin{enumerate} \setcounter{enumi}{1}
\item \label{item:stderr-consist-no-fullrankB}Let $\{\epsilon_{n}: n\}$ satisfy
  $\epsilon_{n}^{-1} =
  O\bkt[2]{n/\max(p, \log n)}$. Under \ref{A-centering}--\ref{A-boundedXes} with
  $\hat{C}_{n}=n^{-1}\hat{A}_{n}^{-1}\hat{B}_{n}\hat{A}_{n}^{-1}$,
  \begin{equation*}
    \label{eq:9}
    \left|\bbkt[2]{1 - \frac{%
        \max\bkt[1]{\epsilon_{n}^{1/2}, |\vec{x}_{i}\hat{C}^{1/2}|_{2}}%
      }{%
        \max\bkt[1]{\epsilon_{n}^{1/2}, |\vec{x}_{i}{C}^{1/2}|_{2}}
      } : i}\right|_{\infty}
    \stackrel{P}{\rightarrow} 0.
  \end{equation*}
  If $\mathbf{x}$ and $\mathcal{E}_{n} \subseteq \{\{i, j\}:
  1\leq i,j\leq n\}$ satisfy $|(|\vec{x}_{i} - \vec{x}_{j}|_{2}^{2}:
  \{i, j\} \in \mathcal{E}_{n})|_{\infty} = O_{P}(\max(p, \log n))$,
  then under  \ref{A-centering}--\ref{A-paramrates}
  \begin{equation*}
        \left|\bbkt[3]{1 - \frac{%
        \max\bkt[2]{\epsilon_{n}^{1/2}, |(\vec{x}_{i}-\vec{x}_{j})\hat{C}^{1/2}|_{2}}%
      }{%
        \max\bkt[2]{\epsilon_{n}^{1/2},
          |(\vec{x}_{i}-\vec{x}_{j}){C}^{1/2}|_{2}}
      }:
  \{i, j\} \in \mathcal{E}_{n}}\right|_{\infty}
    \stackrel{P}{\rightarrow} 0.
  \end{equation*}
\item \label{item:stderr-consist-with-fullrankB} Under \ref{A-centering}--
  \ref{A-fullrankB}, \eqref{eq:1} holds with 
  $\hat{C}_{n}=n^{-1}\hat{A}_{n}^{-1}\hat{B}_{n}\hat{A}_{n}^{-1}$.
\end{enumerate}
\end{prop}
We focus on situations conforming to the hypotheses of (i) or of
(iii), warranting uniform convergence \eqref{eq:1} of index error
distances.

Let us calibrate sizes of index error distances with reference to
Section~\ref{sec:mean-worst-case-Gaussian}'s perfect pairing thought experiment. 
Proposition~\ref{prop:chaosbound} adapts extant results about Gaussian
chaos to characterize that setting's maximum of 
$|(\vec{X}_{i}-\vec{X}_{j}){C}^{1/2}|_{2}$.

\begin{prop} \label{prop:chaosbound}
  Let $\vec{X}_{i}$, $i\leq n$, be independent
  $\mathrm{MVN}(\mu, \Sigma)$. Let $C$ be a second positive semidefinite matrix of the same
  dimension as $\Sigma$, let $\mathcal{E} \subseteq \{\{i, j\}:
  1\leq i\neq j\leq n\}$ and let $\card[0]{\mathcal{E}}\defeq \card{\mathcal{E}}$.
Then
  \begin{multline}
    \bbkt[3]{\EE\bbkt[2]{\big|\{|(\vec{X}_{i} - \vec{X}_{j})C^{1/2}|_{2}
    : \{i, j\} \in \mathcal{E}\}\big|_{\infty}^{2}}}^{1/2}
    \leq  \\ \fip{2\Sigma}{C}^{1/2}\bbkt[2]{1 +
      \bbkt[1]{\frac{\log \card[0]{\mathcal{E}}}{p_{[\Sigma^{1/2}C\Sigma^{1/2}]}}}^{1/2}},     \label{eq:40}
    \end{multline}
    where $p_{[M]}$ denotes {intrinsic dimension},
    $\operatorname{tr}(M)/|M|_{2}$, for positive semidefinite $M$. 
\end{prop}

Proposition~\ref{prop:chaosbound} is proved in
Appendix~\ref{sec:proofs-sect-mean-worst-case-Gaussian}. With $\Sigma=S^{\perp \betatrue}$
and $C=C_{n}$, it explicitly bounds the worst-case pairwise distance
$|(\vec{X}_{i} - \vec{X}_{j})C^{1/2}|_{2}$ within the perfect-pairing thought experiment of
Section~\ref{sec:mean-worst-case-Gaussian}.  To contain PIC errors of
actual experiments to a similar level, \Iwe/ recommend the
match-eligibility requirement that  
\begin{equation}\label{eq:66}
|(\vec{x}_{i} - \vec{x}_{j})\hat{C}^{1/2}|_{2} \leq  \fip{2S^{\perp \hat\beta}}{\hat{C}}^{1/2}\bbkt[2]{2 +
      \bbkt[1]{\frac{\log \min(n_{0}, n_{1})}{p-1}}^{1/2}},
\end{equation}
as a complement to the $z^{*}_{\min(n_0,n_1)}\fip{2S^{\perp \hat\beta}}{\hat{C}}^{1/2}$ limit on PICs $|(\vec{x}_{i} -
\vec{x}_{j})\hat\beta|$ that was recommended in Section~\ref{sec:pic-se-calipers}.
The heuristic by which Proposition~\ref{prop:chaosbound} supports
constraint \eqref{eq:66}, to be explained presently, more simply
suggests the stricter cap on $|(\vec{x}_{i} - \vec{x}_{j})\hat{C}^{1/2}|_{2}$ of
\begin{equation} \label{eq:24}
  \fip{2S^{\perp \hat\beta}}{\hat{C}}^{1/2}\bbkt[2]{1 +
      \bbkt[1]{\frac{\log \min(n_{0}, n_{1})}{p-1}}^{1/2}};
\end{equation}
but it will subsequently be shown that \eqref{eq:66} together with a
softer penalty on index error distances respecting \eqref{eq:66} while
exceeding \eqref{eq:24}
is sufficient for present purposes.

To relate \eqref{eq:40} to \eqref{eq:24}, first recall that
$\min(n_{0}, n_{1})$ is the size of the perfect-pairing thought
experiment (Sec.~\ref{sec:mean-worst-case-Gaussian}) and $\fip{2S^{\perp
    \hat\beta}}{\hat{C}}$ is consistent for $\fip{2S^{\perp
    \betatrue}}{{C}_{n}}$ (Prop.~\ref{prop:picse-consistency}).
In general $0 \leq p_{[\Sigma^{1/2}C\Sigma^{1/2}]}\leq p$, by definition \citep[see also][\S 7.8]{tropp15,vershynin18HDPbook}.  In the
special cases that $(R, \vec{X})$ has linear discriminant
structure with $\operatorname{Cor}(\vec{X})$ known,
or that $R$ is linear in $\vec{X}$, with $\vec{X}$ Gaussian and
$\betatrue$ estimated accordingly in either case, intrinsic and
extrinsic dimensions coincide:
$p_{[\Sigma^{1/2}C_{n}\Sigma^{1/2}]} = p$ or $p-1$, depending as
$\Sigma=S^{(x)}$ or $S^{\perp
  \betatrue}$. 
If supposed to contain $\min(n_{0}, n_{1})$ distinct pairs $\{i,j\}$
for which $i\neq j$ but $\vec{x}_{i}\betatrue = \vec{x}_{j}\betatrue$,
then either of these Gaussian-$\vec{X}$ special cases closely models
Section~\ref{sec:mean-worst-case-Gaussian}'s notional perfect pairing,
with $\Sigma = S^{\perp \betatrue}$, and \eqref{eq:24} estimates the
expected maximum covariate distance
$|(\vec{x}_{i} - \vec{x}_{j}){C}^{1/2}|_{2}$ among perfect pairs.  At
the same time, an $|(\vec{x}_{i} - \vec{x}_{j})\hat{C}^{1/2}|_{2}$
limit of form \eqref{eq:24} should rarely exclude perfect pairs
because of separation in directions orthogonal to the index, since
\eqref{eq:40} approximates such separations' expected maximum
from above.

Regardless of what distribution the covariate may have been drawn
from, limits \eqref{eq:24} or \eqref{eq:66} on covariate distances
$|(\vec{x}_{i} - \vec{x}_{j})\hat{C}^{1/2}|_{2}$ also engender limits on
PIC errors $|(\vec{x}_{i} - \vec{x}_{j})(\hat\beta - \betatrue)|$.  In
part this is because
$(\vec{x}_{i}-\vec{x}_{j})(\tilde\beta-\betatrue)$ is sub-gaussian if
$\tilde\beta$ is, and 
is $\mathrm{N}\bkt[2]{0, |(\vec{x}_{i}-\vec{x}_{j})C_{n}^{1/2}|_{2}^{2}}$
if $\tilde\beta \sim \mathrm{MVN}(\betatrue, C_{n})$.
Proposition~\ref{prop:max-E-PIC-errors}, stated here without proof,
collects the relevant facts reviewed in
Section~\ref{sec:sub-gaussian-random}.

\begin{prop} \label{prop:max-E-PIC-errors}
  Let $\mathcal{E} \subseteq \{\{i, j\}:
  1\leq i, j\leq n\}$ satisfy $|(\vec{x}_{i} -
  \vec{x}_{j})C_{n}^{1/2}|_{2} >0$ for all $\{i, j\} \in
  \mathcal{E}$, and let $\card[0]{\mathcal{E}} =\card{\mathcal{E}}$.
  If $\tilde\beta \sim \mathrm{MVN}(\betatrue,
  C_{n})$,
\begin{equation*}
    \EE
    \bbkt[2]{\left|\left\{
      \frac{(\vec{x}_{i}-\vec{x}_{j})(\tilde\beta-\betatrue)}{%
        |(\vec{x}_{i} - \vec{x}_{j})C_{n}^{1/2}|_{2}}:
      \{i, j\} \in \mathcal{E}\right\}\right|_{\infty}}
\leq z^{*}_{\card[0]{\mathcal{E}}}=\bkt[2]{2\log \bkt[1]{2\card[0]{\mathcal{E}}}}^{1/2}
  \end{equation*}
  and \marginpar{Remove me?}
  \begin{equation}
    \label{eq:42}
    \EE \bbkt[2]{\left|\left\{
          \frac{(\vec{x}^{\perp x\betatrue}_{i}-\vec{x}^{\perp
              x\betatrue}_{j})(\tilde\beta-\betatrue)}{%
            |(\vec{x}^{\perp x\betatrue}_{i} - \vec{x}^{\perp
      x\betatrue}_{j})C_{n}^{1/2}|_{2}
          }: \{i, j\} \in \mathcal{E}\right\}\right|_{\infty}}
    \leq z^{*}_{\card[0]{\mathcal{E}}}.
  \end{equation}
If $\tilde\beta-\betatrue$ is non-Normal but sub-gaussian with
$\|\tilde\beta-\betatrue\|_{\psi_{2}}$ bounded, these expected maximums
continue to be of order $(\log \card[0]{\mathcal{E}})^{1/2}$.
\end{prop}

Together with Prop.~\ref{prop:stderr-consist}, Proposition~\ref{prop:max-E-PIC-errors}
says requiring each of $\min(n_{0}, n_{1})$ pairs $\{i, j\}$ to
have $|(\vec{x}_{i} -
\vec{x}_{j})\hat{C}_{n}^{1/2}|_{2}$ below
\eqref{eq:24} puts the corresponding PIC errors below $z^{*}_{\min(n_{0}, n_{1})}$ times
\eqref{eq:24}. If $p$ increases faster than $\log n$,
then the ratio in \eqref{eq:24} tends to 0, and
\eqref{eq:24} is asymptotically equivalent to the PIC SE.  That is,
confining matching to pairs $\{i,j\}$ for which $|(\vec{x}_{i} -
\vec{x}_{j})\hat{C}_{n}^{1/2}|_{2}$ falls left of \eqref{eq:24} makes
the supremum of matched errors $|(\vec{x}_{i} -
\vec{x}_{j})(\tilde\beta - \betatrue)|$  asymptotically as it would be in the perfect-pairing
thought experiment, given $\log n = o(p)$ but not special conditions on the distribution
of the covariate.  (If $p$ increases no faster than $\log n$, $p =
O(\log n)$,  these
errors somewhat exceed those of the corresponding thought experiment,
but they tend quickly to zero anyway, due to the $p$'s slow increase.)

These considerations suggest \eqref{eq:24} as a hard limit
for pair distances $|(\vec{x}_{i} - \vec{x}_{j})\hat{C}^{1/2}|_{2}$,
but similar control of PIC errors can be had with a simple policy that
encourages matches with $|(\vec{x}_{i} -
\vec{x}_{j})\hat{C}^{1/2}|_{2}$ beneath \eqref{eq:24} while only requiring
\eqref{eq:66}. Make $i$ is eligible for pairing to $j$ if 
\begin{equation} 
  |(\vec{x}_{i} - \vec{x}_{j})\hat\beta| \leq 
z^{*}_{\min(n_0,n_1)}\bkt[2]{ \fip{2S^{\perp \hat\beta}}{\hat{C}}^{1/2}
  -\hat{e}(i,j)}\label{eq:43}
\end{equation}
where 
\begin{equation}
\hat{e}(i,j) \defeq \bbkt[3]{|(\vec{x}_{i} - \vec{x}_{j})\hat{C}^{1/2}|_{2} - \fip{2S^{\perp \hat\beta}}{\hat{C}}^{1/2}\bbkt[2]{1 +
    \bbkt[1]{\frac{\log \min(n_{0}, n_{1})}{p-1}}^{1/2}}}_{+}, \label{eq:81}
\end{equation}
$v_{+}\defeq \max(0,v)$, represents excess in index error distance as
compared to its nominal supremum \eqref{eq:24}. 
If $|(\vec{x}_{i} - \vec{x}_{j})\hat{C}^{1/2}|_{2}$ never exceeds this
nominal supremum, \eqref{eq:43} reduces to the requirement that
$|(\vec{x}_{i} - \vec{x}_{j})\hat\beta| \leq 
z^{*}_{\min(n_0,n_1)}\fip{2S^{\perp \hat\beta}}{\hat{C}}^{1/2}$,
as proposed in Section~\ref{sec:pic-se-calipers}. 
For potential pairings $\{i, j\}$ with
$|(\vec{x}_{i} - \vec{x}_{j})\hat{C}^{1/2}|_{2}$ exceeding
\eqref{eq:24}, Section~\ref{sec:pic-se-calipers}'s PIC allowance of
$z^{*}_{\min(n_0,n_1)}\fip{2S^{\perp
    \hat\beta}}{\hat{C}}^{1/2}$ is reduced in recognition of the pairing's
large standard error. When the index error distance $|(\vec{x}_{i} - \vec{x}_{j})\hat{C}^{1/2}|_{2}$
is so large that its excess $\hat{e}(i,j)$ exceeds the PIC SE ---
or equivalently, so large that \eqref{eq:66} fails --- \eqref{eq:43}
forbids $i$'s pairing with $j$.

This selectively narrowed PIC SE caliper
\marginpar{Abbreviate for version to submit}
has important advantages over
non-varying PIC SE calipers, alone or in combination with calipers of
width \eqref{eq:24} on the pairwise index error distance. Non-varying
PIC SE calipers secure asymptotic exactness of a match only for
sub-gaussian covariates, an assumption that selective narrowing of the
caliper enables us to do without. Coupling a non-varying PIC SE
caliper with a limit on the pairwise index error distance of \eqref{eq:24}
contains the sum
\begin{equation} \label{eq:82}
|[(\vec{x}_{i}-\vec{x}_{j})\hat\beta : i
\stackrel{\mathcal{S}}{\sim} j]|_{\infty}
+ 
  |[(\vec{x}_{i}-\vec{x}_{j})(\tilde\beta-\betatrue) : i
\stackrel{\mathcal{S}}{\sim} j]|_{\infty}
\end{equation}
at the product of $z^{*}_{\min(n_0,n_1)}$
with the right hand side of \eqref{eq:66}, just as the selectively narrowed PIC SE caliper
does, but at the cost of categorically disallowing pairwise index
error distances in excess of \eqref{eq:24}.  In contrast, the
selectively narrowed caliper permits those pairings if their PICs are
sufficient small. This additional tolerance is important because
\eqref{eq:24} systematically underestimates suprema of pairwise index
error distances for some index models, even with Gaussian $\vec{X}$,
because of its use of $p-1$ in lieu of
Proposition~\ref{prop:chaosbound}'s
$p_{[(S^{\perp \betatrue})^{1/2}C (S^{\perp \betatrue})^{1/2}]}$. That
minor embarrassment could be remedied by replacing $p-1$ in
\eqref{eq:24} by
$p_{[(S^{\perp \hat\beta})^{1/2}\hat{C}((S^{\perp
    \hat\beta}))^{1/2}]}$, but then assumption \ref{A-fullrankB}
would become necessary for asymptotic equivalence of \eqref{eq:24} and the
PIC SE
--- an equivalence needed even under \ref{A-consistencyrates},
the weakest of the model dimensionality restrictions considered in
this paper, to force \eqref{eq:82} toward an asymptote of 0.
Assumption~\ref{A-fullrankB} is discussed in the next section. 

\subsection{The contribution of linearization error}
\label{sec:contr-line-error}

Display \eqref{eq:82} omits linearization error.  Unless the estimator
of the index regression is linear in its dependent variable $R$, to
estimate $|[(\vec{x}_{i}-\vec{x}_{j})\betatrue : i \sim j]|$ we must
attend to $|[(\vec{x}_{i}-\vec{x}_{j})(\hat\beta - \tilde\beta) : i
\sim j]|$ as well as $|[(\vec{x}_{i}-\vec{x}_{j})\hat\beta : i \sim j]|$
and
$|[(\vec{x}_{i}-\vec{x}_{j})(\tilde\beta - \betatrue) : i \sim j]|$.
Index estimators are linear in $R$ in the special cases of linear regression
and linear discriminant modeling with fixed correlation, but not for
indices estimated with probit or logistic regression.

Recall from Section~\ref{sec:dimin-index-errors} that with
sub-gaussian covariates (\ref{A-boundedXes}),
$|[(\vec{x}_{i}-\vec{x}_{j})(\hat\beta - \tilde\beta) : i,j \leq
n]|_{\infty}$ tends to 0 provided that
$p = o\bkt[3]{\bkt[2]{n/\bkt[1]{\log n}}^{2/3}}$, and is smaller by an
order of magnitude than
$|[(\vec{x}_{i} - \vec{x}_{j})(\tilde\beta -\betatrue) :
\vec{x}_{i}\betatrue \approx \vec{x}_{j}\betatrue]|_{2}$ measured in
the PIC SE, provided that
$p = o\bkt[3]{\bkt[2]{n/\bkt[1]{\log n}}^{1/2}}$. Matching within
selectively narrowed PIC SE calipers secures these conclusions under
conditions not including \ref{A-boundedXes}. However, depending on the
specific side conditions and estimation routines that are employed,
the matching procedure may need to observe additional caliper
restrictions.

First consider the case that
\ref{A-centering}--\ref{A-consistencyrates} hold, with $C_{n}$
estimated by $n^{-1}\hat{A}_{n}^{-1}$.   The matching requirement \eqref{eq:66},
a consequence of \eqref{eq:43}, ensures that $|[|(\vec{x}_{i} -
\vec{x}_{j})\hat{C}^{1/2}|_{2} : i \sim j]|_{\infty} = O_{P}\bkt[3]{\bkt[2]{
  \max(p, \log n)/n}^{1/2}}$, since (by
Proposition~\ref{prop:picse-consistency}) the PIC SE is
$O_{P}\bkt[2]{(p/n)^{1/2}}$ and since $(p-1)^{-1}\log\min(n_{0}, n_{1}) =
O\bkt[2]{\max\bkt[1]{1, p^{-1}\log n}}$.  Proposition~\ref{prop:stderr-consist} and
$\hat{C} = n^{-1}\hat{A}_{n}^{-1}$ in turn
give $|[|(\vec{x}_{i} - \vec{x}_{j}){A}^{-1/2}|_{2} : i \sim j]|_{\infty} =
O_{P}\bkt[2]{{\max(p, \log n)}^{1/2}}$.  As \ref{A-invertibleA} and
Lemma~\ref{lem:fullrankAinv} in Appendix~\ref{sec:Cproofs} entail
$|A_{n}|_{2}= O(1)$,  $|[|(\vec{x}_{i} - \vec{x}_{j})|_{2} : i \sim j]|_{\infty} =
O_{P}\bkt[2]{{\max(p, \log n)}^{1/2}}$. Whether or not $\vec{X}$ was
drawn from a sub-gaussian distribution, pairs $\{i,j\}$ selected
within selectively narrowed PIC SE
calipers can be no more separated on $\vec{x}$ than they would have been
under sub-gaussian sampling, and Proposition~\ref{lem:betahat-consist} entails
\begin{equation}
  \label{eq:83}
  |[(\vec{x}_{i} -
\vec{x}_{j})(\hat\beta - \tilde\beta) : i \sim j]|_{\infty} =
O_{P}\bkt[2]{pn^{-1}(\log n)^{1/2}{\max(p, \log n)}^{1/2}}.
\end{equation}
Within matched pairs, linearization error is as described in
Section~\ref{sec:dimin-index-errors}, even without \ref{A-boundedXes}
or additional matching restrictions.

When $C_{n}$ is instead estimated by
$n^{-1}\hat{A}_{n}^{-1}B_{n}\hat{A}_{n}^{-1}$,
\ref{A-paramrates} is needed in addition to \ref{A-centering}--\ref{A-consistencyrates}, for consistency of $\hat{C}_{n}$
(by Proposition~\ref{lem:ChatC}).  Then $|[|(\vec{x}_{i} -
\vec{x}_{j})\hat{C}^{1/2}|_{2} : i \sim j]|_{\infty} = O_{P}\bkt[3]{\bkt[2]{
  \max(p, \log n)/n}^{1/2}}$, by a similar argument as above.
Case~(\ref{item:stderr-consist-no-fullrankB}) of Proposition~\ref{prop:stderr-consist}
then gives that
$
{|[|(\vec{x}_{i} - \vec{x}_{j}) {C}^{1/2}|_{2} : i \sim j]|_{\infty}}
= O_{P}\bkt[3]{\bkt[2]{ \max(p, \log n)/n}^{1/2}}$, provided that $\fip{2S^{\perp
    \hat\beta}}{\hat{C}} \bbkt[2]{1 +
    \bbkt[1]{\frac{\log \min(n_{0}, n_{1})}{p-1}}^{1/2}}$ is of the same order as
$\max(p, \log n)/n$. (Proposition~\ref{prop:picse-consistency} gives that it is $O_{P}\bkt[2]{\max(p, \log
  n)/n}$, but here we also require its reciprocal to be $O_{P}\bkt[3]{\bkt[2]{\max(p, \log
  n)/n}^{-1}}$.) At this point
\ref{A-fullrankB} also becomes necessary, to ensure
$|B_{n}^{-1/2}|_{2}=O_{P}(1)$ and thus that
$|C_{n}^{-1/2}|_{2}=O_{P}(n^{1/2})$.  If so, 
Proposition~\ref{lem:betahat-consist} again gives
\eqref{eq:83}. 

The full-rank covariance condition \ref{A-fullrankB} merits careful
consideration in practice, however. It will poorly describe some
otherwise unassailable index models, as partitioners have long been
encouraged to add covariates in such models without regard to their
mutual correlations \citep{rubin:thom:1996}. Fortunately
\ref{A-fullrankB} is straightforward to diagnose, by checking that
neither $S$ nor $\hat{B}_{n}$ is ill-conditioned. If sustainable, it
delivers (in combination with Lemma~\ref{lem:C-rate} and
\ref{A-l2Sfinite}) the needed assurance that the PIC SE declines no
faster than $(p/n)^{1/2}$.  If \ref{A-fullrankB} cannot be sustained
while an assumption that $\fip{S^{\perp \betatrue}}{C_{n}}=O(p/n)$ can
be, we can instead combine that weaker assumption with additional
matching restrictions ensuring that
$\max_{i\sim j}|\vec{x}_{i}-\vec{x}_{j}|=O_{P}\bkt[2]{\max(p, \log
  n)^{1/2}}$.  Matches can be required to fall within Euclidean
distance calipers of width
$\operatorname{tr}(S)^{1/2}\bkt[3]{1 + \bkt[2]{(\log
    n)/(p-1)}^{1/2}}$, or with calipers of width
$s(x_{j})\bkt[3]{1 + \bkt[2]{(\log n)/(p-1)}^{1/2}}$ on each dimension
$j=1, \ldots p$ of $\vec{x}$ separately.  Either way, the combination
of additional assumptions and matching restrictions secures
\eqref{eq:83}, and that the dominant part of the PIC error
$|[(\vec{x}_{i} - \vec{x}_{j})(\hat\beta - \betatrue)]|_{\infty}$ is
$|[(\vec{x}_{i} - \vec{x}_{j})(\tilde\beta - \betatrue)]|_{\infty}$,
not
$|[(\vec{x}_{i} - \vec{x}_{j})(\hat\beta - \tilde\beta)]|_{\infty}$.



\section{Asymptotically exact matching and consistency of impact estimation}
\label{sec:cons-match-estim}
Matching within PIC SE calipers arranges that paired differences of
the index tend uniformly to zero, given mild conditions on the
index model.  For propensity and certain other index models, this
convergence is precisely what is needed to ensure that in the absence
of unmeasured confounding, the matched structure enjoys the same
consistency properties as would be enjoyed were paired differences
on the index uniformly and identically zero.
In this section we assume
$Z \in \{0, 1\}$; the very weak overlap condition
\begin{equation}
  \label{eq:50}
  \PP\bbkt[2]{0<
    \PP(Z=1 \mid \mathbf{X}\betatrue) < 1} =1; 
\end{equation}
that $\EE|Y_{C}|^{1+\delta}, \EE|Y_{T}|^{1+\delta} < \infty$ for some $\delta>0$; and that the mapping
$v \mapsto \operatorname{logit} \bkt[2]{\PP (Z =1 | \mathbf{X}\beta =  v)}$ is Lipschitz. 
If $\mathbf{X}\beta$ is a propensity score modeled on the logit scale,
this mapping is the identity and \eqref{eq:50} follows from
Rosenbaum and Rubin's \citeyearpar{rosenbaum:rubi:1983} overlap
condition, $0< \PP(Z=1 \mid \mathbf{X}) < 1$; if $\mathbf{X}\beta$ is
a risk or prognostic score, \eqref{eq:50} is less restrictive
\marginpar{Add cites\ldots}
than their already weak overlap requirement.

A partition $\mathcal{S}_{n}$ is a \textit{finely stratified design}
\citep{fogarty2018vhatfinelystratified} if it divides
$\{1, \ldots, n\}$ into partition elements
$\mathbf{s} \in \mathcal{S}_{n}$ that satisfy
$\sum_{i \in \mathbf{s}}\indicator{z_{i}=z} \leq 1$ for either or both
of $z=0,1$.  These can be $\atob{1}{m_{0}}$ or $\atob{m_{1}}{1}$
matched sets, for natural numbers $m_{0}, m_{1}$, if not
$\atob{m_{1}}{m_{0}}$ blocks with both $m_{0}, m_{1}\geq 2$; singleton
elements, $\mathbf{s}$ of size $\card[0]{\mathbf{s}}=1$, represent
unmatched units.  \marginpar{Abbreviate for version to submit} Such
$\mathcal{S}_{n}$ may emerge from pair matching, where each
$\mathbf{s}\in \mathcal{S}_{n}$ is either a $\atob{1}{1}$ pair,
$\sum_{i \in \mathbf{s}}\indicator{z_{i}=1} =\sum_{i \in
  \mathbf{s}}\indicator{z_{i}=0} =1$, or an unmatched $\atob{0}{1}$ or
$\atob{1}{0}$ singleton; from matching with multiple controls,
permitting $\atob{1}{m}$, $m\geq 1$, matches as well as singletons;
from 1-nearest neighbor matching, in which $\atob{m}{1}$ but not
$\atob{1}{m}$ sets may arise; or from full matching
\citep{rosenbaum:1991a}, permitting both $\atob{m}{1}$ and
$\atob{1}{m}$ configurations for any $m$; or from full matching with
symmetric restrictions
\citep{stuart:green:2006,fredricksonErricksonHansen20}, permitting
both $\atob{m}{1}$ and $\atob{1}{m}$ matched sets, but only for $m$
falling below a designated $m_{0}$.  The notation
$[i]_{\mathcal{S}_{n}}$ for the partition element
$\mathbf{s} \in \mathcal{S}_{n}$ containing $i$ is abbreviated to
$[i]$ when no partition other than $\mathcal{S}_{n}$ is under
consideration.

Consider\marginpar{Remove denominator from
  $\psi_{\mathcal{S}_{n}}(\eta)$ def?\ldots}
estimates defined as roots of $\psi_{\mathcal{S}_{n}}(\cdot) = 0$, where
\begin{equation}
  \label{eq:44}
\psi_{\mathcal{S}_{n}}(\eta) \defeq \frac{\sum_{\mathbf{s}\in
  \mathcal{S}_{n}}\sum_{i\in \mathbf{s}}\psi_{\mathbf{s}i}(\eta)}{\sum_{\mathbf{s}\in
  \mathcal{S}_{n}}w_{\mathbf{s}}\card[0]{\mathbf{s}}\bar{Z}_{\mathbf{s}}(1-\bar{Z}_{\mathbf{s}})}, \quad
\psi_{\mathbf{s}i}(\eta) \defeq  w_{\mathbf{s}}
\bkt[2]{Y_{i} - \eta \bkt[1]{Z_{i} -
    \bar{Z}_{\mathbf{s}}}}
(Z_{i} -\bar{Z}_{\mathbf{s}}), 
\end{equation}
$\bar{z}_{\mathbf{s}}\defeq
{\card[0]{\mathbf{s}}}^{-1} {\sum_{j\in\mathbf{s}}z_{j}}$ and
$w_{\mathbf{s}}$ is a nonnegative weight determined by
$z_{\mathbf{s}}$ and/or ${\card[0]{\mathbf{s}}}$.
For example, the $z$-coefficient in an ordinary regression of
outcomes $y$ on $z$ and matched-set indicator variables is expressible
as the solution $\hat\tau$ of $\psi_{\mathcal{S}_{n}}(\tau) = 0$ for 
$w_{\mathbf{s}}\equiv 1$, since $(z_{i} - \bar{z}_{[i]}: i)$ is the
residual of $z$'s ordinary regression on matched-set indicator variables.
As a second example, the effect of treatment-on-treated
estimator
\begin{equation*}
\card[0]{\{i:  Z_{i}=1,
  \card[0]{[i]}>1\}}^{-1}\sum_{\{i: Z_{i}=1, \card[0]{[i]}>1\}} Y_{i} -
\operatorname{avg}(Y_{j}: Z_{j}=0, j \sim i)
\end{equation*}
uniquely solves $\psi_{\mathcal{S}_{n}}(\cdot) = 0$ with
$w_{[i]_{\mathcal{S}_{n}}}=0$ for unmatched $i$ and 
$w_{\mathbf{s}}=(1  - \bar{Z}_{s})^{-1}$ for $\mathbf{s}$ with $\card[0]{\mathbf{s}}>1$.

Inferences will reflect $\mathcal{S}_{n}$ by conditioning on
stratum-wise treatment allocations, that is on a sigma field
containing
$\mathcal{F}_{n} \defeq \sigma\bbkt[1]{\sum_{i \in \mathbf{s}}Z_{i}:
  \mathbf{s} \in \mathcal{S}_{n}}$.  Desite this notation,
$\{\mathcal{F}_{n}: n\}$ is not a nested filtration: as a
rule $\mathcal{S}_{m}\not\subseteq \mathcal{S}_{m}$, as strict
containment does not permit the maximum index discrepancy,
$|\{\vec{x}_{i}\betatrue - \vec{x}_{j}\betatrue: i
\stackrel{\mathcal{S}_{n}}{\sim} j\}|_{\infty}$, to decline with
increasing $n$.  The statistic
$[\sum_{i \in \mathbf{s}}Z_{i}: \mathbf{s} \in \mathcal{S}_{n}]$
that defines $\mathcal{F}_{n}$ is
in itself uninformative, $S$-ancillary
\citep{severini2000likelihood,lehmannRomano22} to
parameters defined as roots of $\eta \mapsto
\EE\bkt[2]{\psi_{\mathcal{S}_{n}}(\eta)}$, $\psi_{\mathcal{S}_{n}}$ as
defined in \eqref{eq:44}.
  

Proposition~\ref{prop:Qconsist} says that under mild assumptions about
the regularity of $\{\mathcal{S}_{n}:n\}$ and
$(Y_{T}, Y_{C}, Z)$, the solution of \eqref{eq:44} tends to a
probability limit.  To state the regularity assumptions, write
$\bar{Y}_{\mathbf{s}(z)}=0$ if
$\sum_{i \in \mathbf{s}}\indicator{Z_{i}=z}=0$, for $z=0$ or 1, and
$\bkt[1]{\sum_{i \in \mathbf{s}}\indicator{Z_{i}=z}}^{-1}{\sum_{i\in
    \mathbf{s}}Y_{i}\indicator{Z_{i}=z}}$ otherwise; and let
$\Vns \defeq \bar{Y}_{\mathbf{s}(1)}-\bar{Y}_{\mathbf{s}(0)}-
\EE(\bar{Y}_{\mathbf{s}(1)}-\bar{Y}_{\mathbf{s}(0)} \mid
\mathcal{F}_{n})$.

\begin{prop} \label{prop:Qconsist} Let
  $\{(\vec{X}_{i}, Y_{Ci}, Y_{Ti}, Z_{i}): i\}$ be i.i.d., let
  $\{\mathcal{S}_{n}: n\}$ be finely stratified designs and let
  $\mathcal{F}_{n} =\sigma\bkt[1]{(\sum_{i \in \mathbf{s}}Z_{i}:
    \mathbf{s} \in \mathcal{S}_{n})}$.
  Assume the moment condition that for some $\delta>0$ either: (i) $\EE|Y_{C}|^{1+\delta}$,
  $\EE |Y_{T}|^{1+\delta}< \infty$ and $\card[0]{\mathbf{s}}$ is bounded; or
  (ii) there is a $V$
  and $\delta>0$ with $\EE |V|^{1+\delta} < \infty$ such that for each $n$ and
  $\mathbf{s}  \in \mathcal{S}_{n}$,  $|V|$ stochastically dominates
  $|\Vns|$ given $\mathcal{F}_{n}$.%
  \marginpar{Omit footnote from version to submit.}
\footnote{That is,
the $\mathcal{F}_{n}$-conditional distribution of 
$|\Vns|$ falls at or below the unconditional distribution of $|V|$ in
the usual stochastic ordering where $W \preceq V$ iff $\PP(W>a) \leq
\PP(V>a)$ for all $a$.}
  Let  $\{w_{\mathbf{s}}:
  \mathbf{s} \in \mathcal{S}_{n}\}$ be
  nonnegative, $\mathcal{F}_{n}$-measurable weights, and let $\psi_{\mathcal{S}_{n}}(\cdot)$ be as in \eqref{eq:44}.  Assume that with
  probability one: $m_{n}\rightarrow\infty$, where $m_{n}\defeq \sum_{\mathbf{s}\in \mathcal{S}_{n}}\indicator{w_{\mathbf{s}}\bar{Z}_{\mathbf{s}}(1-\bar{Z}_{\mathbf{s}})>0}$ is the cardinality of $\mathcal{S}_{n}$
  exclusive of unmatched singletons and strata receiving zero weight;
  $w_{\mathbf{s}}\card[0]{\mathbf{s}}\bar{Z}_{\mathbf{s}}(1-\bar{Z}_{\mathbf{s}})$
  is bounded above; and $m_{n}^{-1}\sum_{\mathbf{s}\in
    \mathcal{S}_{n}}w_{\mathbf{s}}\card[0]{\mathbf{s}}\bar{Z}_{\mathbf{s}}(1-\bar{Z}_{\mathbf{s}})$
  is bounded away from 0. 
Conditionally given $\mathcal{F}_{n}$ we then have, for any sigma
fields $\{\mathcal{G}_{n}: n\}$ with $\mathcal{F}_{n} \subseteq \mathcal{G}_{n}$: 
\begin{enumerate}
\item for each $\eta$,
  $\psi_{\mathcal{S}_{n}}(\eta)- \EE\bkt[2]{\psi_{\mathcal{S}_{n}}(\eta)  |
  \mathcal{G}_{n}}  \rightarrow 0$ in probability and in
$L_{1}$; \label{prop:Qconsist:item1} and
\item $\eta \mapsto \psi_{\mathcal{S}_{n}}(\eta)$ and $\eta \mapsto \EE\bkt[2]{\psi_{\mathcal{S}_{n}}(\eta) \mid
    \mathcal{G}_{n}}$ have unique roots
$\hat{\tau}_{n}$ and $\tau_{n}$. \label{prop:Qconsist:item2} 
\item In addition, if there is  $\tau_{0}\in (-\infty, \infty)$ such that
$\tau_{n}\stackrel{P}{\rightarrow}\tau_{0}$, then $\hat{\tau}_{n} \stackrel{P}{\rightarrow}
\tau_{0}$. \label{prop:Qconsist:item3}
\end{enumerate}
\end{prop}

As compared to the classical consistency principle for
i.i.d. observations \citep[Lemma A of
\S~7.2.1]{huber1964robust,serfling80}, Proposition~\ref{prop:Qconsist}
upgrades moment requirements from estimating equation contributions
$\psi(W; \theta)$ being $L_{1}$ to $\Vns$ being $L_{1+\delta}$, some
$\delta >0$. This enables conclusions in terms of $L_{1}$ as well as
in-probability convergence, which in turn accommodates refinement of
$\mathcal{F}_{n}$-conditioning to conditioning on finer sigma fields
reflecting matched variation in index scores.  Specifically,
consider
\begin{equation*}
  \mathcal{G}_{n} \defeq \sigma\bbkt[1]{(\vec{X}_{i}\betatrue : i \leq
n) \cup \mathcal{F}_{n}}.
\end{equation*}
The generating statistic $(\vec{X}_{i}\betatrue, 1\leq i \leq
n; \sum_{j \in \mathbf{s}}Z_{j}, \mathbf{s} \in \mathcal{S}_{n})$ is
again $S$-ancillary to matched treatment-control contrasts such as $\tau_{n}$.

Because remaining information about
$(\vec{X}_{i}, Z_{i})$,  $1\leq i \leq n$, is barred, the
transformed index scores
$\theta_{i}\defeq \operatorname{logit} \bkt[2]{\PP (Z =1 |
  \mathbf{X}\betatrue = \mathbf{x}_{i}\betatrue)}$ determine $\mathcal{G}_{n}$-conditional assignment
probabilities as follows.  If $\mathbf{s} \in \mathcal{S}_{n}$ and
$\zeta: \mathbf{s} \rightarrow \{0,1\}$ satisfies
$\sum_{i \in \mathbf{s}}\zeta_{i} = \sum_{i \in \mathbf{s}}z_{i}$, then
\begin{align}
  \pibs(\zeta)&\defeq \PP(Z_{i}=\zeta_{i} \text{ all } i \in \mathbf{s}\mid \mathcal{G}_{n})
                         \nonumber \\
  \label{eq:45}
  &=
  \begin{cases}
    \frac{\exp(\theta_{i})}{\sum_{j\in \mathbf{s}} \exp(\theta_{j})}, &
    \text{any } i \in \mathbf{s} \text{ s.t. } \zeta_{i}=1, \text{ and } \zeta_{j}=0
    \text{ for all } j\in \mathbf{s}\setminus\{i\}\\
    \frac{\exp(-\theta_{i})}{\sum_{j\in \mathbf{s}} \exp(-\theta_{j})},
    & \text{any } i \in \mathbf{s} \text{ s.t. } \zeta_{i}=0, \text{ and } \zeta_{j}=1
    \text{ for all } j\in \mathbf{s}\setminus\{i\} .\\
  \end{cases}
\end{align}
(Because we assume \eqref{eq:50}, $\theta_{i} \in (-\infty, \infty)$
for all $i$. When $\mathbf{s} = \{i\}$ is an unmatched singleton,
$\pibs(\zeta)=1$ for the sole permissible $\zeta$,
$\{i \mapsto z_{i}\}$.  When $\mathbf{s}$ is a $\atob{1}{1}$ matched
pair $\{i_{1}, i_{2}\}$, one condition in \eqref{eq:45} obtains with
$i=i_{1}$ while the other obtains with $i=i_{2}$, so that
\eqref{eq:45} presents two distinct expressions for
$\pi_{\mathbf{s}}(\zeta)$.  But these expressions then assign the same
value to $\pi_{\mathbf{s}}(\zeta)$, for each 
$\zeta: \{i_{1}, i_{2}\} \rightarrow \{0,1\}$.) Now define
\begin{equation}
  \label{eq:73}
  \tilde{\psi}_{\mathbf{s}}(\eta) \defeq \frac{\sum_{i \in
      \mathbf{s}}\psi_{\mathbf{s}i}(\eta)}{\card[0]{\mathbf{s}}\pibs\bkt[1]{Z_{\mathbf{s}}}}; \quad
  \tilde{\psi}_{\mathcal{S}_{n}}(\eta)\defeq \frac{\sum_{\mathbf{s}
      \in \mathcal{S}_{n}} \tilde{\psi}_{\mathbf{s}}(\eta)}{\sum_{\mathbf{s}\in
  \mathcal{S}_{n}}w_{\mathbf{s}}\card[0]{\mathbf{s}}\bar{Z}_{\mathbf{s}}(1-\bar{Z}_{\mathbf{s}})}.
\end{equation}

In contrast to $\psi_{\mathbf{s}}(\eta)\defeq \sum_{i \in
  \mathbf{s}}\psi_{\mathbf{s}i}(\eta)$,
$\tilde{\psi}_{\mathbf{s}}(\eta)$ cannot be calculated in practice, as its
random denominator involves the unknown
$\betatrue$.  Accordingly $\tilde{\psi}_{\mathcal{S}_{n}}(\cdot)$ lacks
direct application to effect estimation.  However, it is useful for analysis
of estimates based on $\psi_{\mathcal{S}_{n}}(\cdot)$. 

\begin{prop} \label{prop:Pconsist}
  \textit{i.} The unique root of $\eta \mapsto \EE\bkt[2]{\tilde{\psi}_{\mathcal{S}_{n}}(\eta)|\mathcal{G}_{n}}$ is
 \begin{equation}
    \label{eq:72}
\bbkt[1]{\sum_{\mathbf{s}\in
      \mathcal{S}_{n}}\tilde{w}_{\mathbf{s}}\card[0]{\mathbf{s}}}^{-1}
    { \sum_{\mathbf{s}\in \mathcal{S}_{n}}
        \tilde{w}_{\mathbf{s}}
      \sum_{i
    \in \mathbf{s}}{\EE\bkt[1]{Y\mid
        Z=1, \vec{X}\betatrue=\vec{x}_{i}\betatrue} -
      \EE\bkt[1]{Y\mid Z=0, \vec{X}\betatrue=\vec{x}_{i}\betatrue}
    }
  },
\end{equation}
where $\tilde{w}_{\mathbf{s}}\defeq
w_{\mathbf{s}}\bar{z}_{\mathbf{s}}(1-\bar{z}_{\mathbf{s}})$. 
\begin{enumerate} \setcounter{enumi}{1}
\item For all $\eta$ and $n$,
\begin{multline}
      \label{eq:62}
      \left| \EE \bkt[2]{\tilde{\psi}_{\mathcal{S}_{n}}(\eta) -
      \psi_{\mathcal{S}_{n}}(\eta) \mid \mathcal{G}_{n}} \right| \leq \\
      \bkt[2]{\exp\bkt[1]{4 | \{\theta_{i} - \theta_{j}: i \sim
        j\} |_{\infty}} -1} \cdot
      \frac{\sum_{\mathbf{s} \in
          \mathcal{S}_{n}}\tilde{w}_{\mathbf{s}}\card[0]{\mathbf{s}}\EE|\Vns|}{\sum_{\mathbf{s}\in
  \mathcal{S}_{n}}\tilde{w}_{\mathbf{s}}\card[0]{\mathbf{s}}},
\end{multline}
  where $\theta_{i}=\operatorname{logit} \bkt[2]{\PP (Z =1 |
  \mathbf{X}\betatrue =  \mathbf{x}_{i}\betatrue)}$ and $\Vns$ is as
defined in \S~\ref{sec:cons-match-estim}, above.

\item  If $v \mapsto \operatorname{logit}
  \bkt[2]{\PP (Z =1 | \mathbf{X}\beta =  v)}$ is Lipschitz
  and the conditions of Proposition~\ref{prop:Qconsist} hold,
  $|\{(\vec{x}_{i} - \vec{x}_{j})\betatrue: i \sim
  j\}|_{\infty}\rightarrow 0$ entails that the difference of
  $\tau_{n}$ with \eqref{eq:72} tends in probability to 0, where 
$\tau_{n}$ is the unique root 
of $\eta \mapsto \EE\bkt[2]{\psi_{\mathcal{S}_{n}}(\eta) \mid \mathcal{G}_{n}}$.

\item  If the averages \eqref{eq:72} tends in probability to a finite
  limit $\tau_{0}$, then  $\hat{\tau}_{n}\stackrel{P}{\rightarrow}
  \tau_{0}$.
\end{enumerate}
\end{prop}

If the index deconfounds allocation of treatment,  $(Y_{C},
Y_{T}) \perp Z | \mathbf{X}\betatrue$, then \eqref{eq:72} coincides
with the average causal effect
\begin{equation*}
  \frac{\sum_{\mathbf{s}\in \mathcal{S}_{n}}\tilde{w}_{\mathbf{s}}\sum_{i
    \in \mathbf{s}} \EE\bkt[1]{Y_{Ti} - Y_{Ci} \mid \vec{X}\betatrue =
  \vec{x}_{i}\betatrue}}%
{\sum_{\mathbf{s}\in \mathcal{S}_{n}}\tilde{w}_{\mathbf{s}}\card[0]{\mathbf{s}}}.
  \label{eq:65}
\end{equation*}
If the index is a propensity score, this deconfounding 
flows from strong ignorability in the sense of
\citet{rosenbaum:rubi:1983}; if $\mathbf{X}\beta$ is a prognostic score, strong ignorability
entails index strong ignorability under a secondary ``no effect
modification'' condition \citep[Prop.~3]{hansen:2008biometrika}.

\bibliographystyle{asa}
\bibliography{abbrev_long,misc,causalinference,biomedicalapplications}

\appendix
\section{Review of mathematical symbols}

The symbols $|\cdot|_{2}$ and $|\cdot|_{\infty}$ indicate Euclidean
and supremum norms as usual (\S~\ref{defsec:ess-sup-norm}).
For scalar or vector random variables $V$, $\|V\|_{\psi_{2}}$ is the
sub-gaussian norm of $V$; for fixed matrices $M$, $|M|_{2}$ and
$|M|_{F}$ are $M$'s operator and Frobenius norms respectively
(\S~\ref{sec:sub-gaussian-random}). For matrices $M$ and $N$ of like
dimension, $\fip{M}{N}$ is the Frobenius inner product
$\operatorname{tr}(M'N)$ (\S~\ref{defsec:fip}).

Symbols $\hat\beta$, $\psi(r, \vec{x}, \beta)$ and
$\psiC(r, \vec{x}, \beta_{0}+ \vec{x}\beta)$ are defined in
Section~\ref{sec:estim-index-scor}, while
Section~\ref{sec:cons-estim-iss} defines $\betatrue$, $\tilde{\beta}$,
$A_{n}$, and $B_{n}$.  For partitions $\mathcal{S}$ of
$\{1, \ldots, n\}$, $[i]_{\mathcal{S}}$ denotes the subset of
$\{1, \ldots, n\}$ belonging to $\mathcal{S}$ that contains $i$ and
$i \stackrel{\mathcal{S}}{\sim} j$ means there is
$\mathbf{s} \in \mathcal{S}$ with both $i \in \mathbf{s}$ and
$j \in \mathbf{s}$ (\S~\ref{defsec:isimj}).
Section~\ref{sec:mean-worst-case-Gaussian} defines $\card[0]{S}$ (as
$\card{S}$) and $z^{*}_{m}$, for positive integers $m$.
$\indicator{\mathcal{A}}$ is the indicator of event $\mathcal{A}$.
Section~\ref{sec:cons-match-estim} defines
$\bar{v}_{\mathbf{s}}=\card[0]{\mathbf{s}}^{-1}\sum_{j\in
  \mathbf{s}}v_{j}$ for $\mathbf{s} \subseteq \{1, \ldots, n\}$;
associates estimating functions $\psi_{\mathcal{S}_{n}}(\cdot)$ and
$\tilde{\psi}_{\mathcal{S}_{n}}(\cdot)$, and sigma fields
$\mathcal{F}_{n}$ and $\mathcal{G}_{n}$, with partitions
$\mathcal{S}_{n}$ of $\{1, \ldots, n\}$; and also defines
$\bar{v}_{\mathbf{s}(z)}$ for $\mathbf{s} \subseteq \{1, \ldots, n\}$
and $z \in \{0,1\}$.

\section{Proofs for Section~\ref{sec:asympt-exact-post}}
\subsection{Section~\ref{sec:dimin-index-errors}}
\begin{proof}[Proof of Prop.~\ref{prop:btilde-subgnorm}]

  In light of \eqref{eq:4} and IS estimability (\ref{A-estimable}, \ref{A-invertibleA}),
  the difference between 
  $(\tilde{\beta}_{n} - \betatrue)$ and $A_{n}^{-1}
  \frac{1}{n}\bkt[2]{\sum_{j=1}^{n}\psi(R_{j}, \vec{x}_{j};
    \betatrue)}$ has Euclidean norm of order smaller than $n^{-1/2}$. Because it is
  nonrandom, the sub-gaussian norm of this difference is also $o(n^{-1/2})$.  So it
  suffices to show $\|A_{n}^{-1}
  \frac{1}{n}\bkt[2]{\sum_{j=1}^{n}\psi(R_{j}, \vec{x}_{j};
    \betatrue)}\|_{\psi_{2}} = O(n^{-1/2})$.  Given \ref{A-invertibleA},
  for this it suffices in turn to show that $\| \sum_{j=1}^{n}\psi(R_{j}, \vec{x}_{j};
  \betatrue)\|_{\psi_{2}} = O(n^{1/2})$.

  By~\eqref{eq:20}, 
  \begin{align*}
    \| \sum_{i=1}^{n}\psi(R_{i}, \vec{x}_{i}; \betatrue)
    \|_{\psi_{2}} &= \sup_{\gamma: |\gamma|_{2}=1} \| \sum_{i=1}^{n}\gamma'\psi(R_{i}, \vec{x}_{i}; \betatrue)
                    \|_{\psi_{2}} \\
    &= \sup_{\gamma: |\gamma|_{2}=1}\|\sum_{1}^{n} \psiC(R_{i}, \vec{x}_{i},
      \vec{x}_{i}\betatrue) w(\vec{x}_{i})\vec{x}_{i}\gamma\|_{\psi_{2}}.
  \end{align*}
 Let $k_{1}$ be a bound for $\|\psiC(R_{i}, \vec{x}_{i},\vec{x}_{i}\betatrue)
 \|_{\psi_{2}}$,  by~\ref{A-c0moments}. According to the general Hoeffding inequality
 \citep[\S~2.6]{vershynin18HDPbook}, there is a universal $k_{0}$ such that
  \begin{align*}
    \|\sum_{1}^{n} \psiC(R_{i}, \vec{x}_{i},
    \vec{x}_{i}\beta) w(\vec{x}_{i})\vec{x}_{i}\gamma\|_{\psi_{2}}^{2} 
                                                       \leq & 
                                                       k_{0}\sum_{1}^{n}\|\psiC(R_{i}, \vec{x}_{i},
    \vec{x}_{i}\beta) w(\vec{x}_{i})\vec{x}_{i}\gamma
                                                       \|_{\psi_{2}}^{2}\\
                                                         \leq &
                                                           k_{0}k_{1}\sum_{1}^{n}|w(\vec{x}_{i})\vec{x}_{i}\gamma|_{2}^{2}. 
                                                           \text{ So}
    \\
\sup_{\gamma: |\gamma|_{2}=1}\|\sum_{1}^{n} \psiC(R_{i}, \vec{x}_{i},
      \vec{x}_{i}\beta)
    w(\vec{x}_{i})\vec{x}_{i}\gamma\|_{\psi_{2}}^{2}    \leq &
                                                          k_{0}k_{1}
                                                          n\bbkt[1]{\frac{1}{n}\sum_{1}^{n}w(\vec{x}_{i})^{2}}\cdot\\
                   &\sup_{\gamma:
                                                          |\gamma|_{2}=1} \bbkt[1]{\frac{1}{n}\sum_{1}^{n}|\vec{x}_{i}\gamma|_{2}^{2}}.
                                                             \end{align*}
The left-hand side equals the square of $\|\sum_{1}^{n} \psiC(R_{i}, \vec{x}_{i},
      \vec{x}_{i}\beta) \vec{x}_{i}\|_{\psi_{2}}$, whereas
      \ref{A-l2Sfinite} says the product at right is $O(n)$. The proof is complete.   
\end{proof}

\subsection{Proofs for section~\ref{sec:mean-worst-case-Gaussian}}
\label{sec:proofs-sect-mean-worst-case-Gaussian}

The following lemma helps to prove Proposition~\ref{prop:gaussian-max-pic-e}.

\begin{lemma} \label{lem:D-psi-indep}
  Under the conditions of
  Proposition~\ref{prop:gaussian-max-pic-e},  for all $1\leq i< j <
  n$ we have $(\vec{X}_{i} -
  \vec{X}_{j}) \perp \tilde\beta - \betatrue \mid \mathcal{S}, \{\vec{\mu}\bkt[1]{\mathbf{s}}: \mathbf{s}\in \mathcal{S}\}$. 
\end{lemma}
\begin{proof}Recall that $\tilde\beta - \betatrue =
  A_{n}^{-1}\frac{1}{n}\sum_{i=1}^{n}\psi(R_{i}, \vec{X}_{i};
  \betatrue)$.  Suppressing conditioning for $\mathcal{S}, \{\mu_{\mathbf{s}}: \mathbf{s}\in \mathcal{S}\}$  in the notation, 
  \begin{eqnarray*}
    \Cov{A_{n}^{-1}\frac{1}{n} \sum_{i=1}^{n}\psi(R_{i}, \vec{X}_{i}; \betatrue), \vec{X}_{1} -
  \vec{X}_{2}} =& \frac{1}{n} \Cov{A_{n}^{-1}\psi(R_{1}, \vec{X}_{1};
                 \betatrue), \vec{X}_{1}} - \\
    &\frac{1}{n}\Cov{A_{n}^{-1}\psi(R_{2}, \vec{X}_{2};
                 \betatrue), \vec{X}_{2}} \\
    &=0.
  \end{eqnarray*}
  By joint Normality of $\bkt[3]{\bkt[2]{\vec{X}_{i}; \psi'(R_{i}, \vec{X}_{i};
      \betatrue)}: i}$, $\frac{1}{n} A_{n}^{-1}\psi(R_{1}, \vec{X}_{1};
                 \betatrue)$ and $(\vec{X}_{i} - \vec{X}_{j})$ are
                 jointly Normal, and the fact that they are
                 uncorrelated means they are independent.
\end{proof}
\begin{proof}[Proof of Proposition~\ref{prop:gaussian-max-pic-e}]
For fixed $\gamma \in \Re^{p}$ we have, after some algebra that \Iwe/
omit, 
\begin{align*}
\EE \bkt[2]{(\vec{X}_{1} -
        \vec{X}_{2})\gamma}^{2}  &= 2\gamma'\Sigma\gamma \quad \text{and}\\
  \EE \big|\{(\vec{X}_{i} -
        \vec{X}_{j})\gamma
    : i \stackrel{\mathcal{S}}{\sim} j, i \neq j\}\big|_{2}^{2} &= 2\gamma'\Sigma\gamma.
  \end{align*}
(Throughout the proof \Iwe/ write 
````$\EE\bkt[2]{\cdot}$'' for ``$\EE_{\mathcal{S}}\bkt[2]{\cdot}$.'')   By the conditional independence
established in Lemma~\ref{lem:D-psi-indep}, it follows that 
\begin{equation*}
  \EE \big|\{(\vec{X}_{i} -
        \vec{X}_{j})(\tilde\beta - \betatrue)
    : i \stackrel{\mathcal{S}}{\sim} j, i \neq j\}\big|_{2}^{2} = \fip{2\Sigma}{C}.
  \end{equation*}
  
  Also for fixed $\gamma$, \eqref{eq:3} as applied to Normal variables gives
  \begin{equation}
    \EE \big|\{(\vec{X}_{i} -
        \vec{X}_{j})\gamma
    : i \stackrel{\mathcal{S}}{\sim} j, i \neq j\}\big|_{\infty} \leq \bkt[2]{4
      \gamma'\Sigma\gamma \log 2\card[0]{\mathcal{S}}}^{1/2}.\label{eq:5}
  \end{equation}

Since $\tilde{\beta}$ is independent of $\vec{X}_{i} - \vec{X}_{j}$
for each $1\leq i< j \leq n$ (Lemma~\ref{lem:D-psi-indep}), \eqref{eq:5}
entails
  \begin{multline*}
 \EE\bbkt[3]{\EE\bbkt[2]{\big|\{(\vec{X}_{i} -
        \vec{X}_{j})(\tilde{\beta} - \betatrue)
    : i \stackrel{\mathcal{S}}{\sim} j, i \neq j\}\big|_{\infty}\mid
                                                 \tilde{\beta}}^{2}}
    \leq \\ \EE\bbkt[3]{ 4
      (\tilde\beta -\betatrue)'\Sigma(\tilde\beta - \betatrue) \log 2\card[0]{\mathcal{S}}} 
    = 4\fip{\Sigma}{C_{n}}\log 2\card[0]{\mathcal{S}}.
  \end{multline*}
Combining this fact with Jensen's inequality for conditional expectation, 
  \begin{multline*}
    \bbkt[3]{\EE\bbkt[2]{\big|\{(\vec{X}_{i} -
        \vec{X}_{j})(\tilde{\beta} - \betatrue)
    : i \stackrel{\mathcal{S}}{\sim} j, i \neq j\}\big|_{\infty}}}^{2} \leq \\
 \EE\bbkt[3]{\EE\bbkt[2]{\big|\{(\vec{X}_{i} -
        \vec{X}_{j})(\tilde{\beta} - \betatrue)
    : i \stackrel{\mathcal{S}}{\sim} j, i \neq j\}\big|_{\infty}\mid
                                                 \tilde{\beta}}^{2}}
    \leq 4\fip{\Sigma}{C_{n}}\log 2\card[0]{\mathcal{S}}.
  \end{multline*}
\end{proof}

\subsection{Section~\ref{sec:covh-moder-high}}
\label{sec:Cproofs}

\begin{proof}[Proof of Lemma~\ref{lem:C-rate}]
  Letting $k<\infty$ denote the supremum (Condition~\ref{A-c0moments}) of sub-gaussian norms of $\{\psiC(R_{i},\vec{x}_{i}, \vec{x}_{i}\betatrue): i\}$, $\EE[\psiC(R_{i},\vec{x}_{i}, \vec{x}_{i}\betatrue)^4] \leq (2k)^4$ for all $i$ \citep[e.g.,][\S~2.5.2]{vershynin18HDPbook}.   Accordingly  $\sum_{i}\E{\psiC(R_{i},\vec{x}_{i},
 \vec{x}_{i}\betatrue)^{4}} = O(n)$ and in turn $\sum_{i}\bkt[3]{\E{\psiC(R_{i},\vec{x}_{i},
   \vec{x}_{i}\betatrue)^{2}}}^{2} = O(n)$.
Combining this with $\sup_{\gamma:
  |\gamma|_{2}=1} \sum_{i}w(\vec{x}_{i})^{4}(\vec{x}_{i}\gamma)^{4} = O(n)$
(Condition~\ref{A-l2Sfinite}), Cauchy-Schwartz gives $\sup_{\gamma:
  |\gamma|_{2}=1} \sum_{i} \E{\psiC(R_{i},\vec{x}_{i},
 \vec{x}_{i}\betatrue)^{2}}w(\vec{x}_{i})^{2} (\vec{x}_{i}\gamma)^{2} = O(n)$. 
Rearranging terms in light of \eqref{eq:20}, this says
$\sup_{\gamma: |\gamma|_{2}=1} \sum_i \E{[\gamma'\psi(R_i, \vec{x}_i, \beta)]^2} = 
O(n)$, or  $\sup_{\gamma: |\gamma|_{2}=1} \gamma' n B_{n}\gamma =
O(n)$; thus $|B_{n}|_2 = O(1)$.   Condition~\ref{A-invertibleA} gives 
$|{A}_{n}^{-1}|_{2}=O(1)$, so also $|{C}_{n}^{-1}|_{2}=O(n^{-1})$.
\end{proof}

Our demonstration of Proposition~\ref{lem:ChatC} relies on three supporting lemmas, as follows.

\begin{lemma} \label{lem:fullrankAinv}
  Under \ref{A-psismooth} and \ref{A-l2Sfinite},
  $\sup_{\beta}|A_{n}(\beta)|_{2}=O(1)$. 
\end{lemma}

\begin{proof}
  Write $\psiCi(\vec{x}, \eta) \defeq
(\partial/\partial \eta) \EE \bkt[2]{\psiC(R, \vec{x}, \eta) \mid \vec{X}=\vec{x}}$ so that 
$\nabla_{\beta} \EE\bkt[2]{\psi(R, \vec{x},
\beta) \mid \vec{X}=\vec{x}}  = \psiCi(\vec{x}, \vec{x}\beta)
w(\vec{x})\vec{x}'\vec{x}$. By \ref{A-psismooth}, there is $K_{1}<
\infty$ such that 
$|\psiCi(\vec{x}, \beta_{0}+\vec{x}\beta)| < K_{1}$, for any
$\beta$. By Cauchy-Schwartz, \ref{A-l2Sfinite} gives $\sup_{\gamma:
  |\gamma|_{2}=1}\sum_{i=1}^{n}w(\vec{x}_{i}) (\vec{x}_{i}\gamma)^{2}
= O(n)$.  Since 
\begin{equation*}
  n|A_{n}(\beta)|_{2} = \sup_{\gamma:
  |\gamma|_{2}=1} \gamma'\bbkt[2]{\sum_{i}\psiCi(\vec{x}_{i}, \vec{x}_{i}\beta)
w(\vec{x}_{i})\vec{x}'\vec{x}}\gamma  \leq K_{1}\sup_{\gamma:
  |\gamma|_{2}=1}\sum_{i}w(\vec{x}_{i})  (\vec{x}\gamma)^{2}
\end{equation*}
for any $\beta$, the result follows.
\end{proof}
\begin{lemma} \label{lem:ABhatbetahatABhatbetaT}
    Under \ref{A-centering}--\ref{A-consistencyrates},
   $|{A}_{n}(\hat\beta) - {A}_{n}(\betatrue)|_2 \stackrel{P}{\rightarrow} 0$ and
   $|\hat{B}_{n}(\hat\beta) - \hat{B}_{n}(\betatrue)|_2 \stackrel{P}{\rightarrow} 0$.   
\end{lemma}

\begin{lemma} \label{lem:BhatbetatrueB}
  Under \ref{A-centering}--\ref{A-consistencyrates} as well as
  \ref{A-boundedXes}  and sub-$\sqrt{n}$ dimension,
   $|\hat{B}_{n}(\betatrue) - B_{n}(\betatrue)|_2 \stackrel{P}{\rightarrow} 0$.   
\end{lemma}
  
Proofs of Lemmas~\ref{lem:ABhatbetahatABhatbetaT}
and~\ref{lem:BhatbetatrueB} are given following the proof of
Proposition~\ref{lem:ChatC}.

\begin{proof}[Proof of Proposition~\ref{lem:ChatC}]
Since $\hat{A}_{n} \equiv {A}_{n}(\hat\beta)$,  
$|\hat{A}_{n} - A_{n}|_{2} \stackrel{P}{\rightarrow} 0$ 
follows from Lemma~\ref{lem:ABhatbetahatABhatbetaT}.
Since $\hat{B}_{n} \equiv \hat{B}_{n}(\hat\beta)$, 
$|\hat{B}_{n} - B_{n}|_{2} \stackrel{P}{\rightarrow} 0$
follows from
$|\hat{B}_{n}(\hat\beta) - \hat{B}_{n}(\betatrue)|_{2} \stackrel{P}{\rightarrow} 0$
(Lemma~\ref{lem:ABhatbetahatABhatbetaT}) and
$|\hat{B}_{n}(\betatrue) - B_{n}|_2
\stackrel{P}{\rightarrow} 0$ (Lemma~\ref{lem:BhatbetatrueB}).

Since $|{A}_{n}^{-1}|_{2}=O(1)$ and $|\hat{A}_{n}^{-1}|_{2}=
O_{P}(1)$ (Condition~\ref{A-invertibleA} and Prop.~\ref{lem:betahat-consist}),
it follows that $|\hat{A}_{n}^{-1} - A_{n}^{-1}|_{2} \stackrel{P}{\rightarrow} 0$, by applying sub-multiplicativity of the spectral norm to the right-hand side of 
$(\hat{A}_{n}^{-1} - A_{n}^{-1}) = \hat{A}_{n}^{-1} (\hat{A}_{n} - A_{n}) {A}_{n}^{-1}$.  Since also $|B_{n}|_{2} = O_{P}(1)$ (Lemma~\ref{lem:C-rate}), 
the 2-norms of the second and third summands in
\marginpar{Mostly OK but odd\ldots}
$\hat{A}_{n}^{-1}\hat{B}_{n}\hat{A}_{n}^{-1}=$
\begin{equation*}
  A_{n}^{-1}B_{n}A_{n}^{-1}  + {(\hat{A}_{n}^{-1} - A_{n}^{-1})B_{n}A_{n}^{-1}} + {\hat{A}_{n}^{-1}B_{n}(\hat{A}_{n}^{-1} - A_{n}^{-1})}   + \hat{A}_{n}^{-1}(\hat{B}_{n} - B_{n})\hat{A}_{n}^{-1}
\end{equation*}
must tend in probability to 0.  Thus the stochastic order of
$|A_{n}^{-1}B_{n}A_{n}^{-1} -
\hat{A}_{n}^{-1}\hat{B}_{n}\hat{A}_{n}^{-1}|_{2}$ can be no greater
than that of $|\hat{A}_{n}^{-1}(\hat{B}_{n} - B_{n})
\hat{A}_{n}^{-1}|_{2}$.  But as~\ref{A-psismooth} and~\ref{A-l2Sfinite}
entail that $|\hat{A}_{n}|_{2} = O_{P}(1)$, by Lemma~\ref{lem:fullrankAinv}, 
$O_{P}(|\hat{A}_{n}^{-1}(\hat{B}_{n} - B_{n}) \hat{A}_{n}^{-1}|_{2}) =
O_{P}(|\hat{B}_{n} - B_{n}|_{2}) = o_P(1)$; this means
$|C_n - \hat{C}_n|_{2} = n^{-1}|A_{n}^{-1}B_{n}A_{n}^{-1} -
\hat{A}_{n}^{-1}\hat{B}_{n}\hat{A}_{n}^{-1}|_{2} = o_P(n^{-1})$.
\end{proof}
This proof of Lemma~\ref{lem:ABhatbetahatABhatbetaT} was based in part
on Wang's proof of a similar principle for generalized estimating equations 
\citeyearpar[][Thm.~3.10]{wang2011gee}.

\begin{proof}[Proof of Lemma~\ref{lem:ABhatbetahatABhatbetaT}]
  To establish $|A_{n}(\hat\beta) - A_{n}(\betatrue)|_2 = \sup_{\gamma:
  |\gamma|_{2}=1}\gamma'[A_{n}(\hat\beta) -
A_{n}(\betatrue) ]\gamma  \stackrel{P}{\rightarrow} 0$,  let 
$K_1< \infty$ be a Lipschitz constant for $\eta \mapsto \psiCi(r,
\vec{x}, \eta)$, each $r$ and $\vec{x}$, where $\psiCi(\cdot)$ is as
defined in the proof of Lemma~\ref{lem:fullrankAinv}, above. 
(By Condition~\ref{A-psismooth}.%
)
Then 
\begin{multline}
  |\gamma'\bkt[3]{\nabla_{\beta}
  \EE\bkt[2]{\psi(R, \vec{x}_{i},\beta) \mid
    \vec{X}=\vec{x}_{i}}_{\beta = \hat\beta} -
  \nabla_{\beta} \EE\bkt[2]{\psi (R,\vec{x}_{i},\beta) \mid
    \vec{X}=\vec{x}_{i}}_{\beta = \betatrue}}\gamma| \\
\leq K_1 
                                           \bkt[2]{\vec{x}_{i}(\hat\beta -
                                           \betatrue)} w(\vec{x}_{i})
                                           (\vec{x}_{i}\gamma)^{2}. \label{eq:21} 
\end{multline}
Summing over $i$ and applying Cauchy-Schwartz, 
\begin{align}
(\gamma'[A_{n}(\hat\beta) - A_{n}(\betatrue) ]\gamma)^{2} \leq & 
K_1^{2}\left[ \frac{1}{n} \sum_{i=1}^{n}w(\vec{x}_{i})^{2}
                                           (\vec{x}_{i}\gamma)^{4}\right]
\left[ \frac{1}{n} \sum_{i=1}^{n}\left( \vec{x}_{i} \frac{\hat\beta - \betatrue}{|\hat\beta - \betatrue|_{2}}\right)^{2}\right]
\times   |\hat\beta - \betatrue|_{2}^{2}, \nonumber 
\end{align}
interpreting ``$\vec{x}(\delta/|\delta|_{2})$''
as $0$ when $|\delta|_{2} = 0$. 
It follows that
\begin{equation*}
|A_{n}(\hat\beta) - A_{n}(\betatrue)|_2^{2}
\leq
 K_1^{2} \bbkt[2]{\sup_{\gamma: |\gamma|_{2} =1} \frac{1}{n} 
\sum_{i=1}^{n}w(\vec{x}_{i})^{2}(\vec{x}_{i}\gamma)^{4}} \bbkt[2]{
\sup_{\delta: |\delta|_{2} =1} \frac{1}{n} \sum_{i=1}^{n} 
(\vec{x}_{i}\delta)^{2}}
 |\hat\beta - \betatrue |_{2}^{2}.  
\end{equation*}

Observe that the conditions of Proposition~\ref{lem:betahat-consist} follow from
those of Proposition~\ref{lem:ChatC}, so that we may assume $|\hat\beta -
\betatrue|_{2} = O_{P}\bkt[2]{\bkt[1]{p/n}^{1/2}}$.
This combines with Condition~\ref{A-l2Sfinite} to give
$|A_{n}(\hat\beta) - A_{n}(\betatrue)|_{2} =
O(1)O(1) O_{P}\bkt[2]{\bkt[1]{p/n}^{1/2}}$, which by~\ref{A-consistencyrates} 
is $o_{P}(1)$. 

As to $|\hat{B}_{n}(\hat\beta) - \hat{B}_{n}(\betatrue)|_{2}$,
\begin{align}
 \gamma' \{ \Psi(R_{i}, \vec{x}_{i}, \hat\beta) &\Psi(R_{i}, \vec{x}_{i},
\hat\beta)' - \Psi(R_{i}, \vec{x}_{i}, \betatrue) \Psi(R_{i},
\vec{x}_{i}, \betatrue)' \} \gamma \\
=&  \bkt{\psiC^{2}(R_{i},
  \vec{x}_{i}, \hat{\eta}_{i}) - \psiC^{2}(R_{i}, \vec{x}_{i},
                                     {\eta}_{i})} w(\vec{x}_{i})^{2}
                                     (\vec{x}_{i}\gamma)^{2} \nonumber\\
=& \bkt{\psiC(R_{i}, \vec{x}_{i}, \hat{\eta}_{i}) - 
              \psiC(R_{i}, \vec{x}_{i}, {\eta}_{i}) } 
\bkt{\psiC(R_{i}, \vec{x}_{i}, \hat{\eta}_{i}) +
 \psiC(R_{i}, \vec{x}_{i},
   {\eta}_{i})} w(\vec{x}_{i})^{2} (\vec{x}_{i}\gamma)^{2} \nonumber\\
= & \bkt{\psiC(R_{i}, \vec{x}_{i}, \hat{\eta}_{i}) - 
              \psiC(R_{i}, \vec{x}_{i}, {\eta}_{i})
    }^{2}w(\vec{x}_{i})^{2} (\vec{x}_{i}\gamma)^{2} \nonumber \\
& + \bkt{\psiC(R_{i}, \vec{x}_{i}, \hat{\eta}_{i}) - 
              \psiC(R_{i}, \vec{x}_{i}, {\eta}_{i}) } \cdot 2 \psiC(R_{i}, \vec{x}_{i},
   {\eta}_{i}) w(\vec{x}_{i})^{2} (\vec{x}_{i}\gamma)^{2} \nonumber \\
=: & V_{i}(\gamma) + W_{i}(\gamma) . \label{eq:19}
\end{align}
By the Lipschitz property  (\ref{A-psismooth}) of $\psiC(r, \vec{x}, \cdot)$, 
\begin{align}
 \sup_{\gamma: |\gamma|_{2}=1} \frac{1}{n}\sum_{i} |V_{i}(\gamma)| \leq 
& K_{1}^{2}  \sup_{\gamma: |\gamma|_{2}=1} \frac{1}{n}\sum_{i}(\hat{\eta}_{i} -
                                  \eta_{i})^{2}w(\vec{x}_{i})^{2} (\vec{x}_{i}\gamma)^{2} \nonumber\\
\leq &  K_{1}^{2}|\hat\beta -\betatrue|_{2}^{2} 
\sup_{\delta, \gamma:  |\delta|_{2}=|\gamma|_{2}=1}
 \frac{1}{n} \sum_{i}  w(\vec{x}_{i})^{2} (\vec{x}_{i}\delta)^{2}(\vec{x}_{i}\gamma)^{2} \nonumber \\
& = O(1) O_{P}(p/n) O(1) = o_{P}(1), 
  \label{eq:26}
\end{align}
invoking Assumption~\ref{A-l2Sfinite} and consistency of $\hat\beta$ at
\eqref{eq:26}. %
The Lipschitz property of $\psiC(r, \vec{x}, \cdot)$
also gives
\begin{align}
 \sup_{\gamma: |\gamma|_{2}=1} \frac{1}{n}\sum_{i} |W_{i}(\gamma)| \leq 
& 2 K_{1} \sup_{\gamma: |\gamma|_{2}=1} \frac{1}{n}\sum_{i} |\psiC(R_{i}, \vec{x}_{i},
                                                                           \eta_{i}) (\hat{\eta}_{i} - \eta_{i}) |
                                                                           w(\vec{x}_{i})^{2}
                                                                           (\vec{x}_{i}\gamma)^{2} \nonumber\\
\leq & 2 K_{1} |\hat\beta - \betatrue|_{2}
\sup_{\gamma, \delta: |\gamma|_{2}=|\delta|_{2}=1} 
\frac{1}{n}\sum_{i} \psiC(R_{i}, \vec{x}_{i},
  \eta_{i}) (\vec{x}_{i}\delta) w(\vec{x}_{i})^{2} (\vec{x}_{i}\gamma)^{2} \nonumber\\
\leq & 2 K_{1} |\hat\beta - \betatrue|_{2}
\bkt[2]{\frac{1}{n} \sum_{i} \psiC^{4}(R_{i}, \vec{x}_{i},
       \eta_{i})}^{1/4} 
\nonumber \\ 
& \times \sup_{\delta: |\delta|_{2}=1} \bkt[2]{\frac{1}{n} \sum_{i}
       w(\vec{x}_{i})^{4}(\vec{x}_{i}\delta)^{4}}^{1/4}
\sup_{\gamma: |\gamma|_{2}=1} \bkt[2]{\frac{1}{n}\sum_{i} w(\vec{x}_{i})^{4}(\vec{x}_{i}\gamma)^{4}}^{1/2}
\label{eq:34}\\
& = O_{P}(p/n) O_{P}(1) O(1) O(1) = o_{P}(1). \nonumber 
\end{align}
Here we apply Cauchy-Schwartz (twice)
at \eqref{eq:34} and, to pass
to the next line, consistency of $\hat\beta$ as per Proposition~\ref{lem:betahat-consist}, 
Assumption~\ref{A-c0moments} in
combination with Markov's inequality and Assumption~\ref{A-l2Sfinite}.
Since 
\begin{align*}
  |\hat{B}_{n}(\hat\beta) - \hat{B}_{n}(\betatrue)|_{2} =&
\sup_{\gamma: |\gamma|_{2}=1}\gamma' \bkt{\hat{B}_{n}(\hat\beta) -
  \hat{B}_{n}(\betatrue)}\gamma \\
& \leq \sup_{\gamma: |\gamma|_{2}=1} \frac{1}{n}\sum_{i} |V_{i}(\gamma)| 
+ \sup_{\gamma: |\gamma|_{2}=1} \frac{1}{n}\sum_{i} |W_{i}(\gamma)| ,
\end{align*}
the result follows.
\end{proof}

\begin{proof}[Proof of Lemma~\ref{lem:BhatbetatrueB}] 

To control $|\gamma [\hat{B}_{n}(\betatrue) - B_{n}(\betatrue)] \gamma'|$,
fix $\gamma$ with $|\gamma|_2=1$ and consider 
\begin{align*}
\gamma \bkt[2]{\hat{B}_{n}(\betatrue) - B_{n}(\betatrue)} \gamma' =& 
\gamma'\big(n^{-1}\sum_{i\leq n} \bkt[3]{\psiC^{2}(R_{i},  \vec{x}_{i}, \vec{x}_{i}\betatrue) -
  \EE  \bkt[2]{\psiC^{2}(R,  \vec{x}_{i},
  \vec{x}_{i}\betatrue)  \mid \vec{X}=\vec{x}_{i}}} w^{2}(\vec{x}_{i})\vec{x}_{i}'\vec{x}_{i}\big)\gamma \\
=& n^{-1}\sum_{i\leq n} \{\psiC[i]^{2}(R_{i}) -
  \EE  \bkt[2]{\psiC[i]^{2}(R)  \mid \vec{X}=\vec{x}_{i}}\} w^{2}(\vec{x}_{i}) (\vec{x}_{i}\gamma)^{2}. 
\end{align*}
Observe that 
\ref{A-c0moments} entails the random variables $\psiC[i]^2(R_{i}) - \EE\bkt[2]{\psiC[i]^2(R)  \mid \vec{X}=\vec{x}_{i}}$ 
to be sub-exponential with uniformly bounded sub-exponential norm \citep[][\S~2.7]{vershynin18HDPbook}. 
Applying Bernstein's inequality \citep[][\S~2.8]{vershynin18HDPbook},
\begin{multline} \label{eq:32}
  \PP \{ |\gamma [\hat{B}_{n}(\betatrue) - B_{n}(\betatrue)] \gamma'| \geq t\} \leq \\
2\exp\left\{-k_2 \min\left(\frac{n^2t^2}{k_e^2 \sum_{i=1}^n w(\vec{x}_{i})^{4}(\vec{x}_{i}\gamma)^{4}}, \frac{nt}{k_e\max_{i\leq n} w(\vec{x}_{i})^{2} (\vec{x}_{i}\gamma)^{2}}\right) \right\},
\end{multline}
where $k_e$ is a finite upper bound for the sub-exponential norms of
$\psiC[i]^2(R_{i}) - \EE[\psiC[i]^2(R_{i})]$, $i\geq 1$, and $k_2$ is
a universal constant.

Now let $\mathcal{N}$ be a $1/4$-net of the p-dimensional sphere, i.e. a finite subset of $\{\gamma: |\gamma|_2=1\}$ with the property that $\{\gamma: |\gamma|_2=1\}$ is covered by balls centered in $\mathcal{N}$ of radius $1/4$, so that \citep[][\S~4.4.1]{vershynin18HDPbook} 
\begin{equation} \label{eq:28}
\sup_{\gamma: |\gamma|_2=1} \gamma [\hat{B}_{n}(\betatrue) - B_{n}(\betatrue)] \gamma'\leq 
2\sup_{\gamma \in \mathcal{N}} \gamma [\hat{B}_{n}(\betatrue) - B_{n}(\betatrue)] \gamma'. 
\end{equation}
We may select this $\mathcal{N}$ to have cardinality no more than $9^p$ \citep[][Corr.~4.2.13]{vershynin18HDPbook}.  

Since \eqref{eq:32} holds for arbitrary $\gamma$ on the unit sphere, 
it follows that 
\begin{multline} \label{eq:33}
  \frac{1}{2}\PP \{ \sup_{\gamma: |\gamma|_2=1} |\gamma [\hat{B}_{n}(\betatrue) - B_{n}(\betatrue)] \gamma'| \geq t\} \leq \\
\exp\left\{p \log(9) -n k_2 \min\left(\frac{t^2}{k_e^2 n^{-1}\sum_{i=1}^n w(\vec{x}_{i})^{4} (\vec{x}_{i}\gamma)^{4}}, \frac{t}{k_e\max_{i\leq n} w(\vec{x}_{i})^{2} (\vec{x}_{i}\gamma)^{2}}\right) \right\} =\\
\exp\left\{- k_2 \min\left(\frac{nt^2}{k_e^2 n^{-1}\sum_{i=1}^{n} w(\vec{x}_{i})^{4} (\vec{x}_{i}\gamma)^{4}} - {pk_3}, \frac{nt}{k_e\max_{i\leq n} w(\vec{x}_{i})^{2} (\vec{x}_{i}\gamma)^{2}} - {pk_3} \right) \right\},
\end{multline}
where $k_3=\log(9)/k_2$. 
Recalling that $n^{-1}\sum_{i=1}^n w(\vec{x}_{i})^{4}
(\vec{x}_{i}\gamma)^{4} = O(1)$ (by \ref{A-l2Sfinite}), 
${t^2}/\bkt[2]{k_e^2 n^{-1}\sum_{i=1}^n w(\vec{x}_{i})^{4} (\vec{x}_{i}\gamma)^{4}}$ has positive limit infimum and finite limit supremum (for any $t$); thus  
the first of the two quantities of which the minimum is taken tends to
$\infty$ because $p = o(n)$ (\ref{A-consistencyrates}). 
Since $\max_{i\leq n} w(\vec{x}_{i})^{2}(\vec{x}_{i}\gamma)^{2} =
O\bkt[3]{\max\bkt[2]{p \log n,\bkt[1]{\log n}^{2}}}$
(\ref{A-boundedXes}), $p^2\log n = o(n)$
(sub-$\sqrt{n}$ dimension) entails that the second quantity also must
tend to $\infty$. (If $p_{n}\leq \log n$ infinitely often, then on the
subsequence for which this is true the term in question is bounded
below by $(\log n)\bkt[3]{(t/k_{e})n/\bkt[2]{(\log n)\max_{i\leq n}
    w(\vec{x}_{i})^{2}(\vec{x}_{i}\gamma)^{2}}-k_{3}}$, which tends to
$\infty$ with $(\log n)\bkt[2]{n/(\log n)^{3} -1}$.  Otherwise
$p_{n}>\log n$ so that $\max\bkt[2]{p_{n} \log n,\bkt[1]{\log
    n}^{2}}=p_{n}\log n$. The term
in question equals $p\bkt[3]{(t/k_{e})n/\bkt[2]{p\max_{i\leq n}
    w(\vec{x}_{i})^{2}(\vec{x}_{i}\gamma)^{2}}-k_{3}}$, which tends to
$\infty$ because $p \uparrow \infty$, $\max_{i\leq n}
w(\vec{x}_{i})^{2}(\vec{x}_{i}\gamma)^{2} = O(p\log n)$ and $p^2\log n = o(n)$.) 
So the minimum in~\eqref{eq:33} increases without bound, and \eqref{eq:33} itself tends to 0.
\end{proof}

Proposition~\ref{prop:picse-consistency}'s proof uses two supporting lemmas.

\begin{lemma} \label{lem:is-sds-same-limit}
  Under \ref{A-centering} and \ref{A-l2Sfinite}, $s(\mathbf{x}\hat\beta) -
  s(\mathbf{x}\betatrue) = O_{P}(|\hat\beta - \betatrue|_{2})$.  
\end{lemma}
\begin{proof}
  By \ref{A-centering}, 
  \begin{align*}
    s^{2}(\mathbf{x}\hat\beta) -
  s^{2}(\mathbf{x}\betatrue) &= \hat\beta' S \hat\beta -
                               \betatrue'S\betatrue\\
                              &= (\hat\beta +\betatrue)'S(\hat\beta - \betatrue) \\
                              &= (S^{1/2}\hat\beta
                                +S^{1/2}\betatrue)'S^{1/2}(\hat\beta -
                                \betatrue),
  \end{align*}
where $S^{1/2}$ denotes the matrix square root of $S$. Noting
\ref{A-l2Sfinite}'s implication that $|S^{1/2}|_{2}=O(1)$,
  \begin{align*}
    |s(\mathbf{x}\hat\beta) -
  s(\mathbf{x}\betatrue)| =& \bkt[2]{s(\mathbf{x}\hat\beta) +
  s(\mathbf{x}\betatrue)}^{-1}|(S^{1/2}\hat\beta
                                +S^{1/2}\betatrue)'S^{1/2}(\hat\beta -
                                \betatrue)|\\
                           &\leq \bkt[1]{|S^{1/2}\hat\beta|_{2}
                                +|S^{1/2}\betatrue|_{2}}^{-1}
                             |S^{1/2}\hat\beta
                                +S^{1/2}\betatrue|_{2}|S^{1/2}|_{2}|\hat\beta
                             - \betatrue|_{2} \\
    &\leq |S^{1/2}|_{2}|\hat\beta
                             - \betatrue|_{2} = O(1)O_{P}(|\hat\beta
                             - \betatrue|_{2}).
\end{align*}
\end{proof}

In the proof of Lemma~\ref{lem:eq8} below, 
let $(I, J) \subseteq \{1, \ldots, n\}$ be a 
randomly ordered simple random sample of size 2, and let $\vec{D}$ (or $\vec{D}^{\perp w}$, $w$ an $n$-vector) be a  $1 \times p$ random vector representing the paired difference $\vec{x}_{I} - \vec{x}_{J}$  (or $\vec{x}_{I}^{\perp w} - \vec{x}_{J}^{\perp w}$).  Then
\begin{equation}
  \label{eq:15}
\stderr_{r}^{2}(I,J) = \vec{D}^{\perp \mathbf{x}\hat\beta}\hat{C} \vec{D}^{\perp \mathbf{x}\hat\beta \prime}
\end{equation}

Because of the $U$-statistic representation of covariance, 
$\mathrm{Cov} (\vec{D}) = 2 S^{(x)}$. By symmetry of the distribution
of $(I,J)$,
$\E{\vec{D}} 
=0$, so that $\E{\vec{D}'\vec{D}} = \mathrm{Cov}(D) = 2S^{(x)}$.  
Selection of $(I, J)$ pays no attention to the distinction between treatment and control, making $\vec{D}$ independent of $\{Z_{i}\}_{i=1}^{n}$ and, by extension, of $\hat{\beta}$ and $\hat{C}$.  Therefore its conditional and marginal moments coincide: $\E{\vec{D}| \hat\beta, \hat{C} } = \E{\vec{D}} = 0$; $ \E{\vec{D}'\vec{D}| \hat\beta, \hat{C} } = \mathrm{Cov}(\vec{D}| \hat\beta, \hat{C}) = \mathrm{Cov}(\vec{D} ) = 2S^{(x)}$.   

\begin{lemma} \label{lem:eq8}
Let $S=S^{(x)}$ or $S^{\perp v} = (n-L)^{-1}\mathbf{x}^{\perp v\prime}\mathbf{x}^{\perp v}$, some categorical variable $v$ with
$L$ categories.  Then  $\fip{S^{\perp\hat\beta}}{\hat{C}} =
  \fip{S}{ \hat{C}} - 
{s^{-2}(\mathbf{x}\hat\beta)}{\fip{S\hat{\beta}\hat{\beta}'S}{
    \hat{C}}}$ and $\fip{S^{\perp\betatrue}}{{C}} =
  \fip{S}{ \hat{C}} - 
{s^{-2}(\mathbf{x}\betatrue)}{\fip{S{\betatrue}{\betatrue}'S}{
    {C}}}$.
\end{lemma}

\begin{proof}
If $S=S^{(x)}$, let $\vec{D}$ be as defined above.  Otherwise, if $S=S^{\perp v}$, then
let $(I, J) \subseteq \{1, \ldots, n\}$ be a 
randomly ordered stratified simple random sample of size 2 from one of the
categories of $v$ with at least two elements, after selecting one of
these categories with probability proportional to size, and let $\vec{D}$ (or $\vec{D}^{\perp w}$, $w$ an $n$-vector) be a  $1 \times p$ random vector representing the paired difference $\vec{x}_{I} - \vec{x}_{J}$  (or $\vec{x}_{I}^{\perp w} - \vec{x}_{J}^{\perp w}$).

Evaluate $\E{\vec{D}^{\perp \mathbf{x}\hat{\beta}\prime}\vec{D}^{\perp
    \mathbf{x}\hat{\beta}} | \hat\beta, \hat{C}} =
\mathrm{Cov}(\vec{D}^{\perp \mathbf{x}\hat{\beta}} | \hat\beta,
\hat{C})$ using the U-statistic representation of sample covariance to
get $\mathrm{Cov}(\vec{D}^{\perp \mathbf{x}\hat{\beta}} | \hat\beta,
\hat{C}) = 2(n-1)^{-1}\mathbf{x}^{\perp \mathbf{x}\hat\beta
  \prime}\mathbf{x}^{\perp \mathbf{x}\hat\beta}$.   Now compare to:
\begin{align}
\frac{1}{2}\E{\vec{D}^{\perp \mathbf{x}\hat{\beta}\prime}\vec{D}^{\perp \mathbf{x}\hat{\beta}} | \hat\beta, \hat{C}}
=&  ({I} - \hat\beta\hat\beta' S^{(x)}/s^{2}(\mathbf{x}\hat{\beta}))'S^{(x)} ({I} - \hat\beta\hat\beta' S^{(x)}/s^{2}(\mathbf{x}\hat{\beta})) \nonumber \\
=&  S^{(x)} - 2s^{-2}(\mathbf{x}\hat{\beta})S^{(x)}\hat\beta\hat\beta'S^{(x)} + \nonumber \\
 & s^{-4}(\mathbf{x}\hat{\beta}) S^{(x)}\hat{\beta}\cdot
   {\hat{\beta}'S^{(x)}\hat{\beta}}\cdot \hat{\beta}' S^{(x)}\nonumber
  \\ \label{eq:18}
=&  S^{(x)} - s^{-2}(\mathbf{x}\hat{\beta})S^{(x)}\hat\beta\hat\beta'S^{(x)} \eqdef S^{\perp \hat{\beta}}. \nonumber
\end{align} 
\end{proof}

\begin{proof}[Proof of Prop.~\ref{prop:picse-consistency}]
It follows from~\ref{A-l2Sfinite} that
$\|S\|_{F}=\operatorname{tr}(S'S)^{1/2}=O(p^{1/2})$, and from the assumed
consistency of $\hat{C}_{n}$ that $\|\hat{C}_{n} -
C_{n}\|_{F}=o_{P}(p^{1/2}/n)$.  So
\begin{equation*}
  |\fip{S}{\hat{C}_{n}}-\fip{S}{C_{n}}| = |\fip{S}{\hat{C}_{n}-C_{n}}| = o_{P}(p/n).
\end{equation*}
By Lemma~\ref{lem:eq8}, 
  \begin{align}
     \fip{S^{\perp\hat\beta}}{\hat{C}_{n}} -
     \fip{S^{\perp\betatrue}}{C_{n}}  =& \fip{S}{\hat{C}_{n} - C_{n}}
                                          \nonumber\\
    &+ 
     s^{-2}(\mathbf{x}\hat\beta)\fip{S\hat{\beta}\hat{\beta}'S}{\hat{C}_{n}}  -
     s^{-2}(\mathbf{x}\betatrue)\fip{S{\betatrue}{\betatrue}'S}{C_{n}}
    \nonumber \\
                                         =& o_{P}\bbkt[1]{\frac{p}{n}} +
                                            \bkt[2]{s^{-2}(\mathbf{x}\hat\beta)
                                            -
                                            s^{-2}(\mathbf{x}\betatrue)}\fip{S\hat\beta\hat{\beta}'S}{\hat{C}_{n}}
                                             \label{eq:37}\\
   & + s^{-2}(\mathbf{x}\betatrue)\bbkt[1]{
    \fip{S\hat\beta\hat{\beta}'S}{\hat{C}_{n}} - \fip{S\betatrue{\betatrue}'S}{{C}_{n}} }.\label{eq:38}
  \end{align}

To analyze the rightmost summand of \eqref{eq:37}, first observe that
$s^{2}(\mathbf{x}\betatrue) =
{\betatrue}'S\betatrue$, so that $|s(\mathbf{x}\betatrue)| \leq |S|_{2}|\betatrue|_{2}=O(p^{1/2})$ . Together with
Proposition~\ref{lem:betahat-consist},
Lemma~\ref{lem:is-sds-same-limit} entails that $s(\mathbf{x}\hat\beta)
- s(\mathbf{x}\betatrue) = O_{P}(p/n)$.  So \ref{A-consistencyrates} says that
$s(\mathbf{x}\hat\beta)=O_{P}(p^{1/2})$ just as
$|s(\mathbf{x}\betatrue)|=O_{P}(p^{1/2})$, whence
$s(\mathbf{x}\hat\beta) + s(\mathbf{x}\betatrue) =
O_{P}(p^{1/2})$. Combining these facts, $s^{2}(\mathbf{x}\hat\beta)
- s^{2}(\mathbf{x}\betatrue) = \bkt[2]{s(\mathbf{x}\hat\beta)
- s(\mathbf{x}\betatrue)}\bkt[2]{s(\mathbf{x}\hat\beta)
+ s(\mathbf{x}\betatrue)} = O_{P}(p/n^{1/2})$. 
In light of \ref{A-PSvar}, it follows that
$|s^{-2}(\mathbf{x}\hat\beta) -
s^{-2}(\mathbf{x}\betatrue)|=O_{P}\bkt[3]{{p}/\bkt[2]{n^{1/2}s^{4}(\mathbf{x}\betatrue)}}$.

As to the $\fip{S\hat\beta\hat{\beta}'S}{\hat{C}_{n}}$ factor of \eqref{eq:37},
\begin{align*}
  |\fip{S\hat\beta\hat{\beta}'S}{\hat{C}_{n}} |&=
                                               \operatorname{tr}(S\hat\beta\hat{\beta}'S\hat{C}_{n})\\
                                             &= |\operatorname{tr}(\hat{\beta}'S\hat{C}_{n}S\hat\beta)|\\
  &= \hat{\beta}'S\hat{C}_{n}S\hat\beta \leq
    \|\hat{\beta}'S^{1/2}\|_{2}\|S^{1/2}\|_{2}\|\hat{C}_{n}\|_{2}\|S^{1/2}\|_{2}\|S^{1/2}\hat\beta\|_{2}\\
  &= O_{P}(\frac{s^{2}(\mathbf{x}\betatrue)}{n}),
\end{align*}
by \ref{A-PSvar}, \ref{A-l2Sfinite}, consistency of $\hat{C}_{n}$ and
Lemma~\ref{lem:C-rate}.  We now have
\begin{equation*}
|\bkt[2]{s^{-2}(\mathbf{x}\hat\beta)
                                            -
                                            s^{-2}(\mathbf{x}\betatrue)}\fip{S\hat\beta\hat{\beta}'S}{\hat{C}_{n}}|
                                         =
                                         O_{P}\bbkt[1]{\frac{p}{n^{3/2}s^{2}(\mathbf{x}\betatrue)}}= 
                                         o_{P}\bbkt[1]{\frac{p}{n}}.  
\end{equation*}

To bound \eqref{eq:38},  in light of \ref{A-PSvar} we focus on the
second factor:
\begin{multline}
  |\fip{S\hat\beta\hat{\beta}'S}{\hat{C}_{n}} -
  \fip{S\betatrue{\betatrue}'S}{{C}_{n}}| \leq \\
  |\fip{S\hat\beta\hat{\beta}'S -
    S\betatrue{\betatrue}'S}{\hat{C}_{n}}| 
  + |\fip{S\betatrue{\betatrue}'S}{\hat{C}_{n} - C_{n}}|.\label{eq:8}
\end{multline}

We address the left term of \eqref{eq:8} via Cauchy-Schwartz. By 
Lemma~\ref{lem:C-rate} and consistency of $\hat{C}_{n}$,
$\|\hat{C}_{n}\|_{2} = O_{P}(n^{-1})$, so $\|\hat{C}_{n}\|_{F} =
O_{P}(p^{1/2}/n)$.  As $S\hat\beta\hat{\beta}'S -
    S\betatrue{\betatrue}'S = S(\hat\beta\hat{\beta}' -
    \betatrue{\betatrue}')S$, $\|S\hat\beta\hat{\beta}'S -
    S\betatrue{\betatrue}'S\|_{F} \leq \|S^{1/2}\|_{2} \|S^{1/2}\hat\beta\hat{\beta}' S^{1/2} -
    S^{1/2}\betatrue{\betatrue}' S^{1/2}\|_{F}\|S^{1/2}\|_{2}= O_{P}(1)O_{P}\bbkt[2]{\|S^{1/2}\hat\beta\hat{\beta}' S^{1/2} -
      S^{1/2}\betatrue{\betatrue}' S^{1/2} \|_{F}}O_{P}(1)$. This decomposes as 
    \begin{multline*}
      \|S^{1/2}\hat\beta\hat{\beta}' S^{1/2} -
      S^{1/2}\betatrue{\betatrue}' S^{1/2} \|_{F} \leq\|S^{1/2}(\hat\beta -
                                     \betatrue)(\hat\beta +
                                     \betatrue)' S^{1/2}\|_{F} +\\
                                     \|S^{1/2}\hat\beta\betatrue' S^{1/2} - S^{1/2}\betatrue\hat\beta' S^{1/2}\|_{F}
    \end{multline*}
where:
\begin{align*}
  \|S^{1/2} (\hat\beta - \betatrue)(\hat\beta +
  \betatrue)' S^{1/2}\|_{F} \leq& \|S^{1/2}\|_{2}\|(\hat\beta - \betatrue)(\hat\beta +
  \betatrue)' S^{1/2}\|_{F}\\
  &=O_{P}(1)
                      \operatorname{tr}\bkt[2]{S^{1/2}(\hat\beta +
  \betatrue) (\hat\beta - \betatrue)' (\hat\beta - \betatrue)(\hat\beta +
                      \betatrue)' S^{1/2}}^{1/2} \\
  &=O_{P}(1)
                      \operatorname{tr}\bkt[2]{(\hat\beta +
                      \betatrue)'S (\hat\beta +
    \betatrue) (\hat\beta - \betatrue)' (\hat\beta - \betatrue)}^{1/2}\\
  &=O_{P}(1)\bkt[2]{(\hat\beta + \betatrue)'S (\hat\beta + \betatrue)
    \|\hat\beta - \betatrue\|_{2}^{2}}^{1/2}\\
    &=O_{P}(1)\bkt[2]{(\hat\beta + \betatrue)'S (\hat\beta +
      \betatrue)}^{1/2}O_{P}\bkt[2]{\bkt[1]{\frac{p}{n}}^{1/2}} \\
                                &= O_{P}\bkt[2]{\bkt[1]{\frac{p}{n}}^{1/2}}
                                  \bkt[2]{(\hat\beta - \betatrue + 2\betatrue)'S (\hat\beta 
                                  - \betatrue + 2\betatrue)}^{1/2}\\
         &= O_{P}\bkt[2]{\bkt[1]{\frac{p}{n}}^{1/2}}
                                  \bkt[2]{\|S^{1/2}(\hat\beta-\betatrue)\|_{2}^{2}
           +4(\hat\beta - \betatrue)'S\betatrue +
           4\|S^{1/2}\betatrue\|_{2}^{2}}^{1/2} \\
   &= O_{P}\bkt[2]{\bkt[1]{\frac{p}{n}}^{1/2}}
                                  \bbkt[3]{O_{P}\bbkt[1]{\frac{p}{n}}
     +O_{P}\bbkt[2]{\bbkt[1]{\frac{p}{n}}^{1/2}
     s(\mathbf{x}\betatrue)} + O_{P}\bkt[2]{s^{2}(\mathbf{x}\betatrue)}}^{1/2}\\
                                &=
                                  O_{P}\bbkt[2]{\bbkt[1]{\frac{p}{n}}^{1/2}}
                                  O_{P}\bbkt[3]{\bbkt[2]{\bbkt[1]{\frac{p}{n}}^{1/2}
                                  + s(\mathbf{x}\betatrue)}^{2 \cdot 1/2}}
                                  \\
                                &=
                                  O_{P}\bbkt[3]{\max\bbkt[2]{\frac{p}{n},
                                  \bbkt[1]{\frac{p}{n}}^{1/2}s(\mathbf{x}\betatrue)}}; 
\end{align*}
and
\begin{align*}
  \|S^{1/2}\hat\beta\betatrue' S^{1/2} - S^{1/2}\betatrue\hat\beta'
  S^{1/2}\|_{F}=&
\|S^{1/2}(\hat\beta-\betatrue)\betatrue' S^{1/2} -
                 S^{1/2}\betatrue(\hat\beta -\betatrue)'
                 S^{1/2}\|_{F}                 \\
  \leq& \|S^{1/2}(\hat\beta-\betatrue)\betatrue' S^{1/2}\|_{F}           +
                 \|S^{1/2}\betatrue(\hat\beta -\betatrue)'
    S^{1/2}\|_{F}           \\
  &= 2\operatorname{tr}\bkt[2]{S^{1/2}\betatrue(\hat\beta -\betatrue)'
    S(\hat\beta-\betatrue)\betatrue' S^{1/2}}^{1/2}\\
  &=2\bkt[2]{(\hat\beta -\betatrue)'
    S(\hat\beta-\betatrue) \betatrue' S \betatrue}^{1/2}\\
  &=O_{P}\bbkt[2]{\bbkt[1]{\frac{p}{n}}^{1/2}s(\mathbf{x}\betatrue)}.
\end{align*}
Thus the left term at right of \eqref{eq:8} is
$O_{P}\bkt[3]{\max\bkt[2]{{p}/{n},
                                  \bkt[1]{{p}/{n}}^{1/2}s(\mathbf{x}\betatrue)}}$.

The remaining term in \eqref{eq:8} is bounded as follows, using the
definition of $\fip{\cdot}{\cdot}$ and the cyclic property of the trace:
\begin{align*}
  \fip{S\betatrue{\betatrue}'S}{\hat{C}_{n} - C_{n}} =&
\operatorname{tr}\bkt[2]{S\betatrue\betatrue'S
                                                       (\hat{C}_{n} -
                                                       C_{n})}\\
  =& \operatorname{tr}\bkt[2]{\betatrue'S
                                                       (\hat{C}_{n} -
                                                       C_{n})
    S\betatrue}\\
                                                     =&
                                                       \betatrue'S
                                                       (\hat{C}_{n} -
                                                       C_{n}) S\betatrue\\
  \leq& |\betatrue|_{2}^{2}\|S\|_{2}^{2}\|\hat{C}_{n} -
    C_{n}\|_{2}\\
  &= O_{P}(p)O_{P}(1)o_{P}(n^{-1}) = o_{P}(p/n),
\end{align*}
using \ref{A-regPS}, \ref{A-l2Sfinite} and consistency of
$\hat{C}_{n}$.  This shows that \eqref{eq:8} as a whole is $O_{P}\bkt[3]{\max\bkt[2]{{p}/{n},
                                  \bkt[1]{{p}/{n}}^{1/2}s(\mathbf{x}\betatrue)}}$,
                              from which it follows that \eqref{eq:38}
is $O_{P}\bkt[3]{\max\bkt[2]{\bkt[1]{p/n} s^{-2}(\mathbf{x}\betatrue),
                                  \bkt[1]{{p}/{n}}^{1/2}s^{-1}(\mathbf{x}\betatrue)}}
                              = o_{P}(p/n)$, completing the proof of the proposition's
consistency claim.  The remainder of the proposition follows from
\ref{A-l2Sfinite} and Lemma~\ref{lem:C-rate}.
\end{proof}

\section{Proofs for section~\ref{sec:errors-paired-index}}

\subsection{Proof for Section~\ref{sec:benchm-index-s.e.s}}

\begin{proof}[Proof of Proposition~\ref{prop:stderr-consist}]
  For parts (i) and (iii) it suffices to show that $\sup_{i}|1 -
    |\vec{x}_{i}\hat{C}^{1/2}|_{2}^{2}/|\vec{x}_{i}{C}^{1/2}|_{2}^{2}|$ and
    $\sup_{i,j}|1 - |(\vec{x}_{i}-\vec{x}_{j})\hat{C}^{1/2}|_{2}^{2}/|(\vec{x}_{i}-\vec{x}_{j}){C}^{1/2}|_{2}^{2}|$ tend in
      probability to 0 (from the definition of in-probability
      convergence), and to consider only $(i, j)$ with
      $\vec{x}_{i}\neq 0$ and $\vec{x}_{i}\neq \vec{x}_{j}$.  The
      latter ensures $|\vec{x}_{i}{C}^{1/2}|_{2}>0$ and
      $|(\vec{x}_{i}-\vec{x}_{j}){C}^{1/2}|_{2}>0$, since $C_{n}$ is of full rank.
      Indeed, Lemma~\ref{lem:fullrankAinv} in the appendix
      gives $|A_{n}|_{2}=O(1)$, covering case (i); for
      case (iii), \ref{A-fullrankB} gives $|B_{n}^{-1}|_{2}=O(1)$,
      and in turn $|A_{n}B_{n}^{-1}A_{n}|_{2} = O(1)$. Note that this
      shows not only that $C_{n}$ has full rank but also that $|C_{n}^{-1}|_{2} =O(1)$.
      
      For arbitrary nonzero $p$-vectors $v$ we have
      \begin{equation*}
        1 - \frac{v'\hat{C}_{n}v}{v'C_{n}v} = \frac{v'(C_{n} - \hat{C}_{n})v}{v'v}\frac{v'v}{v'C_{n}v}
      \end{equation*}
      so that 
      \begin{equation*}
        \left|1 - \frac{v'\hat{C}_{n}v}{v'C_{n}v} \right|   \leq \left|\frac{v'(C_{n} -
              \hat{C}_{n})v}{v'v}\frac{v'C_{n}^{-1}v}{v'v}\right|  
        \leq |C_{n} - \hat{C}_{n}|_{2}|C_{n}^{-1}|_{2},\nonumber
      \end{equation*}
      using general inequalities for positive
      definite matrices \citep[e.g.,][\S~8.4]{gentle07}.
      Proposition~\ref{lem:ChatC} gives $|\hat{C}_{n} - C_{n}|_{2} =
      o_{P}(n^{-1})$, completing the proof of parts (i) and (iii).

      For part (ii), write $a\vee b \defeq \max(a,b)$ and $a\wedge b
      \defeq \min(a,b)$.  For arbitrary nonzero $p$-vectors $v$ 
      \begin{align}
        \left|\frac{v'C_{n}v \vee \epsilon_{n} - v'\hat{C}_{n}v \vee
        \epsilon_{n}}{v'v} \right| &\leq  \left|\frac{v'C_{n}v  -
                                     v'\hat{C}_{n}v}{v'v} \right|,\,
                                     \text{since} \nonumber \\
        \nonumber 
        \left|{v'C_{n}v \vee \epsilon_{n} - v'\hat{C}_{n}v \vee
        \epsilon_{n}} \right| &\leq  \left|{v'C_{n}v  -
                                     v'\hat{C}_{n}v } \right|.
      \end{align}
     (The latter is true because the expressions on either side of
      the inequality are equal if $v'C_{n}v $ and $v'\hat{C}_{n}v $ belong to
      the same half-interval, $(0, \epsilon_{n}]$ or $[\epsilon_{n},
      \infty)$, whereas if they are separated by $\epsilon_{n}$ then the
      left side expression is smaller.)  So
      \begin{equation}\label{eq:11}
        \left|\frac{v'C_{n}v \vee \epsilon_{n} - v'\hat{C}_{n}v \vee
        \epsilon_{n}}{v'v} \right|\leq
                                 |C_{n}-\hat{C}_{n}|_{2} = o_{P}(n^{-1}).
       \end{equation}
     According to \ref{A-boundedXes}, $\sup_{i} \vec{x}_{i}
     \vec{x}_{i}' =O\bkt[2]{\max(p, \log n)}$.  So
     \begin{equation*}
       \sup_{i}\frac{\vec{x}_{i}\vec{x}_{i}'}{\epsilon_{n}},
       \sup_{\{i,j\} \in \mathcal{E}_{n}}\frac{(\vec{x}_{i}-\vec{x}_{j})
     (\vec{x}_{i} - \vec{x}_{j})'}{\epsilon_{n}} =
   O_{P}\bkt[1]{n}. 
\end{equation*}
Thus
     \begin{equation*}
       \sup_{i}\frac{\vec{x}_{i}\vec{x}_{i}'}{\vec{x}_{i}C_{n}\vec{x}_{i}'\vee
         \epsilon_{n}},        \sup_{\{i,j\} \in \mathcal{E}_{n}}\frac{(\vec{x}_{i}-\vec{x}_{j})
     (\vec{x}_{i} - \vec{x}_{j})'}{(\vec{x}_{i}-\vec{x}_{j})C_{n}
     (\vec{x}_{i} - \vec{x}_{j})' \vee \epsilon_{n}} =
   O_{P}\bkt[1]{n}.
 \end{equation*}
 By \eqref{eq:11},
 \begin{equation*}
      \sup_{i: |\vec{x}_{i}|_{2}\neq 0}
   \frac{\vec{x}_{i}C_{n}\vec{x}_{i}' \vee \epsilon_{n} - \vec{x}_{i}\hat{C}_{n}\vec{x}_{i}' \vee
        \epsilon_{n}}{\vec{x}_{i}\vec{x}_{i}'} 
   \frac{\vec{x}_{i}\vec{x}_{i}'}{\vec{x}_{i}C_{n}\vec{x}_{i}'\vee
   \epsilon_{n}} = o_{P}(1),
 \end{equation*}
and likewise
 \begin{align*}
                   \sup_{\substack{\{i,j\} \in \mathcal{E}_{n}:  \\ \vec{x}_{i}\neq \vec{x}_{j}}}
   \frac{(\vec{x}_{i}-\vec{x}_{j})C_{n}(\vec{x}_{i}-\vec{x}_{j})' \vee \epsilon_{n} - (\vec{x}_{i}-\vec{x}_{j})\hat{C}_{n}(\vec{x}_{i}-\vec{x}_{j})' \vee
        \epsilon_{n}}{(\vec{x}_{i}-\vec{x}_{j})
   (\vec{x}_{i}-\vec{x}_{j})'} & \times\\
\frac{(\vec{x}_{i}-\vec{x}_{j})
     (\vec{x}_{i} - \vec{x}_{j})'}{(\vec{x}_{i}-\vec{x}_{j})C_{n}
     (\vec{x}_{i} - \vec{x}_{j})' \vee \epsilon_{n}} &= o_{P}(1).
 \end{align*}
\end{proof}

\begin{proof}[Proof of Proposition~\ref{prop:chaosbound}]
  The assumption on $\{\vec{X}_{i}:i\leq n\}$ entails that $(\vec{X}_{i} -\vec{X}_{j})$ is MVN(0, $2\Sigma$) and
$(\vec{X}_{i} -\vec{X}_{j})C^{1/2}$ is MVN(0,
$2 C^{1/2}\Sigma C^{1/2}$).  Let $2
C^{1/2}\Sigma C^{1/2}$ have
eigendecomposition $Q'\Lambda Q$, with $\Lambda$ nonnegative real
diagonal and $Q$ an orthogonal matrix. Then
$(\vec{X}_{i}-\vec{X}_{j})C^{1/2}Q' \sim \mathrm{MVN}(0,
\Lambda)$.
Writing $W_{ij1}, \ldots, W_{ijp}$ for the $p$
coordinates of $(\vec{X}_{i}
-\vec{X}_{j})C^{1/2}Q' \eqdef \vec{W}_{ij}$,
we see that
$W_{ij1}, \ldots, W_{ijp}$ are mutually independent mean-0 Gaussians
with variances $v_{1}, \ldots, v_{p}$, the diagonal entries of
$\Lambda$ and eigenvalues of $2
C^{1/2}\Sigma C^{1/2}$, while
$(\vec{X}_{i} -\vec{X}_{j})C(\vec{X}_{i} -\vec{X}_{j})'
= (\vec{X}_{i} -\vec{X}_{j})C^{1/2}Q'Q C^{1/2}(\vec{X}_{i} -\vec{X}_{j})'=
\vec{W}_{ij}\vec{W}_{ij}'$.

Straightforwardly, for any $i\neq j$ $\EE \vec{W}_{ij}\vec{W}_{ij}' = \sum_{i}
v_{i}$, or $\operatorname{tr}(2C^{1/2}\Sigma C^{1/2})=2\fip{C}{\Sigma}$.  We
proceed to characterize the moment generating function (MGF) of
$\vec{W}_{ij}\vec{W}_{ij}' - \EE
\vec{W}_{ij}\vec{W}_{ij}'$, or $\sum_{k=1}^{p}(W_{ijk}^{2} - \EE
W_{ijk}^{2})$. The centered
$\chi_{1}^{2}$ distribution having MGF
$e^{-t}(1-2t)^{-1/2}$, valid for $t< 1/2$, $W_{ijk}^{2} - \EE W_{ijk}^{2}$ has log MGF
$(1/2)\bkt[2]{-2v_{k}t - \log(1-2v_{k}t)}$, valid for $t< 1/(2v_{k})$.
Applying the relation
$-\log(1-x) - x \leq x^{2}/\bkt[2]{2(1-x)}$,  $0< x < 1$ 
(a transformation of $\log(1-x)$'s degree 2 Taylor expansion), we now
have
\begin{eqnarray*}
  \log \EE\bkt[3]{ \exp t\bkt[2]{W_{ijk}^{2}-\EE W_{ijk}^{2}}} \leq &  \frac{(v_{k}t)^{2}}{1-2v_{k}t} &
                                                                      \big(\text{for }
                                                                      t<
                                                                      \frac{1}{2v_{k}}\big)
  \\
  &\leq \frac{(v_{k}t)^{2}}{1-2(\max_{k}v_{k})t} & \big(t < \frac{1}{2\max_{k}v_{k}}\big); \text{ so}\\
   \log\EE \exp\bkt[3]{t\sum_{k=1}^{p}\bkt[2]{W_{ijk}^{2}-\EE W_{ijk}^{2})}} & \leq
                                              \frac{t^{2}\sum_{k}v_{k}^{2}}{1-2(\max_{k}v_{k})t},
                                                                     & t < \frac{1}{2\max_{k}v_{k}}
\end{eqnarray*}
due to independence of uncorrelated Gaussians. That is, $\sum_{k=1}^{p}W_{ijk}^{2}-\EE W_{ijk}^{2}$ is sub-gamma on its right
tail with variance factor $\sum_{k}v_{k}^{2}$ and scale factor $\max_k
v_k$; in symbols, $\sum_{k=1}^{p}W_{ijk}^{2}- \EE W_{ijk}^{2}\in
\Gamma_{+}(\sum_{k}v_{k}^{2}, \max_{k}v_{k})$ \citep[\S~2.4]{boucheronLugosiMassart13}.  Since
$\sum_{k}v_{k}^{2}=\operatorname{tr}(\Lambda^{2})$, the cyclic property of
the trace combines with definitions above to reduce this variance factor
to $4|C^{1/2}\Sigma C^{1/2}|_{F}^{2}$, as
$\operatorname{tr}(Q' \Lambda QQ'\Lambda Q)=
4\operatorname{tr}(C^{1/2}\Sigma C^{1/2}C^{1/2}\Sigma C^{1/2})$;
and the scale factor reduces to
$2|C^{1/2}\Sigma C^{1/2}|_{2}$,
as $\max_{k}v_{k}=|\Lambda|_{2}$ and $|\Lambda|_{2}= |Q'\Lambda Q|_{2}$. 

For an MGF characterization in terms of more familiar quantities, note that $|2C^{1/2}\Sigma C^{1/2}|_{F}^{2}=
\sum_{i=1}^{p}v_{i}^{2}\leq (\max_{k}v_{k})(\sum_{k}v_{k})=
2
|C^{1/2}\Sigma C^{1/2}|_{2}(\sum_{k}v_{k})$,
while $\sum_{k}v_{k} = \operatorname{tr}(D^{2})=\operatorname{tr}(2C^{1/2}\Sigma C^{1/2})=2\fip{\Sigma}{C}$.
So the variance factor can be taken as $4
  |C^{1/2}\Sigma C^{1/2}|_{2}|S^{1/2}C^{1/2}|_{F}^{2}$,
  or $4 |C^{1/2}\Sigma C^{1/2}|_{2}\fip{\Sigma}{C}$.

These MGF characterizations give control of the supremum of
$\{\vec{W}_{ij}\vec{W}_{ij}' - \EE \vec{W}_{ij}\vec{W}_{ij}': {i \neq j\leq
  n}\}$.  This class containing ${n \choose 2}$ distinct
$\Gamma_{+}\big[
4 |C^{1/2}\Sigma C^{1/2}|_{2}
\fip{\Sigma}{C}$,
$2 |C^{1/2}\Sigma C^{1/2}|_{2}\big]$
 random variables, Corollary 2.6 of \citet{boucheronLugosiMassart13}
 yields
\begin{multline*}
\EE\max_{ i \neq j \leq n}\bkt[2]{\vec{W}_{ij}\vec{W}_{ij}'
  -\EE(\vec{W}_{ij}\vec{W}_{ij}')}\leq \\
  2\bbkt[2]{2 |C^{1/2}\Sigma C^{1/2}|_{2}\fip{\Sigma}{C} \log {n \choose 2}}^{1/2} +   2 |C^{1/2}\Sigma C^{1/2}|_{2} \log {n \choose 2} .
\end{multline*}
Simplifying via $\EE \vec{W}_{ij}\vec{W}_{ij}'=2\fip{\Sigma}{C}$,  $i \neq
j$, and $2a^{2}+2ab+b^{2}=a^{2}\bkt[2]{1+ (1+b/a)^{2}}$,
\begin{align}
  \EE\max_{ i \neq j \leq n}\vec{W}_{ij}\vec{W}_{ij}'
  \leq &
  \fip{\Sigma}{C}\bbkt[3]{1 + \bbkt[2]{1 +
    \bbkt[1]{\frac{2
    |C^{1/2}\Sigma C^{1/2}|_{2} \log
    {n\choose 2}}{\fip{\Sigma}{C}}}^{1/2}}^{2}
                                                        } \nonumber \\
       &=
         \fip{\Sigma}{C}\bbkt[3]{1 + \bbkt[2]{1 +
                                                             \bbkt[1]{\frac{2\log
         {n\choose 2}}{p_{[C^{1/2}\Sigma C^{1/2}]}}}^{1/2}}^{2}}
                                                             . \label{eq:chaos-complete-sq}
\end{align}
Observe next that  on the positive real line $x\mapsto \sqrt{2} + x$ dominates  $x\mapsto \bkt[2]{1 + (1+x)^{2}}^{1/2}$. (The
functions coincide at $x=0$; otherwise the latter has derivative equal
to the square root of
$(1+x)^{2}/\bkt[2]{1 + (1+x)^{2}}$, which is
nowhere greater than 1.)   Thus
\eqref{eq:chaos-complete-sq} gives
\begin{align*}
    \bbkt[1]{\EE\max_{ i \neq j \leq
  n}\vec{W}_{ij}\vec{W}_{ij}'}^{1/2} \leq
  & \fip{\Sigma}{C}^{1/2}
    \bbkt[2]{\sqrt{2} + \bbkt[1]{%
    \frac{2\log
    {n\choose 2}
    }{p_{[C^{1/2}\Sigma C^{1/2}]}}
    }^{1/2}} \\
  &=\sqrt{2}\fip{\Sigma}{C}^{1/2}
\bbkt[2]{1 + \bbkt[1]{ \frac{\log
         {n\choose
    2}}{p_{[C^{1/2}\Sigma C^{1/2}]}}
    }^{1/2}}.
\end{align*}
\end{proof}
\section{Section~\ref{sec:cons-match-estim}}
\begin{proof}[Proof of Prop.~\ref{prop:Qconsist}]
  Throughout the proof, expected value is interpreted to be conditional
on $\mathcal{F}_{n}$. (Because each
$\mathcal{F}_{n}$ is the sigma field of finitely many discrete random
variables, this introduces no measure-theoretic considerations that
were not already present.)  Also assume version (ii) of the moment condition, noting
that it is entailed by version (i) and $0< \PP(Z=1) <1$ (which follows
from $0< \PP(Z=1 \mid \vec{X}\betatrue) <1$): $\EE|Y_{T}|^{p} <
\infty$ means $\EE\bkt[1]{|Y_{T}|^{p}\mid Z} <\infty$ also, which entails
$\EE|\Vns|^{p}< \infty$.  If
$\card[0]{\mathbf{s}}$ is bounded, then there are finitely many
$(\sum_{i \in \mathbf{s}}\indicator{Z_{i}=0}, \sum_{i \in
  \mathbf{s}}\indicator{Z_{i}=1})$ achievable configurations for fine
strata $\mathbf{s}$, each with a characteristic distribution
function $x \mapsto \PP(\Vns \leq x)$ such that $\E{|\Vns|^{p}} =
p\int_{0}^{\infty}y^{p-1}\PP(|\Vns| > y) dy < \infty$
\citep[e.g.,][Lemma 2.2.13]{durrett19}. The distribution for $|V|$ on $\Re^{+}$ given
by setting $\PP(|V|>y)$ to the maximum over these distributions of
$\PP(\Vns > y)$ stochastically dominates the relevant $\Vns$
distributions.  It also satisfies $p\int_{0}^{\infty}y^{p-1}\PP(|V| >
y) dy < \infty$, because the sum of a finite collection of functions
$y \mapsto \PP(|\Vns| > y)$ dominates their maximum; so $V \in L_{p}$. 

\paragraph{Part~\ref{prop:Qconsist:item1} of proposition.} From~\eqref{eq:44},
  \begin{align}
    \sum_{i \in \mathbf{s}}\psi_{\mathbf{s}i}(\eta) &= w_{\mathbf{s}}
\bkt[2]{Y_{i} - \eta \bkt[1]{Z_{i} -
    \bar{Z}_{\mathbf{s}}}}
                                                (Z_{i} -\bar{Z}_{\mathbf{s}})\nonumber\\
   &= w_{\mathbf{s}}\bkt[3]{\sum_{\substack{i\in \mathbf{s}:\\
    Z_{i}=1}}\bkt[2]{Y_{i}(1-\bar{Z}_{\mathbf{s}})
    -\eta(1-\bar{Z}_{\mathbf{s}})^{2}}  +
    \sum_{\substack{i\in \mathbf{s}:\\
    Z_{i}=0}}\bkt[2]{-Y_{i}\bar{Z}_{\mathbf{s}}
    -\eta\bar{Z}_{\mathbf{s}}^{2}}} \nonumber\\
                                              &= w_{\mathbf{s}}\bkt[2]{(1-\bar{Z}_{\mathbf{s}})
                                                \sum_{\substack{i\in \mathbf{s}:\\
    Z_{i}=1}}Y_{i}
    -\eta\card[0]{\mathbf{s}}\bar{Z}_{\mathbf{s}}(1-\bar{Z}_{\mathbf{s}})^{2} 
    -\bar{Z}_{\mathbf{s}}\sum_{\substack{i\in \mathbf{s}:\\
    Z_{i}=0}}Y_{i}
    -\eta\card[0]{\mathbf{s}}(1-\bar{Z}_{\mathbf{s}})\bar{Z}_{\mathbf{s}}^{2}}
\nonumber\\
                                               &= w_{\mathbf{s}} \card[0]{\mathbf{s}}\bar{Z}_{\mathbf{s}}(1-\bar{Z}_{\mathbf{s}})
          \bbkt[2]{\bar{Y}_{\mathbf{s}(1)} -\eta(1-
      \bar{Z}_{\mathbf{s}}) - \bar{Y}_{\mathbf{s}(0)}-
      \eta\bar{Z}_{\mathbf{s}}}\nonumber \\
    &= w_{\mathbf{s}} \card[0]{\mathbf{s}}\bar{Z}_{\mathbf{s}}(1-\bar{Z}_{\mathbf{s}})
          \bkt[1]{\bar{Y}_{\mathbf{s}(1)} - \bar{Y}_{\mathbf{s}(0)}-
      \eta},\label{eq:54}
  \end{align}
so that $\EE\bkt[2]{\sum_{i \in \mathbf{s}}\psi_{\mathbf{s}i}(\eta) } = w_{\mathbf{s}}
\card[0]{\mathbf{s}}\bar{z}_{\mathbf{s}}(1-\bar{z}_{\mathbf{s}})\EE\bkt[1]{
\bar{Y}_{\mathbf{s}(1)} - \bar{Y}_{\mathbf{s}(0)}-
\eta}$.   Recalling that $\Vns = \bar{Y}_{\mathbf{s}(1)}-\bar{Y}_{\mathbf{s}(0)}-
\EE(\bar{Y}_{\mathbf{s}(1)}-\bar{Y}_{\mathbf{s}(0)})$, we have
\begin{equation}
  \label{eq:53}
  \psi_{\mathcal{S}_{n}}(\eta)- \EE\bkt[2]{\psi_{\mathcal{S}_{n}}(\eta)} = \frac{\sum_{\mathbf{s}
\in \mathcal{S}_{n}}
\tilde{w}_{n,\mathbf{s}}\card[0]{\mathbf{s}}
\Vns}{\sum_{\mathbf{s}\in
  \mathcal{S}_{n}}\tilde{w}_{n, \mathbf{s}} \card[0]{\mathbf{s}}}
\end{equation}
where
$\tilde{w}_{n,\mathbf{s}}\defeq w_{\mathbf{s}}\bar{z}_{\mathbf{s}}(1-\bar{z}_{\mathbf{s})}
\in \mathcal{F}_{n}$. (When there is no risk of ambiguity,
``$\tilde{w}_{n,\mathbf{s}}$'' is abbreviated to
``$\tilde{w}_{\mathbf{s}}$.'') Convergence of $\psi_{\mathcal{S}_{n}}(\eta)- \EE\bkt[2]{\psi_{\mathcal{S}_{n}}(\eta)  |
  \mathcal{G}_{n}}$ will be seen to follow from suitable convergence
of \eqref{eq:53}, i.e. $\psi_{\mathcal{S}_{n}}(\eta)- \EE\bkt[2]{\psi_{\mathcal{S}_{n}}(\eta)  |
  \mathcal{F}_{n}}$. 

\subparagraph{In-probability convergence of $\psi_{\mathcal{S}_{n}}(\eta)- \EE\bkt[2]{\psi_{\mathcal{S}_{n}}(\eta)}$.}\label{subparag:prob-conv-psi_m}
We adapt to the independent non-identically
distributed case an argument for the $L_{1}$-weak law of
large numbers by truncation of increments.  Following Durrett
\citeyearpar[\S~2.2.3]{durrett19},  set $\Vnstrunc{} \defeq
\Vns\indicator{w_{n,\mathbf{s}}|\Vns|\leq m_{n}}$ and
$\bar{V}_{n}\defeq V\indicator{|V|\leq m_{n}}$.
(Recall $m_{n}=\sum_{\mathbf{s}\in\mathcal{S}_{n}}
\indicator{w_{\mathbf{s}}\bar{z}_{\mathbf{s}}(1-\bar{z}_{\mathbf{s}})>0}$.)
Whereas $\EE \Vns =0$ by
definition, $\EE \Vnstrunc{}$ may differ from 0.  Our first task is to
show that $ (\sum_{\mathbf{s} \in
  \mathcal{S}_{n}}\tilde{w}_{n,\mathbf{s}})^{-1}\EE \sum_{\mathbf{s} \in
  \mathcal{S}_{n}}\tilde{w}_{n,\mathbf{s}} \Vnstrunc{} \rightarrow 0$.

We have
\begin{equation*}
  (\sum_{\mathbf{s} \in \mathcal{S}_{n}}\tilde{w}_{\mathbf{s}}\card[0]{\mathbf{s}})^{-1}\EE \sum_{\mathbf{s} \in \mathcal{S}_{n}}\tilde{w}_{\mathbf{s}}\card[0]{\mathbf{s}} \Vnstrunc{}
  =    (\sum_{\mathbf{s} \in \mathcal{S}_{n}}\tilde{w}_{\mathbf{s}}\card[0]{\mathbf{s}})^{-1}\bbkt[1]{\EE \sum_{\mathbf{s} \in
    \mathcal{S}_{n}}\tilde{w}_{\mathbf{s}}\card[0]{\mathbf{s}} \Vnstrunc{+} + \EE \sum_{\mathbf{s} \in
    \mathcal{S}_{n}}\tilde{w}_{\mathbf{s}}\card[0]{\mathbf{s}} \Vnstrunc{-}}
  \label{eq:64}
\end{equation*}
where $a_{+}$ and $a_{-}$ denote positive and negative parts of
$a$, $\max(a, 0)$ and $\min(a,0)$.  Now
\begin{align}
  (\sum_{\mathbf{s} \in \mathcal{S}_{n}}\tilde{w}_{\mathbf{s}}\card[0]{\mathbf{s}})^{-1}
  \EE \sum_{\mathbf{s} \in
    \mathcal{S}_{n}}\tilde{w}_{\mathbf{s}}\card[0]{\mathbf{s}} \Vnstrunc{+} &=
  (\sum_{\mathbf{s} \in \mathcal{S}_{n}}\tilde{w}_{\mathbf{s}}\card[0]{\mathbf{s}})^{-1}
  \sum_{\mathbf{s} \in
    \mathcal{S}_{n}}\tilde{w}_{\mathbf{s}}\card[0]{\mathbf{s}} \EE \Vnstrunc{+} \nonumber\\
&=  (\sum_{\mathbf{s} \in \mathcal{S}_{n}}\tilde{w}_{\mathbf{s}}\card[0]{\mathbf{s}})^{-1}
  \sum_{\mathbf{s} \in
    \mathcal{S}_{n}}\tilde{w}_{\mathbf{s}}\card[0]{\mathbf{s}} \int_{0}^{\infty}\PP(\Vnstrunc{} >
                                                                          x)dx \nonumber\\
&=  \int_{0}^{\infty}\frac{\sum_{\mathbf{s} \in
    \mathcal{S}_{n}}\tilde{w}_{\mathbf{s}}\card[0]{\mathbf{s}} \PP(\Vnstrunc{} >
  x)}{\sum_{\mathbf{s} \in \mathcal{S}_{n}}\tilde{w}_{\mathbf{s}}\card[0]{\mathbf{s}}}dx.   \label{eq:57}
\end{align}
By the stochastic ordering assumption, $\PP(\Vns >x)$  is dominated by
$\PP(V>x)$; but $\PP(\Vnstrunc{} >x) \leq \PP(\Vns >x)$, so the
integrand in \eqref{eq:57} is dominated by $\PP(V>x)$ as well.  As $\int_{0}^{\infty} \PP(V>x) dx = \EE V_{+} < \infty$, dominated
convergence gives
$(\sum_{\mathbf{s} \in \mathcal{S}_{n}}\tilde{w}_{\mathbf{s}}\card[0]{\mathbf{s}})^{-1}
  \EE \sum_{\mathbf{s} \in
    \mathcal{S}_{n}}\tilde{w}_{\mathbf{s}}\card[0]{\mathbf{s}} \Vnstrunc{+} -
   (\sum_{\mathbf{s} \in \mathcal{S}_{n}}\tilde{w}_{\mathbf{s}}\card[0]{\mathbf{s}})^{-1}
  \EE \sum_{\mathbf{s} \in
    \mathcal{S}_{n}}\tilde{w}_{\mathbf{s}}\card[0]{\mathbf{s}} \Vns_{+} = o(1)$.

  Similarly
  $(\sum_{\mathbf{s} \in \mathcal{S}_{n}}\tilde{w}_{\mathbf{s}}\card[0]{\mathbf{s}})^{-1}
  \EE \sum_{\mathbf{s} \in
    \mathcal{S}_{n}}\tilde{w}_{\mathbf{s}}\card[0]{\mathbf{s}} \Vnstrunc{-} -
   (\sum_{\mathbf{s} \in \mathcal{S}_{n}}\tilde{w}_{\mathbf{s}}\card[0]{\mathbf{s}})^{-1}
  \EE \sum_{\mathbf{s} \in
    \mathcal{S}_{n}}\tilde{w}_{\mathbf{s}}\card[0]{\mathbf{s}} \Vns_{-} = o(1)$, and 
  \begin{equation*}
    \bbkt[1]{\sum_{\mathbf{s} \in
      \mathcal{S}_{n}}\tilde{w}_{\mathbf{s}}\card[0]{\mathbf{s}}}^{-1}
    \bbkt[1]{\EE \sum_{\mathbf{s} \in
  \mathcal{S}_{n}}\tilde{w}_{\mathbf{s}}\card[0]{\mathbf{s}} \Vnstrunc{} -
  \EE \sum_{\mathbf{s}
    \in \mathcal{S}_{n}}\tilde{w}_{\mathbf{s}}\card[0]{\mathbf{s}} \Vns} = o(1).
  \end{equation*}
  Since $\EE \Vns =0$, this means
  $(\sum_{\mathbf{s} \in
  \mathcal{S}_{n}}\tilde{w}_{\mathbf{s}}\card[0]{\mathbf{s}})^{-1}\EE \sum_{\mathbf{s} \in
  \mathcal{S}_{n}}\tilde{w}_{\mathbf{s}}\card[0]{\mathbf{s}} \Vnstrunc{} \rightarrow 0$.

For in-probability convergence of \eqref{eq:53} it now suffices
to show
\begin{equation}
  \label{eq:56}
  \frac{\sum_{\mathbf{s}
\in \mathcal{S}_{n}}
\tilde{w}_{\mathbf{s}}\card[0]{\mathbf{s}}
\bkt[1]{\Vns - \EE \Vnstrunc{}}}{\sum_{\mathbf{s}\in
  \mathcal{S}_{n}}\tilde{w}_{\mathbf{s}}\card[0]{\mathbf{s}}} = o_{P}(1), 
\end{equation}
the conclusion of the weak law for triangular arrays.  We now
verify the premises of that principle, as it is given in Durrett's
\citeyearpar{durrett19} Theorem~2.2.11. For each $n$ $\{\Vns: \mathbf{s}
\in \mathcal{S}_{n}\}$ are independent because $\{(Y_{Ti}, Y_{Ci},
Z_{i}): i \}$ are unconditionally independent and conditioning on
$\mathbb{F}_{n}$ induces dependence within but not across strata
$\mathbf{s}$.  Premise (i) of the theorem, $\sum_{\mathbf{s} \in
  \mathcal{S}_{n}}\PP(\tilde{w}_{\mathbf{s}}\card[0]{\mathbf{s}}\Vns > \sum_{\mathbf{s}
\in \mathcal{S}_{n}}
\tilde{w}_{\mathbf{s}}\card[0]{\mathbf{s}}) \rightarrow 0$, will follow, by the assumptions of stochastic
dominance and boundedness of $\tilde{w}_{\mathbf{s}}\card[0]{\mathbf{s}}$, from
convergence to 0 of
\begin{equation}
  \label{eq:55}
  \sum_{\substack{\mathbf{s} \in
  \mathcal{S}_{n}: \tilde{w}_{\mathbf{s}}\card[0]{\mathbf{s}}>0}}\PP\bkt[1]{u_{\tilde{w}}V > \sum_{\mathbf{s}
\in \mathcal{S}_{n}} \tilde{w}_{\mathbf{s}}\card[0]{\mathbf{s}}} = m_{n}\PP\bkt[1]{u_{\tilde{w}}V > \sum_{\mathbf{s}
\in \mathcal{S}_{n}} \tilde{w}_{\mathbf{s}}\card[0]{\mathbf{s}}} , 
\end{equation}
where $u_{\tilde{w}}$
is an upper bound for $\{\tilde{w}_{\mathbf{s}}\card[0]{\mathbf{s}}: n; \mathbf{s} \in
\mathcal{S}_{n}\}$.
By hypothesis $m_{n}/\bkt[1] {\sum_{\mathbf{s}\in
  \mathcal{S}_{n}}\tilde{w}_{\mathbf{s}}\card[0]{\mathbf{s}}} = O_{P}(1)$, so
\eqref{eq:55}=$o_{P}(1)$ will follow from $\bkt[1]{\sum_{\mathbf{s}\in
  \mathcal{S}_{n}}\tilde{w}_{\mathbf{s}}\card[0]{\mathbf{s}}} \PP\bkt[1]{u_{\tilde{w}}V > \sum_{\mathbf{s}
\in \mathcal{S}_{n}} \tilde{w}_{\mathbf{s}}\card[0]{\mathbf{s}}} \rightarrow 0$.  As we
also hypothesize that $m_{n}\rightarrow \infty$ as $n\uparrow \infty$, we also have $\sum_{\mathbf{s}\in
  \mathcal{S}_{n}}\tilde{w}_{\mathbf{s}}\card[0]{\mathbf{s}} \rightarrow \infty$ as $n$
increases, and convergence to 0 of \eqref{eq:55} follows if 
$x\PP( u_{\tilde{w}}V>x) \rightarrow 0$ as $x\uparrow \infty$.  This
is true by dominated convergence, since $x\PP( u_{\tilde{w}}V>x) \leq \EE
\bkt[1]{u_{\tilde{w}}V\indicator{u_{\tilde{w}}V>x}}$,
$u_{\tilde{w}}V\indicator{u_{\tilde{w}}V>x} \rightarrow 0$ a.s. as $x\uparrow\infty$,
and $\E{u_{\tilde{w}} V} < \infty$.

Premise (ii) of Durrett's
\citeyearpar{durrett19} Theorem~2.2.11 is that
\begin{equation}
\label{eq:47}
  \bkt[1]{\sum_{\mathbf{s}\in
  \mathcal{S}_{n}}\tilde{w}_{\mathbf{s}}\card[0]{\mathbf{s}}}^{-2}\sum
\card[0]{\mathbf{s}}^2\tilde{w}_{\mathbf{s}}\card[0]{\mathbf{s}}^{2} \E{(\Vnstrunc{})^{2}} \rightarrow 0.
\end{equation}
To verify this, observe first that
\begin{align}
\E{(w_{\mathbf{s}}\Vnstrunc{})^{2}} =&
                                       2\int_{0}^{\infty}y\PP(|w_{\mathbf{s}}\Vnstrunc{}|>y) dy =
                                       2\int_{0}^{m_{n}}y\PP(|w_{\mathbf{s}}\Vnstrunc{}|>y) dy\nonumber\\
  &\leq 2\int_{0}^{m_{n}}y\PP(|w_{\mathbf{s}}\Vns|>y) dy \label{eq:58}\\
&\leq 2\int_{0}^{m_{n}}y\PP(|u_{\tilde{w}}V|>y) dy,  \label{eq:59}
\end{align}
with \eqref{eq:59} following from \eqref{eq:58} by the stochastic
dominance assumption.  In consequence,
\begin{equation*}
  \label{eq:60}
    \bkt[1]{\sum_{\mathbf{s}\in
  \mathcal{S}_{n}}\tilde{w}_{\mathbf{s}}\card[0]{\mathbf{s}}}^{-2}\sum
\card[0]{\mathbf{s}}^2\tilde{w}_{\mathbf{s}}^{2} \E{(\Vnstrunc{})^{2}}  \leq
\bkt[1]{\sum_{\mathbf{s}\in
  \mathcal{S}_{n}}\tilde{w}_{\mathbf{s}}\card[0]{\mathbf{s}}}^{-2}2m_{n}\int_{0}^{m_{n}}y\PP(|u_{\tilde{w}}V|>y) dy.
\end{equation*}
As we assume $m_{n}/\bkt[1] {\sum_{\mathbf{s}\in
  \mathcal{S}_{n}}\tilde{w}_{\mathbf{s}}\card[0]{\mathbf{s}}} = O_{P}(1)$,
for \eqref{eq:47} it suffices to show
$m_{n}^{-1}\int_{0}^{m_{n}}y\PP(|u_{\tilde{w}}V|>y) \rightarrow 0$ as
$n\uparrow \infty$.  By the hypothesis that $m_{n}\rightarrow \infty$,
this flows from $x^{-1}\int_{0}^{x}y\PP(|u_{\tilde{w}}V|>y)
\rightarrow 0$, which is a consequence of $\EE |u_{\tilde{w}}V| <
\infty$, as shown by \citet{durrett19} in the proofs of Theorems
2.1.12 and 2.1.14.  This completes the verification that \eqref{eq:53}
converges in probability to 0. 

\subparagraph{$L_{1}$ convergence of  $\psi_{\mathcal{S}_{n}}(\eta)- \EE\bkt[2]{\psi_{\mathcal{S}_{n}}(\eta)}$.} By hypothesis there
is $p>1$ such that $\| \Vns\|_{L_{p}} \leq \| V\|_{L_{p}}$ for all $n$
and $\mathbf{s}\in \mathcal{S}_{n}$. Accordingly
$\| \bkt[1]{\sum_{\mathbf{s}\in
    \mathcal{S}_{n}}\tilde{w}_{n, \mathbf{s}}}^{-1}
{\sum_{\mathbf{s} \in \mathcal{S}_{n}}
  \tilde{w}_{n,\mathbf{s}} \Vns}\|_{L_{p}} \leq
\|V\|_{L_{p}}$ also. By the dominated convergence principle for
random variables as in Theorem 1.6.8 and Exercise 2.3.5 of
\citet{durrett19}, therefore, \eqref{eq:53} converges to 0 in $L_{1}$ as well as
in probability.

\subparagraph{$L_{1}$- and in-probability convergence of  $\psi_{\mathcal{S}_{n}}(\eta)-
  \EE\bkt[2]{\psi_{\mathcal{S}_{n}}(\eta) \mid \mathcal{G}_{n}}$.}
Observe that
\begin{align}
  \|\EE\bkt[2]{\psi_{\mathcal{S}_{n}}(\eta)| \mathcal{G}_{n}} -
  \EE\bkt[2]{\psi_{\mathcal{S}_{n}}(\eta)}\|_{L_{1}}
  =&
    \|\EE\bkt[3]{\psi_{\mathcal{S}_{n}}(\eta) -
    \EE\bkt[2]{\psi_{\mathcal{S}_{n}}(\eta)}
    \mid \mathcal{G}_{n}} \|_{L_{1}} \nonumber\\
   &\leq   \|{\psi_{\mathcal{S}_{n}}(\eta) -
    \EE\bkt[2]{\psi_{\mathcal{S}_{n}}(\eta)}
    } \|_{L_{1}}, \nonumber
       \label{eq:61}
\end{align}
$\mathcal{F}_{n}$ being the smaller of the sigma fields
$\{\mathcal{F}_{n}, \mathcal{G}_{n}\}$,
$\EE\bkt[2]{\psi_{\mathcal{S}_{n}}(\eta)}$ being the same as
$\EE\bkt[2]{\psi_{\mathcal{S}_{n}}(\eta)\mid \mathcal{F}_{n}}$ and the conditional expectation operator
being a contraction in $L_{1}$ \citep[Thm.~4.1.11]{durrett19}.  In
tandem with
\begin{equation*}
  \| \psi_{\mathcal{S}_{n}}(\eta)- \EE\bkt[2]{\psi_{\mathcal{S}_{n}}(\eta)  |
  \mathcal{G}_{n}}\|_{L_{1}} \leq
  \| \psi_{\mathcal{S}_{n}}(\eta)- \EE\bkt[2]{\psi_{\mathcal{S}_{n}}(\eta)}\|_{L_{1}} +
   \| \EE\bkt[2]{\psi_{\mathcal{S}_{n}}(\eta)| \mathcal{G}_{n}} -
  \EE\bkt[2]{\psi_{\mathcal{S}_{n}}(\eta)}\|_{L_{1}}
\end{equation*}
this means $L_{1}$ convergence of $\psi_{\mathcal{S}_{n}}(\eta)- \EE\bkt[2]{\psi_{\mathcal{S}_{n}}(\eta)  |
  \mathcal{F}_{n}}$ entails that $\psi_{\mathcal{S}_{n}}(\eta)-
\EE\bkt[2]{\psi_{\mathcal{S}_{n}}(\eta)  | \mathcal{G}_{n}}$ also converges
to zero in $L_{1}$.  Finally, $L_{1}$ convergence entails convergence in
probability. 
\paragraph{Part~\ref{prop:Qconsist:item2} of proposition.}
  Provided that $\sum_{\mathbf{s}} \tilde{w}_{\mathbf{s}}\card[0]{\mathbf{s}}$ is
  positive, both $\eta \mapsto \psi_{\mathcal{S}_{n}}(\eta)$ and $\eta \mapsto \EE\bkt[2]{\psi_{\mathcal{S}_{n}}(\eta) \mid
    \mathcal{G}_{n}}$ are everywhere differentiable with slope $-1$, and can be seen to tend to $\pm
  \infty$ as $\eta$ tends to $\mp \infty$.  It follows that they have
  unique roots.
\paragraph{Part~\ref{prop:Qconsist:item3} of proposition.}
 If the solutions $\tau_{n}$ of $\EE\bkt[2]{\psi_{\mathcal{S}_{n}}(\eta) \mid
    \mathcal{G}_{n}}=0$ tend to a limit $\tau_{0} \in (-\infty, \infty)$,
  then the following adaptation of the
  Huber argument for consistency of scalar solutions of monotone
  estimating equations \citep[Lemma~5.10]{huber1964robust, vdvaart:1998} 
  shows that $\hat{\tau}_{n} \rightarrow
  \tau_{0}$ in probability.  Fix $\epsilon>0$.  Then
  \begin{equation*}
    \label{eq:63}
    \PP\bkt[2]{\psi_{\mathcal{S}_{n}}(\tau_{0} - \epsilon) >
    \epsilon/2, \psi_{\mathcal{S}_{n}}(\tau_{0} + \epsilon) < -\epsilon/2
    } \leq \PP(\tau_{0} - \epsilon < \hat{\tau}_{n} <
    \tau_{0}+\epsilon).
  \end{equation*}
  The left side tends to 1 because
    $\psi_{\mathcal{S}_{n}}(\tau_{0}\pm \epsilon) - \EE\bkt[2]{\psi_{\mathcal{S}_{n}}(\tau_{0}\pm \epsilon) \mid
    \mathcal{G}_{n}}  = o_{P}(1)$,  $\PP( |\tau_{n} - \tau_{0}| <
  \epsilon/2) \rightarrow 1$,
  $\EE\bkt[2]{\psi_{\mathcal{S}_{n}}(\eta) \mid
    \mathcal{G}_{n}} > \epsilon/2$ if $\eta < \tau_{n}-\epsilon/2$ and 
   $\EE\bkt[2]{\psi_{\mathcal{S}_{n}}(\eta) \mid
    \mathcal{G}_{n}} < -\epsilon/2$ if $\eta > \tau_{n}+\epsilon/2$.
  Therefore the right  hand side tends to 1 as well.
\end{proof}

\begin{lemma}
   \label{lem:msPSerr}
   Let $\theta_{i}=\operatorname{logit} \bkt[2]{\PP (Z =1 |
  \mathbf{X}\betatrue =  \mathbf{x}_{i}\betatrue)}$.
  If $|\theta_{i} - \theta_{j}| < \delta$ whenever $i
  \stackrel{\mathcal{S}_{n}}{\sim} j$, then for all $\mathbf{s} \in
  \mathcal{S}_{n}$ and $z: \mathbf{s} \rightarrow \{0,1\}$ such that $\sum_{i \in
  \mathbf{s}}\zeta_{i} = \sum_{i \in \mathbf{s}}z_{i} \in \{ 1, \card[0]{\mathbf{s}}-1\}$, 
  \begin{equation}
    \label{eq:46}
    \left| \frac {\pibs(\zeta)} {\card[0]{\mathbf{s}}^{-1}} -1\right|
    \leq \bkt[1]{1- \card[0]{\mathbf{s}}^{-1}}\bkt[1]{e^{2\delta}-1}
  \end{equation}
  and
  \begin{equation}
    \label{eq:80}
    \left| \frac {\card[0]{\mathbf{s}}^{-1}}{\pibs(\zeta)} -1\right|
    \leq \bkt[1]{1- \card[0]{\mathbf{s}}^{-1}}\bkt[1]{e^{4\delta}-1}.
  \end{equation}
\end{lemma}

\begin{proof}[Proof of Lemma~\ref{lem:msPSerr}]
For $\mathbf{s}$ an $\atob{1}{m}$ matched set, some
nonnegative integer $m$, by \eqref{eq:45} we have
  \begin{align}
    \frac{\pibs(\zeta)}{\card[0]{\mathbf{s}}^{-1}} -1
    =& \frac{1}{\card[0]{\mathbf{s}}^{-1}} \cdot
                                            \frac{\exp{\theta_{i}}}{
                                            \sum_{j \in [i]}\exp({\theta_{j}})
                                  }  -1 \nonumber \\
    =& \frac{\card[0]{\mathbf{s}} - \sum_{j \in [i]}\exp({\theta_{j} - {\theta}_{i}})
}{
       \sum_{j \in [i]} \exp({\theta_{j} - {\theta}_{i}})}\\
     =& \frac{m - \sum_{j \in [i]\setminus\{i\}}\exp({\theta_{j} - {\theta}_{i}})
}{
\sum_{j \in [i]} \exp({\theta_{j} - {\theta}_{i}})} ; \quad
       \mathrm{so} \nonumber \\
    \frac{m - m\exp({\delta})}{(m+1)\exp({-\delta})}   &\leq 
\frac{\pibs(\zeta)}{\card[0]{\mathbf{s}}^{-1}} -1\leq \frac{m -
                                                    m\exp({-\delta})}{(m+1)\exp({-\delta})},\\
-\frac{m}{m+1}e^{\delta}(e^{\delta} -1) &\leq 
\frac{\pibs(\zeta)}{\card[0]{\mathbf{s}}^{-1}} -1\leq
                                \frac{m}{m+1}(e^{\delta} -1)
                                                        \quad
                                                    \mathrm{and}
                                                        \nonumber \\
\left| \frac{\pibs(\zeta)}{\card[0]{\mathbf{s}}^{-1}} -1\right| & \leq  \frac{\card[0]{\mathbf{s}}-1}{\card[0]{\mathbf{s}}}e^{\delta}(e^{\delta} -1) .\label{eq:49}
\end{align}
Now observe that $e^{2\delta} - e^{\delta} < e^{2\delta} -1$ for
positive $\delta$; \eqref{eq:46} follows. 

If $\mathbf{s}$ is an $\atob{m}{1}$ matched set, $m\geq 0$, the argument
culminating in \eqref{eq:49}
again applies after substitution of $-\theta_{i}$ and $-\theta_{j}$ for
$\theta_{i}$ and $\theta_{j}$. Again \eqref{eq:46} follows.

With \eqref{eq:46} established in all cases, \eqref{eq:80} follows
by \eqref{eq:46}'s consequence that
\begin{equation*}
  \left|\frac {\pibs(\zeta)} {\card[0]{\mathbf{s}}^{-1}}\right| \leq e^{2\delta};
\end{equation*}
the identity $|x^{-1}-1|\leq |x^{-1}|\cdot |x-1|$, valid for $x\neq
0$; and $e^{2\delta}(e^{2\delta} -1) \leq e^{4\delta} -1$.
\end{proof}

\begin{proof}[Proof of Prop.~\ref{prop:Pconsist}.] 
  Claim \textbf{(i)} follows from
  \begin{equation}
    \label{eq:74}
    \EE\bkt[2]{\tilde{\psi}_{\mathbf{s}}(\eta) \mid \mathcal{G}_{n}} =
    \tilde{w}_{\mathbf{s}}\sum_{i \in \mathbf{s}}\bkt[2]{\EE\bkt[1]{Y\mid
        Z=1, \vec{X}\betatrue=\vec{x}_{i}\betatrue} -
      \EE\bkt[1]{Y\mid Z=0, \vec{X}\betatrue=\vec{x}_{i}\betatrue}  - \eta}.
  \end{equation}
  To show this, first observe that \eqref{eq:54} and \eqref{eq:73}
  combine to give
  \begin{equation}
\label{eq:67}
    \tilde{\psi}_{\mathbf{s}}(\eta) = 
\frac{\tilde{w}_{\mathbf{s}}}%
{\pibs\bkt[1]{Z_{\mathbf{s}}}}
\bkt[1]{\bar{Y}_{\mathbf{s}(1)} - \bar{Y}_{\mathbf{s}(0)}-
      \eta}.
  \end{equation}
As $\mathcal{S}_{n}$ is assumed to be a fine stratification, one or
both of $\bar{Y}_{\mathbf{s}(1)} = \operatorname{avg}\bkt[2]{(Y_{i}: i\in \mathbf{s}, Z_{i}=1)}$ and
  $\bar{Y}_{\mathbf{s}(0)} = \operatorname{avg}\bkt[2]{(Y_{i}: i\in
    \mathbf{s}, Z_{i}=0)}$ is in actuality a single observation.
  \eqref{eq:67} is taken to be zero, as is $\psi_{\mathbf{s}}(\eta) =
\tilde{w}_{\mathbf{s}}\sum_{i \in \mathbf{s}}\psi_{\mathbf{s}i}(\eta)
=\tilde{w}_{\mathbf{s}}\bkt[1]{\bar{Y}_{\mathbf{s}(1)} - \bar{Y}_{\mathbf{s}(0)}-
      \eta}$, when $\tilde{w}_{\mathbf{s}}=0$ because $\bar{Z}_{\mathbf{s}}=0$ or 1.
 
  For $\mathbf{s}$ with $\sum_{i\in \mathbf{s}}Z_{i}=1$,
  by \eqref{eq:45} the expectation of \eqref{eq:67}  evaluates to
  $\EE\bkt[2]{\tilde{\psi}_{\mathbf{s}}(\eta) \mid \mathcal{G}_{n}} =$
  \begin{multline}
    \sum_{i \in \mathbf{s}}
    \bkt[3]{\EE(Y \mid Z=1,
    \vec{X}\betatrue=\vec{x}_{i}\betatrue) -
    \underset{j \in \mathbf{s}\setminus\{i\}}{\operatorname{avg}}\bkt[2]{
    \EE(Y \mid Z=0,
    \vec{X}\betatrue=\vec{x}_{j}\betatrue)} - \eta} \\
\cdot \frac{\tilde{w}_{\mathbf{s}}}%
    {\pibs({0}^{(\mathbf{s})}_{+i})}
\pibs({0}^{(\mathbf{s})}_{+i})\\
= \tilde{w}_{\mathbf{s}}\sum_{i \in \mathbf{s}}\bkt[3]{\EE(Y \mid Z=1,
    \vec{X}\betatrue=\vec{x}_{i}\betatrue) -
    \underset{j \in \mathbf{s}\setminus\{i\}}{\operatorname{avg}}\bkt[2]{
    \EE(Y \mid Z=0,
    \vec{X}\betatrue=\vec{x}_{j}\betatrue)} - \eta},
    \label{eq:48}
\end{multline}
where we use ${0}^{(\mathbf{s})}_{+i}$ to denote the mapping on ${\mathbf{s}}$
taking $i$ to $1$ and remaining elements to $0$; \eqref{eq:74}
follows.  For $\mathbf{s}$ with
$\sum_{i\in \mathbf{s}}\indicator{z_{i}=0} =1$, this argument
gives $\EE\bkt[2]{\tilde{\psi}_{\mathbf{s}}(\eta) \mid \mathcal{G}_{n}} =$
\begin{equation}
\tilde{w}_{\mathbf{s}} \sum_{i \in \mathbf{s}}
\bkt[3]{
  \underset{j \in \mathbf{s}\setminus\{i\}}{\operatorname{avg}}\bkt[2]{\EE(Y \mid Z=1,
    \vec{X}\betatrue=\vec{x}_{j}\betatrue)} -
    \EE(Y \mid Z=0,
    \vec{X}\betatrue=\vec{x}_{i}\betatrue) - \eta} \label{eq:75}
\end{equation}
where ${1}^{(\mathbf{s})}_{-i}$ denotes the mapping of ${\mathbf{s}}$
that takes $i$ to $0$ and remaining elements to $1$.  Again
\eqref{eq:74} follows.

\subparagraph{Part (ii).}
Given $\mathbf{s}\in \mathcal{S}_{n}$, write $\mu^{(\mathbf{s})}_{0}=
\EE(\bar{Y}_{\mathbf{s}(0)} \mid \mathcal{F}_{n})$ and $\mu^{(\mathbf{s})}_{1}=
\EE(\bar{Y}_{\mathbf{s}(1)} \mid \mathcal{F}_{n})$.  By symmetry,
$\mu^{(\mathbf{s})}_{z} = \EE\bkt[1]{Y\mid Z_{1}=z, \sum_{i=1}^{\card[0]{\mathbf{s}}}Z_{i}=\sum_{i \in
    \mathbf{s}}z_{i}}$, $z=0$, 1. For $\mathbf{s} \in \mathcal{S}_{n}$ with $\sum_{i \in
  \mathbf{s}}\indicator{z_{i}=1} =1$, one has
$\tilde{w}_{\mathbf{s}}^{-1}\EE\bkt[2]{\psi_{\mathbf{s}}(\eta) \mid \mathcal{G}_{n}}=$
\begin{align*}
      &\sum_{i \in \mathbf{s}}\bkt[3]{\bkt[2]{\EE(Y \mid Z=1,
    \vec{X}\betatrue=\vec{x}_{i}\betatrue) -
    \underset{j \in \mathbf{s}\setminus\{i\}}{\operatorname{avg}}\bkt[2]{\EE(Y \mid Z=0,
    \vec{X}\betatrue=\vec{x}_{j}\betatrue)}} -
        \eta}{\pibs({0}^{(\mathbf{s})}_{+i})} \\
  =&\sum_{i \in \mathbf{s}}\bkt[3]{\bkt[2]{\EE(Y-\mu^{(\mathbf{s})}_{1} \mid Z=1,
    \vec{X}\betatrue=\vec{x}_{i}\betatrue) -
    \underset{j \in \mathbf{s}\setminus\{i\}}{\operatorname{avg}}\bkt[2]
     {\EE(Y - \mu^{(\mathbf{s})}_{0}\mid Z=0,
     \vec{X}\betatrue=\vec{x}_{j}\betatrue)}}}{\pibs({0}^{(\mathbf{s})}_{+i})}\\
  &+ \mu^{(\mathbf{s})}_{1} - \mu^{(\mathbf{s})}_{0} -\eta, 
\end{align*}
whereas if $\sum_{i \in
  \mathbf{s}}\indicator{z_{i}=0} =1$ then 
$\tilde{w}_{\mathbf{s}}^{-1}\EE\bkt[2]{\psi_{\mathbf{s}}(\eta) \mid \mathcal{G}_{n}}=$
\begin{align*}
      &\sum_{i \in \mathbf{s}}\bkt[3]{\bkt[2]{\underset{j \in \mathbf{s}\setminus\{i\}}{\operatorname{avg}}\bkt[2]{\EE(Y \mid Z=1,
    \vec{X}\betatrue=\vec{x}_{j}\betatrue)} -
    \EE(Y \mid Z=0,
    \vec{X}\betatrue=\vec{x}_{i}\betatrue)} -
        \eta}{\pibs({1}^{(\mathbf{s})}_{-i})}\\
  =& \sum_{i \in
     \mathbf{s}}\bkt[3]{{\underset{j \in \mathbf{s}\setminus\{i\}}{\operatorname{avg}}\bkt[2]{\EE(Y - \mu^{(\mathbf{s})}_{1}\mid Z=1,
    \vec{X}\betatrue=\vec{x}_{j}\betatrue)} -
    \EE(Y - \mu^{(\mathbf{s})}_{0}\mid Z=0,
    \vec{X}\betatrue=\vec{x}_{i}\betatrue)}}{\pibs({1}^{(\mathbf{s})}_{-i})}\\
  &+ \mu^{(\mathbf{s})}_{1} - \mu^{(\mathbf{s})}_{0} -\eta. 
\end{align*}
At the same time, \eqref{eq:48} and \eqref{eq:75} give
\begin{multline*}
  \tilde{w}_{\mathbf{s}}^{-1}\EE\bkt[2]{\psi_{\mathbf{s}}(\eta) \mid
    \mathcal{G}_{n}}=\mu^{(\mathbf{s})}_{1} - \mu^{(\mathbf{s})}_{0}
  -\eta +\\
  \sum_{i \in \mathbf{s}}\bkt[3]{\bkt[2]{\EE(Y-\mu^{(\mathbf{s})}_{1} \mid Z=1,
    \vec{X}\betatrue=\vec{x}_{i}\betatrue) -
    \underset{j \in \mathbf{s}\setminus\{i\}}{\operatorname{avg}}\bkt[2]
     {\EE(Y - \mu^{(\mathbf{s})}_{0}\mid Z=0,
     \vec{X}\betatrue=\vec{x}_{j}\betatrue)}}}\card[0]{\mathbf{s}}^{-1}
\end{multline*}
or
\begin{multline*}
  \tilde{w}_{\mathbf{s}}^{-1}\EE\bkt[2]{\psi_{\mathbf{s}}(\eta) \mid
    \mathcal{G}_{n}}=\mu^{(\mathbf{s})}_{1} - \mu^{(\mathbf{s})}_{0}
  -\eta +\\
\sum_{i \in
     \mathbf{s}}{\bkt[3]{\underset{j \in
       \mathbf{s}\setminus\{i\}}{\operatorname{avg}}\bkt[2]%
     {\EE(Y - \mu^{(\mathbf{s})}_{1}\mid Z=1,
    \vec{X}\betatrue=\vec{x}_{j}\betatrue)} -
    \EE(Y - \mu^{(\mathbf{s})}_{0}\mid Z=0,
    \vec{X}\betatrue=\vec{x}_{i}\betatrue)}}\card[0]{\mathbf{s}}^{-1},
\end{multline*}
depending as $\sum_{i \in \mathbf{s}}\indicator{z_{i}=1} =1$ or
$\sum_{i \in \mathbf{s}}\indicator{z_{i}=0} =1$, respectively. 
Differencing these expressions,
\begin{multline}
\tilde{w}_{\mathbf{s}}^{-1}\EE\bkt[2]{\tilde{\psi}_{\mathbf{s}}(\eta)
  - \psi_{\mathbf{s}}(\eta) \mid
    \mathcal{G}_{n}}= \sum_{i \in
    \mathbf{s}}\bbkt[1]{\frac{\card[0]{\mathbf{s}}^{-1}}{\pibs(0^{\mathbf{s}}_{+i})}-1}
  \\
  \cdot \bkt[3]{%
    \bkt[2]{\EE(Y-\mu^{(\mathbf{s})}_{1} \mid Z=1,
    \vec{X}\betatrue=\vec{x}_{i}\betatrue) -
    \underset{j \in \mathbf{s}\setminus\{i\}}{\operatorname{avg}}\bkt[2]
     {\EE(Y - \mu^{(\mathbf{s})}_{0}\mid Z=0,
     \vec{X}\betatrue=\vec{x}_{j}\betatrue)}}}
  \pibs(0^{\mathbf{s}}_{+i})
  \label{eq:68}
\end{multline}
if  $\sum_{i \in \mathbf{s}}\indicator{z_{i}=1} =1$, and if  $\sum_{i
  \in \mathbf{s}}\indicator{z_{i}=0} =1$ then
\begin{multline}
\tilde{w}_{\mathbf{s}}^{-1}\EE\bkt[2]{\tilde{\psi}_{\mathbf{s}}(\eta)
  - \psi_{\mathbf{s}}(\eta) \mid
    \mathcal{G}_{n}}= \sum_{i \in
    \mathbf{s}}\bbkt[1]{\frac{\card[0]{\mathbf{s}}^{-1}}{\pibs(1^{\mathbf{s}}_{-i})}-1}
  \\
  \cdot \bkt[3]{{\underset{j \in \mathbf{s}\setminus\{i\}}{\operatorname{avg}}\bkt[2]{\EE(Y - \mu^{(\mathbf{s})}_{1}\mid Z=1,
    \vec{X}\betatrue=\vec{x}_{j}\betatrue)} -
    \EE(Y - \mu^{(\mathbf{s})}_{0}\mid Z=0,
    \vec{X}\betatrue=\vec{x}_{i}\betatrue)}}{\pibs({1}^{(\mathbf{s})}_{-i})}.
  \label{eq:69}
\end{multline}
Observing that $\card[0]{\mathbf{s}}\EE(|\Vns| \mid \mathcal{G}_{n})\geq |\EE(\card[0]{\mathbf{s}}\Vns \mid \mathcal{G}_{n})| =$
\begin{equation*}
  \left|\sum_{i \in
    \mathbf{s}}
  \bkt[3]{%
    \bkt[2]{\EE(Y-\mu^{(\mathbf{s})}_{1} \mid Z=1,
    \vec{X}\betatrue=\vec{x}_{i}\betatrue) -
    \underset{j \in \mathbf{s}\setminus\{i\}}{\operatorname{avg}}\bkt[2]
     {\EE(Y - \mu^{(\mathbf{s})}_{0}\mid Z=0,
     \vec{X}\betatrue=\vec{x}_{j}\betatrue)}}}
  \pibs(0^{\mathbf{s}}_{+i})\right|
\end{equation*}
or
\begin{equation*}
  \left|\sum_{i \in
    \mathbf{s}}
\bkt[3]{{\underset{j \in \mathbf{s}\setminus\{i\}}{\operatorname{avg}}\bkt[2]{\EE(Y - \mu^{(\mathbf{s})}_{1}\mid Z=1,
    \vec{X}\betatrue=\vec{x}_{j}\betatrue)} -
    \EE(Y - \mu^{(\mathbf{s})}_{0}\mid Z=0,
    \vec{X}\betatrue=\vec{x}_{i}\betatrue)}}{\pibs({1}^{(\mathbf{s})}_{-i})}\right|
\end{equation*}
depending as 
$\sum_{i \in \mathbf{s}}\indicator{z_{i}=1} =1$ or
$\sum_{i \in \mathbf{s}}\indicator{z_{i}=0} =1$, and that
under the same respective conditions 
\begin{equation*}
  \left|\frac{\card[0]{\mathbf{s}}^{-1}}{\pibs(0^{\mathbf{s}}_{+i})}-1\right|
    \leq \exp\bkt[1]{2\sup_{i,j \in \mathbf{s}}|\theta_{i} - \theta_{j}|} -1
    \text{ or }
    \left|\frac{\card[0]{\mathbf{s}}^{-1}}{\pibs(1^{\mathbf{s}}_{-i})}-1\right|
    \leq \exp\bkt[1]{2\sup_{i,j \in \mathbf{s}}|\theta_{i} - \theta_{j}|} -1
\end{equation*}
by Lemma~\ref{lem:msPSerr}, we have:
\begin{align}
  \label{eq:70}
  |\EE\bkt[2]{\tilde{\psi}_{\mathbf{s}}(\eta)
  - \psi_{\mathbf{s}}(\eta) \mid
    \mathcal{G}_{n}}| \leq &
                             \tilde{w}_{\mathbf{s}}\card[0]{\mathbf{s}}\bkt[2]{\exp\bkt[1]{2\sup_{i,j \in \mathbf{s}}|\theta_{i} - \theta_{j}|} -1} \EE(|\Vns| \mid \mathcal{G}_{n}).
\end{align}
This establishes~\eqref{eq:62}.

\subparagraph{Part (iii).} Because the conditional expectation operator is a
contraction in $L_{p}$, 
\begin{align*}
  \EE\bbkt[1]{\frac{\sum_{\mathbf{s} \in
          \mathcal{S}_{n}}\tilde{w}_{\mathbf{s}}\card[0]{\mathbf{s}}\EE\bkt[1]{|\Vns|
  \mid \mathcal{G}_{n}}}{\sum_{\mathbf{s}\in
          \mathcal{S}_{n}}\tilde{w}_{\mathbf{s}}\card[0]{\mathbf{s}}}} \leq&
    \frac{\sum_{\mathbf{s} \in
          \mathcal{S}_{n}}\tilde{w}_{\mathbf{s}}\card[0]{\mathbf{s}}\EE|\Vns|}{\sum_{\mathbf{s}\in
                                                         \mathcal{S}_{n}}\tilde{w}_{\mathbf{s}}\card[0]{\mathbf{s}}}\\
  \leq& \frac{\sum_{\mathbf{s} \in
          \mathcal{S}_{n}}\tilde{w}_{\mathbf{s}}\card[0]{\mathbf{s}}\EE|V|}{\sum_{\mathbf{s}\in
  \mathcal{S}_{n}}\tilde{w}_{\mathbf{s}}\card[0]{\mathbf{s}}} = \EE|V|.
        \label{eq:51}
\end{align*}
Markov's inequality now gives that ${\sum_{\mathbf{s}\in
  \mathcal{S}_{n}}\tilde{w}_{\mathbf{s}}\card[0]{\mathbf{s}}}^{-1}{\sum_{\mathbf{s} \in
  \mathcal{S}_{n}}\tilde{w}_{\mathbf{s}}\card[0]{\mathbf{s}}
\EE\bkt[1]{|\Vns|\mid \mathcal{G}_{n}}} =
      O_{P}(1)$.   Accordingly, $| \{\theta_{i} - \theta_{j}: i \sim
        j\} |_{\infty} = o_{P}(1)$ combines with~\eqref{eq:62} to entail
\begin{equation}
          \label{eq:71}
          \sup_{\eta}|\EE \bkt[2]{\tilde{\psi}_{\mathcal{S}_{n}}(\eta) - \psi_{\mathcal{S}_{n}}(\eta)\mid \mathcal{G}_{n}}| = o_{P}(1). 
\end{equation}
By Prop.~\ref{prop:Qconsist}, $\eta \mapsto \EE
\bkt[1]{\psi_{\mathcal{S}_{n}}(\eta)\mid \mathcal{G}_{n}}$ has the
unique root $\tau_{n}$, and by part (i) of this proposition,  $\eta \mapsto \EE
\bkt[1]{\tilde{\psi}_{\mathcal{S}_{n}}(\eta)\mid \mathcal{G}_{n}}$ has
a unique root given by \eqref{eq:72}; as either of these functions'
slopes are bounded away from zero, \eqref{eq:71} entails that these
roots must converge together. 
\subparagraph{Part (iv).}  Part (iv) of the proposition now follows
from conclusion \ref{prop:Qconsist:item3} of Proposition~\ref{prop:Qconsist}. 
\end{proof}
\end{document}